\pdfoutput=1
\documentclass{llncs}
\usepackage{etex}
\usepackage{lmodern}
\usepackage[T1]{fontenc}
\usepackage[utf8]{inputenc} 
\usepackage[pdftex,colorlinks=true,bookmarks=false,linkcolor=blue,citecolor=blue,urlcolor=blue] {hyperref}
\usepackage[normalem]{ulem}
\usepackage[all]{xy}
\usepackage{hhline}
\usepackage{subfigure}
\usepackage{rotating}
\usepackage[table]{xcolor}
\usepackage{colorbar}
\usepackage{amsmath,amssymb,epsfig,graphicx,verbatim,bm,galois,xcolor,xspace,todonotes}
\usepackage{microtype}

\newcommand{\mquoted}[1]{\text{`$#1\,$'}}

\renewcommand{\phi}{\varphi}

\newcommand{\sbr}[1]{\lbrack \! \lbrack #1 \rbrack \! \rbrack}

\newcommand{\FeatExp}{\textit{FeatExp}}
\newcommand{\Exp}{\textit{Exp}}
\newcommand{\Stm}{\textit{Stm}}
\newcommand{\Const}{\textit{Const}}
\newcommand{\Var}{\textit{Var}}
\newcommand{\sat}{\textrm{sat}}

\newcommand{\confprojname}[1]{\ensuremath{\pi_{#1}}}
\newcommand{\confproj}[2]{\ensuremath{\confprojname{#1}(#2)}}
\newcommand{\poset}[2]{\ensuremath{\langle{#1},{#2}\rangle}}
\DeclareMathOperator{\lfp}{\mathrm{lfp}}

\newcommand{\true}{\textrm{true}}
\newcommand{\false}{\textrm{false}}
\newcommand{\entails}{\models}
\newcommand{\abar}{\ensuremath{\overline a}}
\newcommand{\dbar}{\ensuremath{\overline d}}
\newcommand{\VarIMP}{\ensuremath{\overline{\text{IMP}}}\xspace}
\newcommand{\impskip}[1][]{\ensuremath{\mbox{\texttt{skip}}^{#1}}}
\newcommand{\impassign}[3][]{\ensuremath{{#2} ~\mbox{\texttt{:=}}^{#1}~ {#3}}}
\newcommand{\impseq}[3][]{\ensuremath{{#2} ~\mbox{\texttt{;}}^{#1}~ {#3}}}
\newcommand{\impif}[4][]{\mbox{\texttt{if}}^{#1} ~\ensuremath{{#2}~ \mbox{\texttt{then}} ~{#3}~ \mbox{\texttt{else}} ~{#4}}}
\newcommand{\impwhile}[3][]{\mbox{\texttt{while}}^{#1} ~\ensuremath{{#2}~ \mbox{\texttt{do}} ~{#3}}}
\newcommand{\impifdef}[3][]{\mbox{\texttt{\#if}}^{#1} ~\ensuremath{{#2}} ~{#3}}
\newcommand{\func}[2]{\ensuremath{\lambda{#1}.\,{#2}}}
\newcommand{\varsignstore}{\ensuremath{\overline{a}}}

\newcommand{\varsignfuncname}{\ensuremath{\overline{\Phi}}}
\newcommand{\varsignfunc}[1]{\varsignfuncname({#1})}
\newcommand{\impparbinopname}{\widehat{\oplus}}
\newcommand{\impparbinop}[2]{\ensuremath{ {#1} ~\impparbinopname~ {#2}}}
\newcommand{\impbinop}[3][]{\ensuremath{ {#2} \oplus^{#1} {#3} }}

\newcommand{\pwbot}{\ensuremath{\dot{\bot}}}
\newcommand{\pwsqcup}{\dot{\sqcup}}

\newcommand{\pwsqsubseteq}{\dot\sqsubseteq}

\newcommand{\varstmtin}[1]{ \sbr{#1}_{\overline{\mathsf{in}}} }
\newcommand{\varstmtout}[1]{ \sbr{#1}_{\overline{\mathsf{out}}} }

\newcommand{\hoalpha}{\ensuremath{\alpha_{\rightarrow}}}
\newcommand{\hogamma}{\ensuremath{\gamma_{\rightarrow}}}
\newcommand{\Kpsi}{{\ensuremath{\mathbb{K}_{\psi}}}}
\newcommand{\Kk}{\ensuremath{\mathbb{K}}}
\newcommand{\Ff}{\ensuremath{\mathbb{F}}}
\newcommand{\Aa}{\ensuremath{\mathbb{A}}}

\newcommand{\Atwo}{\ensuremath{\overline{\mathcal{A}}}}
\newcommand{\Athree}{\ensuremath{\overline{\mathcal{S}}}}
\newcommand{\Done}{\ensuremath{\overline{\mathcal{D}}}}

\newcommand{\joinasym}{\ensuremath{\bm{\alpha}^{\textnormal{\textrm{join}}}}}
\newcommand{\joingsym}{\ensuremath{\bm{\gamma}^{\textnormal{\textrm{join}}}}}
\newcommand{\joina}[1]{\ensuremath{\joinasym_{#1}}}
\newcommand{\joing}[1]{\ensuremath{\joingsym_{#1}}}
\newcommand{\projasym}{\ensuremath{\bm{\alpha}^{\textnormal{\textrm{proj}}}}}
\newcommand{\projgsym}{\ensuremath{\bm{\gamma}^{\textnormal{proj}}}}
\newcommand{\proja}[1]{\ensuremath{\projasym_{#1}}}
\newcommand{\projg}[1]{\ensuremath{\projgsym_{#1}}}
\newcommand{\fprojasym}{\ensuremath{\bm{\alpha}^{\textnormal{\textrm{fproj}}}}}
\newcommand{\fprojgsym}{\ensuremath{\bm{\gamma}^{\textnormal{\textrm{fproj}}}}}
\newcommand{\fproja}[1]{\ensuremath{\fprojasym_{#1}}}
\newcommand{\fprojg}[1]{\ensuremath{\fprojgsym_{#1}}}
\newcommand{\fignoreasym}{\ensuremath{\bm{\alpha}^{\textnormal{\textrm{fignore}}}}}
\newcommand{\fignoregsym}{\ensuremath{\bm{\gamma}^{\textnormal{\textrm{fignore}}}}}
\newcommand{\fignorea}[1]{\ensuremath{\fignoreasym_{#1}}}
\newcommand{\fignoreg}[1]{\ensuremath{\fignoregsym_{#1}}}

\definecolor{darkblue}{rgb}{0,0,0.5}
\definecolor{lightgray}{rgb}{0.8,0.8,0.8}

\title{Variability Abstractions: Trading Precision for Speed in Family-Based Analyses}
\subtitle{Extended Version
\thanks{Partially supported by The Danish Council for Independent
			Research under a Sapere Aude project, VARIETE.}
}
\author{Aleksandar S. Dimovski \and Claus Brabrand \and Andrzej W{\k a}sowski
}
\institute{IT University of Copenhagen, Denmark}

\begin{document}
\maketitle

\begin{abstract}
  Family-based (lifted) data-flow analysis for Software Product Lines
  (SPLs) is capable of analyzing all valid products (variants) without
  generating any of them explicitly.  It takes as input only the
  common code base, which encodes all variants of a SPL, and produces
  analysis results corresponding to all variants.  However, the
  computational cost of the lifted analysis still depends inherently
  on the number of variants (which is exponential in the number of
  features, in the worst case).  For a large number of features, the
  lifted analysis may be too costly or even infeasible.  In this
  paper, we introduce variability abstractions defined as Galois
  connections and use abstract interpretation as a formal method for
  the calculational-based derivation of approximate (abstracted)
  lifted analyses of SPL programs, which are sound by construction.
  Moreover, given an abstraction we define a syntactic transformation
  that translates any SPL program into an abstracted version of it,
  such that the analysis of the abstracted SPL coincides with the
  corresponding abstracted analysis of the original SPL.  We implement
  the transformation in a tool, \texttt{reconfigurator} that works on
  Object-Oriented Java program families, and evaluate the practicality
  of this approach on three Java SPL benchmarks.
\end{abstract}

\section{Introduction and Motivation}\label{sec:introduction}
Software Product Lines (SPLs) are an effective strategy for developing
and maintaining a family of related programs.  Any valid program
(\emph{variant}) of an SPL is specified in terms of features selected.  A
\emph{feature} is a distinctive aspect, quality, or characteristic from the
problem-domain of a system.  SPLs have been adopted by the industry because of
improvements in productivity and
time-to-market\,\cite{SPL}. While there are many implementation strategies, many industrial product lines are implemented using annotative approaches such as
conditional compilation; in particular, via the C-preprocessor $\texttt{\#ifdef}$
construct\,\cite{ifdef}.

Recently, formal analysis and verification of SPLs have been a topic
of considerable research (see \cite{TRfin2012} for a survey).  The
challenge is to develop analysis and verification techniques that work
at the level of program families, rather than the level of individual
programs.  Given that the number of variants grows exponentially with
the number of features, the need for efficient analysis and
verification techniques is essential.  To address this, a number of
so-called \emph{lifted} techniques have emerged, essentially lifting
existing analysis and verification techniques to work on program
families, rather than on individual programs.  This includes
lifted type checking\,\cite{TypeCheckingSPL}, lifted
data-flow analysis\,\cite{TAOSD,SPLLIFT}, lifted model
checking\,\cite{model-checking-spls}. 
They are also known as
family-based (\emph{variability-aware} or
feature-sensitive) techniques.
%
Lifted techniques are capable of analyzing the entire code base (all
variants at once), without having to explicitly generate and
analyze all individual variants, one at a time.  Also, lifted
techniques are capable of pin-pointing errors directly in the product line, as
opposed to reporting errors in an individual product derived
from the SPL. 

There are two ways to speed up  analyses: improving \emph{representation} and
increasing \emph{abstraction}.  The former has received considerable attention
in the field of family-based analysis.  In this paper, we investigate the
latter.  We consider a range of abstractions at the \emph{variability level}
that may tame the combinatorial explosion of configurations and reduce it to
something more tractable by manipulating the configuration space of a program.
Such variability abstractions enable deliberate trading of
precision for speed in family-based analyses, even turn infeasible analyses
into feasible ones, while retaining an intimate relationship back to the
original analysis (via the abstraction).

We organize our variability abstractions in a calculus that
provides convenient, modular, and compositional declarative
specification of abstractions.  We propose two basic
abstraction operators (\emph{project} and \emph{join}) and two
compositional abstraction operators (\emph{sequential composition} and
\emph{parallel composition}).
%
%
Each abstraction expresses a compromise between precision and speed in the
induced abstracted analysis.
We show how to apply each of these abstractions
to data-flow lifted analyses, to extract (derive) their corresponding efficient
and sound (correct) abstracted lifted analysis based on the
calculational approach of abstract interpretation developed in \cite{Cousot99}.
Note that the approach is applicable to \emph{any} analysis phrased as an
abstract interpretation; in particular, it is not limited to data-flow
analysis.


\begin{figure}[t]
$$
\xymatrix@C=1.5pc@R=1pc{
\text{SPL} \ar[rrr]^{\textit{derive analysis}} \ar[dd]_{\rotatebox{90}{\scriptsize\textit{reconfigure}}}^{\rotatebox{90}{\scriptsize\textit{abstract}}} &&& \ar[dd]_{\rotatebox{90}{\scriptsize\textit{abstract}}} \text{lifted analysis} \ar@{-->}@[gray][rrr]^{\textcolor{gray}{\textit{run analysis}}} &&& \textcolor{gray}{\text{\begin{tabular}{c}precise\\[-0.75ex]lifted analysis\\[-0.75ex]information\end{tabular}}}\\
&& \ar@{=>}[l] &&&& \ar@[gray]@{..>}[u] \textcolor{gray}{\scriptsize\text{\begin{tabular}{c}\underline{TRADEOFF}:\\[-0.5ex]\textit{precision} -vs-\ \textit{speed}\end{tabular}}} \ar@[gray]@{..>}[d]\\
\text{\begin{tabular}{c}\textbf{abstracted}\\[-0.75ex]SPL\end{tabular}} \ar[rrr]_{\textit{derive analysis}} &&& \text{\begin{tabular}{c}\textbf{abstracted}\\[-0.75ex]lifted analysis\end{tabular}} \ar@{-->}@[gray][rrr]_{\textcolor{gray}{\textit{run \textbf{faster} analysis}}} &&& \textcolor{gray}{\text{\begin{tabular}{c}\textbf{less precise}\\[-0.75ex]lifted analysis\\[-0.75ex]information\end{tabular}}}\\
}
$$
\caption{Diagram illustrating the role and intended usage of the
  \texttt{reconfigurator} transformation.  Instead of abstracting an already
  existing (or derived) lifted analysis, our transformation allows abstraction
  to be applied directly to the SPL.  The resulting ``abstracted SPL''
  can then be analyzed using existing techniques.  The two paths from
  SPL to ``abstracted lifted analysis'' are guaranteed to produce the
  same abstracted lifted analysis.}
\label{fig:diagram}
\end{figure}
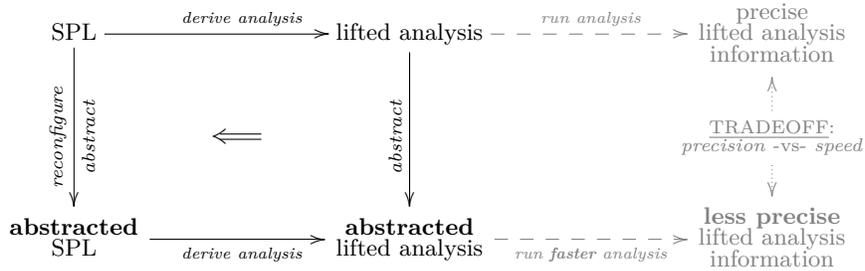

We observe that for variability abstractions, \emph{analysis
  abstraction} and \emph{analysis derivation} commute.
Figure\,\ref{fig:diagram} illustrates how analysis abstraction is
classically undertaken and how we propose to optimize it.  The top
left corner shows a product line that we want to analyze.  A lifted analyzer
will take an SPL as input and derive a ``lifted analysis'' (rightward
arrow).  We can then run that lifted analysis (next rightward dashed
arrow) and obtain our ``\emph{precise} lifted analysis information''.
(Note that for some analyzers, the phases \emph{derive analysis} and
subsequent \emph{run analysis} may be so intertwined that they are not
independently distinguishable.)  Since running the analysis might be
too slow or infeasible, we may decide to use abstraction to obtain a
faster, although less precise analysis.  Classically, an abstraction
is applied to the derived analysis before it is run (middle
arrow down) which, after an often long and complex process, produces an ``abstracted lifted analysis''.  When that
analysis is subsequently run, it will produce less precise
analysis information, but it will do so faster than the
original analysis (i.e., there is a \emph{precision vs.\ speed tradeoff}).

Interestingly, for lifted analyses and variability abstractions, the analysis abstraction (down) and derivation
(right) commute and  we may swap their order of application, as indicated by
the short double leftward arrow in the center.  The implications are quite
significant.  It means that variability abstractions can be applied
\emph{before}, and independently of, the subsequent analysis.  This also means
that the same variability abstractions might be applicable to all sorts of
analyses that are specifiable via abstract interpretation; including, but not
limited to: data-flow analysis\,\cite{CousotCousot79}, model
checking\,\cite{ModelCheckingAsAbsInt1}, type
systems\,\cite{TypeSystemsAsAbsInt} and testing\,\cite{TestingAsAbsInt}.

We exploit this observation to define a \emph{stand-alone} source-to-source
transformation for programs with \texttt{\#ifdef}s, implemented in a tool,
\texttt{reconfigurator}.  It takes an input SPL program and a
variability abstraction and produces an abstracted SPL program such for which
the subsequent lifted analysis agrees with ``abstracted lifted analysis'' of
the original unabstracted SPL.
Since the \texttt{reconfigurator} is based on a source-to-source transformation, and like a
preprocessor it is essentially unaware of the programming language syntax, it
can be used for any analysis.   Many existing analysis methods that are unable to
abstract variability benefit from this work instantly. Almost no extension or
adaptation is required as the abstraction is applied to source code before
analysis. 

We evaluate our approach by comparing analyses of a range of increasingly
abstracted SPLs against their origins without abstraction, quantifying to what
extent precision can be traded for speed in lifted analyses.


In summary, the paper makes the following contributions:%
\begin{itemize}%
\item[\textbf{C1:}] \emph{Variability abstraction} as a method for trading precision for
  speed in family-based analysis (based on abstract interpretation);
\item[\textbf{C2:}] A \emph{calculus} for modular specification of variability
  abstractions;
\item[\textbf{C3:}] The observation that certain analysis derivations
  and analysis abstractions \emph{commute}, meaning that variability
  abstractions can be applied directly on an SPL \emph{before} (and
  independently of) subsequent lifted analysis;
\item[\textbf{C4:}] A stand-alone \emph{transformation}, \texttt{reconfigurator}, based on the above ideas;
\item[\textbf{C5:}] An \emph{evaluation} of the above ideas;
  in particular, an evaluation of the tradeoff between precision and
  speed in family-based analyses.
\end{itemize}
\noindent
We direct this work to program analysis and software engineering
researchers.  The method of \emph{variability abstractions}
(\textbf{C1--C3}) is directed at designers of lifted
analyses for product lines.  They may use our insights to design
improved abstracted analyses that appropriately trade precision for
speed.  Note that the ideas apply beyond the context of data-flow
analyses (e.g., to model checking, type systems, verification, and
testing).  The \texttt{reconfigurator} (\textbf{C4}) and the
evaluation lessons (\textbf{C5}) are relevant for software engineers
working on preprocessor-based product lines and who would like to
speed up existing analyzers.

We proceed by  introducing the basics of lifting analyses in
Section\,\ref{sec:family_based_analysis}.
Section\,\ref{sec:variability_abstractions} defines a calculus for
specification of variability abstractions.  Section\,\ref{sec:analysis}
explains how to apply an abstraction to a lifted analysis. It uses constant
propagation as an example.  The \texttt{reconfigurator} is described in
Section\,\ref{sec:reconfigurator} along with correctness for our example
analysis.   Section \ref{sec:evaluation} presents the evaluation on three
Java Object-Oriented SPLs.
Finally, we discuss the relation to other works
and conclude.

\section{Program Families and Lifted Analyses} \label{sec:family_based_analysis}

In this section we summarize the prerequisites for presenting our
work.  We define \emph{features}, \emph{configurations}, \emph{feature
	expressions}, and a \emph{feature model} which designates a set of
\emph{valid} configurations.  Hereafter, we describe a simple
imperative language \VarIMP\ for writing program families.  Finally,
we briefly sketch a lifted constant propagation analysis for this
language, formally derived in\,\cite{MBW}.  We focus on constant
propagation for presentation purposes; our approach is generically
applicable to any lifted analysis phrased as an abstract
interpretation.

\paragraph{{Features, Configurations, and Feature Expressions.}}
Let $\Ff=\{A_1, \ldots, A_n\}$ be a finite set of \emph{features}, each
of which may be \emph{enabled} or \emph{disabled} in a particular
program variant.
A \emph{feature expression}, \textit{FeatExp} formula, is a propositional logic
formula over $\Ff$, defined inductively by:
\[
\phi ~~~ ::= ~~~ A {\color{gray}\;\in \mathbb F} ~~~ | ~~~ \neg \phi ~~~ | ~~~ \phi_1 \land \phi_2 ~~~ 
\]
%
A truth assignment or \emph{valuation} is a mapping $v$
assigning a truth value to all features. Every feature expression
evaluates to some truth value under the valuation $v$.
We say that $\phi$ is \emph{valid}, denoted as $\models
\phi$, if $\phi$ evaluates to $true$ for all valuations $v$.  We say that
$\phi$ is \emph{satisfiable}, denoted as $\sat(\phi)$, if there exists
a valuation $v$ such that $\phi$ evaluates to $\true$ under $v$.
We say that the formula $\theta$ is a semantic consequence of $\phi$,
denoted as $\phi \models \theta$, if for all satisfiable valuations
$v$ of $\phi$
it follows that $\theta$ evaluates to $true$ under $v$.  Otherwise, we
have $\phi \not\models \theta$.

\paragraph{{Feature Model.}}
A feature model describes the set of \emph{valid} configurations
(variants) of a product line in terms of features and relationships
among them.  For our purposes a feature model can be equated to a
propositional formula\,\cite{Bat}, say $\psi \in \FeatExp$, as the
semantic aspects of feature models beyond the configuration semantics,
are not relevant here.  We write $\Kpsi$ to denote the set of all
\emph{valid} configurations described by the feature model, $\psi$;
i.e., the set of all satisfiable valuations of $\psi$.  One
satisfiable valuation $v$ represents a valid configuration,
%
%
and it can be also encoded as a conjunction of literals: $k_v =
v(A_1) \cdot A_1 \land \cdots \land v(A_n) \cdot A_n$, where $true \cdot A = A$
and $false \cdot A = \neg A$, such that $k_v \models \psi$.  The truth value of
a feature in $v$ indicates whether the given feature is \emph{enabled}
(included) or \emph{disabled} (excluded) in the corresponding configuration.
Let $k_{v_1}, \ldots, k_{v_{n}}$ ($1 \leq n \leq 2^{|\mathbb{F}|}$) represent
all satisfiable valuations of $\psi$ expressed as formulas, then $\Kpsi = \{
k_{v_1}, \, \ldots, \, k_{v_{n}} \}$.  For example, the set of features,
$\mathbb F=\{A, B\}$, and the feature model, $\psi=A \lor B$, yield the
following set of \emph{valid} configurations: $\Kpsi=\{A \land B, A \land \neg
B, \neg A \land B\}$.

\paragraph{{The Programming Language.}}
\VarIMP\ is an extension of the imperative language IMP \cite{winskel}
often used in semantic studies. \VarIMP\ adds a compile-time
conditional statement for encoding multiple variants of a program.
The new statement ``$\impifdef{(\theta)}{s}$'' contains a feature
expression $\theta\in\FeatExp$ as a condition and a statement $s$ that
will be run, i.e.\ included in a variant, iff the condition $\theta$ is
satisfied by the corresponding configuration $k \in \Kk_{\psi}$. The
abstract syntax of the language is given by the following grammar:%
$$
\begin{array}{c}
s~~::=~~\impskip \, \mid \, \impassign{\texttt{x}}{e} \,\mid\, \impseq{s}{s} \,\mid\, \impif{e}{s}{s} \,\mid\, \impwhile{e}{s} \,\mid\, \impifdef{(\theta)}{s}  \\[0.5mm]
e~~::=~~n \,\mid\, \texttt{x} \,\mid\, e \oplus e
\end{array}
$$
%
%
where $n$ ranges over integers, $\texttt{x}$ ranges over variable names \Var, and
$\oplus$ over binary arithmetic operators.  The set of all generated statements
$s$ (respectively expressions $e$) is denoted by \Stm\ (respectively \Exp).
Notice that \VarIMP\ is only used for presentational purposes as a well established
minimal language. Still, the introduced methodology is
not limited to $\overline{\text{IMP}}$ or its features. In fact, we evaluate our
approach on Object-Oriented program families written in Java.

The semantics of $\overline{\text{IMP}}$ has two stages. First, a preprocessor
takes as input an $\overline{\text{IMP}}$ program and a configuration $k \in
\Kk_{\psi}$, and outputs a variant, i.e.\ an $\text{IMP}$ program without
\texttt{\#if}-s, corresponding to $k$.  All ``$\impifdef{(\theta)}{s}$'' statements are
appropriately resolved in the generated valid product, i.e.\ $s$ is included in
it iff $k \models \theta$.  Then, the obtained variant is executed (compiled)
using the standard $\text{IMP}$ semantics\,\cite{winskel}.

\paragraph{{Constant Propagation Analysis.}}
In the context of \VarIMP\, lifting means taking a static analysis that
works on IMP programs, and transforming it into an analysis that works
on \VarIMP\ programs, without preprocessing them (so on all the
variants simultaneously).  The lifted constant propagation analysis
for \VarIMP\ was derived in \cite{MBW}.  We first define a constant
propagation lattice $\poset{\textit{Const}}{\sqsubseteq_C}$, whose
partial ordering $\sqsubseteq_C$ is given by:
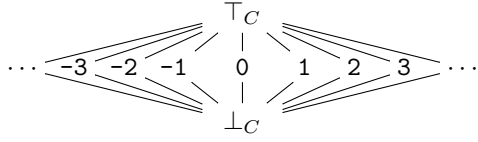
\[
\xymatrix@C=.2pc@R=.6pc{
& & & & \ar@{-}[dllll] \ar@{-}[dlll] \ar@{-}[dll] \ar@{-}[dl] {\phantom{p}}\top_C{\phantom{I}} \ar@{-}[d] \ar@{-}[dr] \ar@{-}[drr] \ar@{-}[drrr] \ar@{-}[drrrr] & & & \\
\cdots & \texttt{-3} & \texttt{-2} & \texttt{-1} & \texttt{0} & \texttt{1\phantom{-}} & \texttt{2\phantom{-}} & \texttt{3\phantom{-}} & \cdots \\
& & & & \ar@{-}[ullll] \ar@{-}[ulll] \ar@{-}[ull] \ar@{-}[ul] {\phantom{p}}\bot_C{\phantom{I}} \ar@{-}[u] \ar@{-}[ur] \ar@{-}[urr] \ar@{-}[urrr] \ar@{-}[urrrr] & & &
}
\]
In this domain $\top_C$ indicates a \emph{non-constant} value, and
$\bot_C$ indicates \emph{unanalyzed} information.  All other elements
indicate constant values.  The partial ordering $\sqsubseteq_C$
induces a least upper bound, $\sqcup_C$, and a greatest lower bound
operator, $\sqcap_C$, on the lattice elements.  For example, we have
$0 \sqcup_C 1 = \top_C$, $\top_C \sqcap_C 1 = 1$, etc.

The constant propagation analysis is given in terms of abstract
\emph{constant propagation stores}, denoted by $a$, essentially
mappings of variables to elements of \Const.  Thus $a(\texttt{x})$
informs whether the variable $\texttt{x}$ is a constant, and, in this case, what is its value.
We write $\Aa = \Var \to \Const$ meaning the
domain of all constant propagation stores.  Since \Const\ is a
complete lattice then so is $\langle \mathbb{A}, \sqsubseteq_{\Aa},
\sqcup_{\Aa}, \sqcap_{\Aa}, \bot_{\Aa}, \top_{\Aa} \rangle$ obtained
by point-wise lifting \cite{winskel}.  For example, for $a, a' \in
\Aa$ we have $a \sqsubseteq_{\Aa} a'$ iff $\forall \texttt{x} \in
\textit{Var}$, $a(\texttt{x}) \sqsubseteq_{C} a'(\texttt{x})$.
We omit the subscripts $C$ and $\Aa$ whenever they are clear in context.

\paragraph{{Lifted Constant Propagation Analysis.}}
For the lifted constant propagation analysis,
we work with the lifted property domain $\langle \mathbb{A}^{\Kpsi}, \dot\sqsubseteq, \dot \sqcup, \dot \sqcap, \dot \bot, \dot \top \rangle$,
where $\mathbb{A}^{\Kpsi}$ is shorthand for
the $|\Kpsi|$-fold product $\prod_{k \in \Kpsi} \mathbb{A}$, i.e.\ there is one separate copy of $\mathbb{A}$
for each valid configuration of  $\Kpsi$.
The ordering $\dot\sqsubseteq$ is lifted configuration-wise; i.e., for $\overline{a}, \overline{a}' \in \Aa^{\Kk_{\Psi}}$
we have $\overline{a} ~\dot\sqsubseteq~ \overline{a}' \equiv_\textit{def} ~\confproj{k}{\overline{a}} \sqsubseteq_{\Aa}
\confproj{k}{\overline{a}'}$ for all $k \in \Kpsi$. Here $\pi_{k}$ selects the ${k}^\text{th}$
component of a tuple.
Similarly, we lift configuration-wise all other elements of the complete lattice $\mathbb{A}$,
obtaining $\dot \sqcup, \dot \sqcap, \dot \bot, \dot \top$.
E.g., $\dot \top = \prod_{k \in \Kpsi} \top_{\Aa} = (\top_{\Aa}, \ldots, \top_{\Aa})$.

The lifted analysis $\overline{\mathcal{A}}\sbr{s}$
 should be a function from $\mathbb{A}^{\Kpsi}$ to $\mathbb{A}^{\Kpsi}$.
 However, using a tuple of $|\Kpsi|$ independent simple
 functions of type $\mathbb{A} \to \mathbb{A}$ is sufficient. 
Thus, the lifted analysis is given by the function
$\overline{\mathcal{A}}\sbr{s}:(\Aa \to \Aa)^{\Kpsi}$, which represents a tuple
of $|\Kpsi|$ functions of type $\mathbb{A} \to \mathbb{A}$.  The $k$-th
component of $\overline{\mathcal{A}}\sbr{s}$ defines the analysis corresponding
to the valid configuration described by the formula $k$.  Thus, an analysis
$\overline{\mathcal{A}}\sbr{s}$ transforms a lifted store, $\overline{a} \in
\mathbb{A}^{\Kk_{\psi}}$, into another lifted store of the same type.  For
simplicity, we overload the $\lambda$-abstraction notation, so creating a tuple
of functions looks like a function on tuples: we write $\lambda \overline{a} .
\prod_{k \in \Kk} f_{k}(\confproj{k}{\overline{a}})$ to mean $\prod_{k \in \Kk}
\lambda {a}_{k}.  f_{k}({a}_{k})$.  Similarly, if $\overline{f}:(\mathbb{A} \to
\mathbb{A})^{\Kk}$ and $\overline{a} \in \mathbb{A}^{\Kk}$, then we write
$\overline{f} (\overline{a})$ to mean $\prod_{k \in \Kk}
\confproj{k}{\overline{f}} (\confproj{k}{\overline{a}})$.

\begin{figure}
\begin{scriptsize}
\begin{align*}
 \overline{\mathcal{A}}\sbr{\impskip} &=
    \func{\overline{a}}{ \overline{a} }
\\
  \overline{\mathcal{A}}\sbr{\impassign{\texttt{x}}{e}} &=
       \func{\overline{a}}{\prod_{k \in \Kpsi}
           (\confproj{k}{\overline{a}}) [ \texttt{x} \mapsto \confproj{k}{\overline{\mathcal{A'}}\sbr{e} \overline{a}} ] }
\\
  \overline{\mathcal{A}}\sbr{\impseq{s_0}{s_1}} &=
    \overline{\mathcal{A}}\sbr{s_1} \circ \overline{\mathcal{A}}\sbr{s_0}
\\
  \overline{\mathcal{A}}\sbr{\impif{e}{s_0}{s_1}} &=
    \func{\overline{a}}{\overline{\mathcal{A}}\sbr{s_0} \overline{a} \,\dot\sqcup\, \overline{\mathcal{A}}\sbr{s_1} \overline{a}}
\\
  \overline{\mathcal{A}}\sbr{\impwhile{e}{s}} &=
    \mathrm{lfp}
      \func{\varsignfuncname}{\func{\varsignstore}{
           { \varsignstore ~\dot\sqcup~ \varsignfunc{\overline{\mathcal{A}}\sbr{s} \, \varsignstore} }}}
\\[1.5ex]
  \overline{\mathcal{A}}\sbr{\impifdef{(\theta)}{s}} &=
       \func{\varsignstore}{
         \prod_{k \in \Kpsi}
         \left\{ \begin{array}{ll} \confproj{k}{ \overline{\mathcal{A}}\sbr{s} \varsignstore} & \quad \textrm{if} \ k \models \theta \\[1.5ex]
		\confproj{k}{\varsignstore} & \quad \textrm{if} \  k \not \models \theta \end{array} \right.
         }
\\[1.5ex]
 \overline{\mathcal{A'}}\sbr{\mathit{n}} &=
        \func{\varsignstore}{
          \prod_{k \in \Kpsi}{
            \texttt{n} }}
\\
  \overline{\mathcal{A'}}\sbr{\texttt{x}} &=
        \func{\varsignstore}{
          \prod_{k \in \Kpsi}{
            \confproj{k}{\varsignstore} (\texttt{x}) }}
\\
  \overline{\mathcal{A'}}\sbr{\impbinop{e_0}{e_1}} &=
        \func{\varsignstore}{
          \prod_{k \in \Kpsi}{
            \impparbinop
              { \confproj{k}{ \overline{\mathcal{A'}}\sbr{e_0} \varsignstore}}
              { \confproj{k}{ \overline{\mathcal{A'}}\sbr{e_1} \varsignstore}} }}
\end{align*}
\end{scriptsize}
\caption{Definitions of
$\overline{\mathcal{A}}\sbr{s} : (\mathbb{A} \to \mathbb{A})^{\Kpsi}$
and
$\overline{\mathcal{A'}}\sbr{e} : (\mathbb{A} \to \textit{Const})^{\Kpsi}$.%
\label{fig:liftedanalysis}}
\vspace{-3mm}
\end{figure}

The equations for lifted analysis $\overline{\mathcal{A}}\sbr{s} : (\mathbb{A} \to
\mathbb{A})^{\Kpsi}$ and $\overline{\mathcal{A'}}\sbr{e} : (\mathbb{A} \to
\textit{Const})^{\Kpsi}$ that analyse all valid configurations simultaneously
are given in Fig~\ref{fig:liftedanalysis}.  They are systematically derived in
\cite{MBW} by following the steps of the calculational approach to abstract
interpretation \cite{Cousot99}: define collecting semantics, specify a series
of Galois connections and compose them with the collecting semantics to obtain
the resulting analysis, which is thus sound (correct) by construction.
Monotonicity of $\overline{\mathcal{A}}\sbr{s}$ and
$\overline{\mathcal{A'}}\sbr{e}$ was shown in \cite{MBW} as well.

The (transfer) function $\overline{\mathcal{A}}\sbr{s}$ captures the effect of
analysing the statement $s$ in an input store $\overline{a}$ by computing an
output store $\overline{a}'$.  For the $\impskip$ statement, the analysis
function is an identity on lifted stores.  For the assignment statement,
$\impassign{\texttt{x}}{e}$, the value of variable $\texttt{x}$ is updated in
every component of the input store $\overline{a}$ by the value of the
expression $e$ evaluated in the corresponding component of $\overline{a}$. The
$\texttt{if}$ case results in the least upper bound (join) of the effects from
the two corresponding branches, and it abstracts away the analysis information
at the guard (condition) point.  For the $\texttt{while}$ statement, we compute
the least fixed point of a functional\footnote{The functional of the
	$\texttt{while}$ rule is: $\func{\varsignfuncname}{\func{\varsignstore}{ {
				\varsignstore ~\dot\sqcup~
				\varsignfunc{\overline{\mathcal{A}}\sbr{s} \, \varsignstore} }}}$.}
in order to capture the effect of running all possible iterations of the
$\texttt{while}$ loop.  This fixed point exists and is computable by Kleene's
fixed point theorem, since the functional is a monotone function over complete
lattice with finite height \cite{MBW,CousotCousot79}.  For the
$\impifdef{(\theta)}{s}$ statement, we check for each valid configuration $k$
\footnote{Since any $k \in \Kk_{\psi}$ is a valuation, we have that $k \not \models \theta$ and $k \models \neg \theta$
are equivalent for any $\theta \in FeatExp$.}
whether the feature constraint $\theta$ is satisfied and, if so, it updates the corresponding component
of the input store by the effect of evaluating the statement $s$. Otherwise,
the corresponding component of the store is not updated. The function $\overline{\mathcal{A'}}\sbr{e}$
describes the result of evaluating the expression $e$ in a lifted store.
Note that, for each binary operator $\oplus$, we define the corresponding constant propagation operator $\impparbinopname$,
which operates on values from $\textit{Const}$, as follows:
\begin{align*}
  \impparbinop{v_0}{v_1} =
  \begin{cases}
        \bot &\text{if}~ v_0 = \bot \lor v_1 = \bot
    \\  \texttt{n} &\text{if}~ v_0 = \texttt{n}_0 \land v_1 =
    \texttt{n}_1, ~\text{where}~ \texttt{n} = \impbinop{\texttt{n}_0}{\texttt{n}_1}
    \\  \top &\text{otherwise}
  \end{cases}
\end{align*}
We lift the above operation configuration-wise, and in this way obtain a new operation
$\dot\impparbinopname$ on tuples of $\textit{Const}$ values.

\begin{example} \label{exp_lif2}
Consider the $\overline{\text{IMP}}$ program $\ensuremath{S_1}$:
\[
	\begin{array}{l}
		\texttt{x} := 0;  \\
		\impifdef{(A)}{\texttt{x} := \texttt{x} + 1;} \\
		\impifdef{(B)}{\texttt{x} := 1}
	\end{array}
\]
with the set $\Kpsi=\{ A \land B, A \land \neg B, \neg A \land B\}$.  By using
the rules of Fig.\,\ref{fig:liftedanalysis}, we can calculate
$\overline{\mathcal{A}}\sbr{\ensuremath{S_1}}$ for a store in which
$\texttt{x}$ is uninitialized, i.e.\ it has the value $\top$.  We assume a
convention here that the first component of the store corresponds to
configuration $A \land B$,  the second to $A \land \neg B$, and the third to
$\neg A \land B$.  We write $\overline{a_0}
\stackrel{\overline{\mathcal{A}}\sbr{s}}{\longmapsto} \overline{a_1}$ when
$\overline{\mathcal{A}}\sbr{s} \overline{a_0} = \overline{a_1}$.  We have:
   \begin{align}
     & \big([\texttt{x} \! \mapsto \! \top],\! [\texttt{x} \! \mapsto \! \top],\! [\texttt{x} \! \mapsto \! \top] \big)
        \ \stackrel{\overline{\mathcal{A}}\sbr{\texttt{x} := 0}}{\longmapsto} \
        \big([\texttt{x} \! \mapsto \! 0],\! [\texttt{x} \! \mapsto \! 0],\!  [\texttt{x} \! \mapsto \! 0] \big)
      \notag
\\ &  \quad \stackrel{\overline{\mathcal{A}}\sbr{\mathsf{\#if} \, (A) \, \texttt{x} := \texttt{x} + 1}}{\longmapsto}
        \big([\texttt{x} \! \mapsto \! 1],\! [\texttt{x} \! \mapsto \! 1],\!  [\texttt{x} \! \mapsto \! 0] \big)
       \stackrel{\overline{\mathcal{A}}\sbr{\mathsf{\#if} \, (B) \, \texttt{x} := 1}}{\longmapsto}
        \big([\texttt{x} \! \mapsto \! 1],\! [\texttt{x} \! \mapsto \! 1],\!  [\texttt{x} \! \mapsto \! 1] \big)
      \notag
   \end{align}
After evaluating $\ensuremath{S_1}$, the variable
$\texttt{x}$ has the constant value 1 for all valid configurations.
Observe that in the above lifted stores many components are the same, i.e.\
many configurations have equivalent analysis information. Such lifted stores
can be more compactly represented using sharing (e.g., bit vectors or formulae), which in effect
will result in more efficient implementation of the lifted analysis.

Let $\ensuremath{S_2}$ be a program obtained from $\ensuremath{S_1}$, such that
$\impifdef{(B)}{\texttt{x} := 1}$ is replaced with
$\impifdef{(B)}{\texttt{x} := \texttt{x}-1}$. Then, we have:
\[
\overline{\mathcal{A}}\sbr{\ensuremath{S_2}} \big([\texttt{x} \! \mapsto \! \top],\! [\texttt{x} \! \mapsto \! \top],\! [\texttt{x} \! \mapsto \! \top] \big)
= \big([\texttt{x} \! \mapsto \! 0],\! [\texttt{x} \! \mapsto \! 1],\! [\texttt{x} \! \mapsto \! -1] \big) \qquad \qquad
\]
We will use programs $S_1$ and $S_2$ as running examples throughout the paper.\qed
\end{example}

\section{Variability Abstractions}\label{sec:variability_abstractions}

When the set of configurations $\Kpsi$ is large, calculations on the property
domain $\Aa^\Kpsi$ become expensive, even if using symbolic representations  or
sharing to avoid direct storage of $|\Kpsi|$-sized tuples as done in
\cite{TAOSD}.   We want to replace $\mathbb{A}^{\Kpsi}$ with a smaller domain
obtained by abstraction and perform an approximate, but feasible, lifted
analysis.

\subsection{Basic Abstractions}\label{sec:basic-abstractions}

We describe a compositional way of constructing abstractions over the domain
\(\Aa^\Kk\!\), where $\Kk$ represents an arbitrary set of valid configurations,
  using two basic constructors, join and projection, along
with a sequential and parallel composition of abstractions.  The set of
abstractions $Abs$ is generated by the following grammar:
\begin{equation}
\alpha ~::=~ \joina{} \, \mid \, \proja{\phi}
			\, \mid \,  \alpha \circ \alpha
			\, \mid \, \alpha \, \otimes \, \alpha \enspace
 \end{equation}
where $\phi \in FeatExp$. Below we define the constructors
and motivate them with examples.  For readability, we use the constant
propagation lattice \Aa\, however the results hold for any complete lattice.

\paragraph{{Join.}} Consider the following scenario.  An analysis is run
interactively, while a developer is typing in a development environment.  The
analysis finds simple errors and warnings.  In this scenario, the analysis must
be fast and it should consider all legal configurations $\Kk$. It is not
problematic if some spurious errors are introduced, since, like previously, a more thorough
analysis is run regularly.  Here, the precision with respect to configurations
can be reduced by confounding the control-flow of all the products, obtaining
an analysis that runs as if it was analyzing a single product, but involving
code variants that participate in all products.

The \emph{join} abstraction gathers the information about all valid
configurations $k \in \Kk$ into one value of $\mathbb{A}$.  We formulate the
abstraction $\joina{}: \Aa^\Kk \to \mathbb{A}^{\{ \bigvee_{k \in
		\Kk } k \}}$ and the concretization function $\joing{}:
\Aa^{\{ \bigvee_{k \in \Kk } k \}} \to \mathbb{A}^{\Kk}$ as follows:%
\begin{equation}
	\joina{}(\overline{a}) =
	\left( \textstyle\bigsqcup_{k \in \Kk }
				\confproj{k}{\overline{a}} \right)
	\quad\text{ and }\quad
	\joing{}(a) =  \prod_{k \in \Kk} a
\end{equation}
We overload abstraction names ($\alpha$) to apply not only to domain elements
but also to sets of features, sets of configurations, and, later, to program
code.  
The new set of valid configurations is $\joina{}(\Kk)=\{ \bigvee_{k \in \Kk
} k \}$.  Thus, we have only one valid configuration denoted by the formula
$\bigvee_{k \in \Kk} k$.  Observe that this means that the obtained abstract
domain is effectively $\Aa^1$, which is isomorphic to \Aa.  The proposed
abstraction--concretization pair is a Galois connection, which means that it
can be used to construct analyses using calculational abstract interpretation:
\begin{theorem}\label{thm:join-galois}
	$\poset{\Aa^\Kk}{\pwsqsubseteq}
	\galois{\joina{}}{\joing{}}
	\poset{\mathbb{A}^{\joina{}(\Kk)} }{\pwsqsubseteq}$
	is a Galois connection \footnote{$\poset{L}{\leq_L} \galois{\alpha}{\gamma} \poset{M}{\leq_M}$
is a \emph{Galois connection} between complete lattices $L$ and $M$ iff $\alpha$ and $\gamma$ are total functions that
satisfy: $\alpha(l) \leq_M m \iff l \leq_L \gamma(m)$ for all $l \in L, m \in M$.
}
\footnote{The proofs of all theorems in this section can be found in App.~\ref{App:oper}. }.
\end{theorem}

\begin{example} \label{exp1}
	Let us return to the scenario of using \emph{join} for improving analysis
	performance. Assume that the feature model is given by $\psi = A
	\lor B$ with valid configurations $ \Kk_{\psi} = \{  A \land B, A \land \neg
	B, \neg A \land B\}$. Now, the final stores we obtain by analyzing programs
	$S_1$ and $S_2$ from Example~\ref{exp_lif2} are $\overline a_{S_1} =
	\big([\texttt{x} \! \mapsto \! 1],\! [\texttt{x} \! \mapsto \! 1],\!
	[\texttt{x} \! \mapsto \! 1] \big)$ and $\overline a_{S_2} =
	\big([\texttt{x} \! \mapsto \! 0],\! [\texttt{x} \!  \mapsto \! 1],\!
	[\texttt{x} \! \mapsto \! -1] \big)$. Applying the join abstraction we
	obtain $\joina{}(\overline a_{S_1}) = \big([\texttt{x} \!  \mapsto \! 1]
	\big)$ and $\joina{}(\overline a_{S_2}) = \big([\texttt{x} \!  \mapsto \!
	\top] \big)$. In both cases the state representation has been significantly
	decreased.  In the former case, the abstraction promptly notices that
	\texttt{x} is a constant regardless of the configuration. In the latter
	case, the abstraction looses precision by saying that \texttt{x} is not a
	constant in general, even if it was a constant in each of the configurations
	considered.  We will continue using stores $\overline a_{S_1}$ and
	$\overline a_{S_2}$ in the subsequent examples.
\qed
\end{example}

\paragraph{{Projection.}} In industrial practice  the number of
products actually deployed is often only a small subset of
$\Kk$\,\cite{DBLP:conf/models/BergerNRACW14}.
 In such case, analyzing all legal (valid) configurations seems unnecessary,
and performance of analyses can be improved by abstracting many products away.
This is achieved by a configuration projection, which removes configurations
that do not satisfy a given constraint, for instance a disjunction of product
configurations of interest.  Projection can be helpful in other similar
scenarios;  for instance, 
to parallelize the analysis---by partitioning the product
space using project and analyzing each partition separately.

Let $\phi$ be a formula over feature names. We define a \emph{projection}
	abstraction mapping $\mathbb{A}^{\Kk}$ into the domain $\mathbb{A}^{\{k \in
	\Kk | k \models \phi\}}$, which preserves only  the values corresponding to
configurations from $\Kk$ that satisfy $\phi$.  The information about
configurations violating $\phi$ is disregarded.  The abstraction and
concretization functions between $\mathbb{A}^{\Kk}$ and $\mathbb{A}^{\{k \in
	\Kk | k \models \phi\}}$ are defined as follows:
\begin{align}
& \proja\phi(\abar) =
	\textstyle\prod_{k \in \Kk, k \models \phi} \confproj k\abar
\label{eq:alpha-2}
\\
& \projg\phi(\abar') =  \textstyle\prod_{k \in \Kk}
	\begin{cases}
		\confproj{k}{\abar'} 	& \textrm{if } k \models \phi \\
		\top 							& \textrm{if } k \not\models \phi
	\end{cases}
\end{align}
%
%
%
The new set of configurations is
$\proja\phi(\Kk) = \{k \in \Kk \mid k \models \phi\}$. Naturally, we also have a Galois connection here:
\begin{theorem}\label{thm:proj-galois}
	$\poset{\Aa^\Kk}{\pwsqsubseteq}
	\galois{\proja\phi}{\projg\phi}
	\poset{\Aa^{\proja\phi(\Kk)}}{\pwsqsubseteq}$
	is a Galois connection.
\end{theorem}

\vspace{-2mm}

\noindent Notice that $\proja\true$ is the identity function, since $k \models
\true$ for all $k \in \Kk$.  On the other hand $\proja\false$ is the coarsest
collapsing abstraction that maps any tuple into an empty one, since $k
\not \models \false$, for all $k$.

\begin{example} \label{exp2}
	Let us revisit our scenario, where a set of deployed configurations is much
	smaller than the set of configurations defined by the feature model $\psi$.
	 Let us consider the store
	$\overline{a}_{S_2}$ with the set of valid configurations $\Kpsi$ from
	Example\,\ref{exp1}.  The set of deployed products is defined by formula
	$\phi = A$ (so all possible programs with feature $A$ are marketed).  By
	definition of projection (\ref{eq:alpha-2}), we have:
	$\proja{A}(\overline{a}_{S_2}) = \big( \confproj{A \land
		B}{\overline{a}_{S_2}}, \confproj{A \land  \neg B}{\overline{a}_{S_2}}
	\big) = \big( [\texttt{x} \! \mapsto \! 0], [\texttt{x} \! \mapsto \! 1]
	\big)$, and $\proja{\neg A}(\overline{a}_{S_2}) = ( \confproj{ \neg A \land
		B}{\overline{a}_{S_2}} ) = ([\texttt{x} \! \mapsto \! -1])$.  The state
	representation is effectively decreased to two, respectively
        one, components.
	\qed
\end{example}

An attentive reader, might
discount the idea of the projection abstraction as being overly heavy. In the
end, it appears to be equivalent to running the original analysis, just with a
strengthened feature model ($\psi\land\phi$).  However, as we shall see in the
subsequent developments, projection is indeed useful.  Thanks to the
composition operators it can enter intricate scenarios, which cannot be
expressed using a simple strengthening of a global feature model.

\paragraph{{Sequential Composition.}} We use composition to build complex
abstractions out of the basic ones, which also allows us to keep the number of
operators in the framework and in the implementation low.

Recall that a composition of two Galois connections is also a Galois connection
\cite{CousotCousot92}.  Let $\poset{\mathbb{A}^{\Kk}}{\dot\sqsubseteq}
\galois{\alpha_{1}}{\gamma_{1}} \poset{\mathbb{A}^{\alpha_1(\Kk)}
}{\dot\sqsubseteq}$ and $\poset{\mathbb{A}^{\alpha_1(\Kk)}}{\dot\sqsubseteq}
\galois{\alpha_{2}}{\gamma_{2}} \poset{\mathbb{A}^{\alpha_2(\alpha_1(\Kk))}
}{\dot\sqsubseteq}$ be two Galois connections. 
Then, we define their \emph{composition} as $ \poset{\Aa^\Kk}{\pwsqsubseteq}
\galois{\alpha_{2} \circ \alpha_{1}}{\gamma_{1} \circ \gamma_{2} } \poset{
	\Aa^{(\alpha_2 \circ \alpha_1)(\Kk)} }{\pwsqsubseteq} $, where
\begin{equation}
	(\alpha_2\circ\alpha_1)(\abar) = \alpha_{2}(\alpha_1(\abar))
	\quad\text{ and }\quad
	(\gamma_1\circ\gamma_2)(\abar') = \gamma_1(\gamma_2(\abar'))
\end{equation}
for $\abar\in\Aa^\Kk$ and $\abar'\in\Aa^{(\alpha_2 \circ \alpha_1)(\Kk)}$.
Also 
$(\alpha_2 \circ \alpha_1) (\Kk) = \alpha_2 (\alpha_1 (\Kk))$.

\begin{example}\label{exmpl:sequential}
%
%
	Now consider the process of deriving an analysis, which only considers
	products actually deployed described by a formula $\phi$ (see previous
	example), but which should trade precision for speed, by confounding their
	execution.  Such an analysis is derived using the composed abstraction:
	\(\joina{}\circ\proja{\phi}\).

	Let $\phi=A$. Configurations $A \land B$ and $A \land \neg B$ satisfy
	$\phi$, whereas $\neg \phi$ is satisfied only by $ \neg A \land B$.  We
	have: $\joina{}\circ\proja{A} (\abar_{S_2}) = ( \confproj{A \land B}
	{\overline{a}_{S_2}} \sqcup \confproj{A \land \neg B}{\overline{a}_{S_2}} )
	= ([\texttt{x} \! \mapsto \! \top])$ and \(\joina{}\circ\proja{\neg A}
		(\abar_{S_2}) = ( \confproj{\neg A \land B}{\overline{a}_{S_2}}) =
		([\texttt{x} \! \mapsto \! -1])\).
	\qed
%
%
\end{example}

\paragraph{{Parallel Composition.}} Consider a product line where two disjoint
groups of products share the same code base: one group is correctness critical, the
other comprises correctness non-critical products.  The former should be
analyzed with highest precision possible to obtain the most precise analysis
results, the latter can be analyzed faster.
We can set up such analyses by using a projection abstraction to analyze the
correctness critical group precisely, and the join abstraction to analyze
the non-critical group.  However running the analyses twice, ignores the fact
that the code is shared between the groups.  We can combine two separate
analyses by creating a compound abstraction: a \emph{product} of the two.  The
product abstraction will correspond exactly to executing the projection on the
correctness critical products, and join on the non-critical ones.  But since the
product creates a single Galois connection of the two, it can be used to derive
an analysis which will deliver this in a single run, which is more efficient
overall, due to reuse of the states explored.

Galois connections $\poset{\Aa^\Kk}\pwsqsubseteq$
$\smash{\galois{\alpha_1}{\gamma_1}} \poset{\Aa^{\alpha_1(\Kk)\!\!}
}\pwsqsubseteq$ and $\poset{\Aa^\Kk\!\!}\pwsqsubseteq$
\smash{$\galois{\alpha_2}{\gamma_2}$}
$\poset{\Aa^{\alpha_2(\Kk)}\!\!}\pwsqsubseteq$ over the same domain $\Aa^\Kk$
can be composed into one that combines the abstraction results "side-by-side".
The result is a new compound abstraction, $\alpha_1\otimes\alpha_2$, of the
domain $\mathbb{A}^{\Kk}$ obtained by applying the two simpler abstractions in
parallel. The parallel composition of abstractions is defined using a direct tensor
product.
For the resulting Galois connection, we have
$\alpha_1 \otimes \alpha_2(\Kk) = \alpha_1 (\Kk) \cup \alpha_2 (\Kk)$.  Given
$\abar_1\in\Aa^{\alpha_1(\Kk)}$ and $\abar_2\in \Aa^{\alpha_2(\Kk)}$, we first
define $\abar_1 \times \abar_2 \in \alpha_1(\Kk) \cup \alpha_2(\Kk)$ as:%
\begin{equation}
	\abar_1 \times \abar_2 = \prod_{k \in \alpha_1(\Kk) \cup \alpha_2(\Kk)}
		\begin{cases}
			\pi_k(\abar_1)
					& \text{if } k\in\alpha_1(\Kk) \setminus \alpha_2(\Kk) \\
	   	\pi_{k}(\abar_1) \sqcup \pi_{k}(\abar_2)
					& \text{if } k\in\alpha_1(\Kk) \cap \alpha_2(\Kk) \\
    		\pi_{k}(\abar_2)
					& \text{if } k\in\alpha_2(\Kk) \setminus \alpha_1(\Kk)
		\end{cases}
\end{equation}
The direct tensor product is given as $\poset{\Aa^\Kk}\pwsqsubseteq
\galois{\alpha_1 \otimes \alpha_2}{\gamma_1 \otimes \gamma_2 }
\poset{\Aa^{(\alpha_1 \otimes \alpha_2)(\Kk)}}\pwsqsubseteq $, where%
\begin{align} \label{eq:alpha-3}
	(\alpha_1\otimes\alpha_2)(\abar)
			& = \alpha_1(\abar) \times \alpha_2(\abar)\\
	(\gamma_1\otimes\gamma_2)(\abar')
			& = \gamma_1(\pi_{\alpha_1(\Kk)}(\abar')) \sqcap
				\gamma_2(\pi_{\alpha_2(\Kk)}(\abar'))
				\enspace, \text{ where}
\end{align}
$\pi_{\!\alpha_1(\Kk)}(\abar')\!=\! \prod_{k \in \alpha_1(\Kk)}
\!\pi_k(\abar')$ and $\pi_{\!\alpha_2(\Kk)}(\abar') \!=\! \prod_{k \in
	\alpha_2(\Kk)} \!\pi_{k}(\abar')$, for $\abar' \!\in\!
\Aa^{(\alpha_1\otimes\alpha_2)(\Kk)}$\!.

%
%
%
%
%
%

\begin{theorem}\label{thm:product-galois}
$\poset{\Aa^\Kk}\pwsqsubseteq
\galois{\alpha_1 \otimes \alpha_2}{\gamma_1 \otimes \gamma_2 }
\poset{\Aa^{(\alpha_1 \otimes \alpha_2)(\Kk)}}\pwsqsubseteq$ is a Galois connection.
\end{theorem}

\begin{example}
Let us assume that for products with feature $A$ we need precise analysis results, and
for products without this feature we do not need so precise results.  We are interested in analyzing products
with $A$ thoroughly, while the analysis of the products without $A$ can be
speeded up.  To this end we build the following abstraction: $\proja{A} \otimes
(\joina{} \circ \proja{\neg A})$.\qed
\end{example}




\subsection{Derived Abstractions}\label{sec:derived-abstractions}

We shall now discuss three more abstractions that can be derived from the above basic constructors.

\paragraph{{Join-Project.}} Recall the construction of
Example~\ref{exmpl:sequential}, where we combined projection with a join in
order to confound a subset of legal configurations. This pattern has occurred
so often in our exercises that we introduced a syntactic sugar for
it.  For a formula $\phi$ over features, the abstraction $\joina{\phi}$ gathers
the information about all valid configurations $k \in \Kk$ that satisfy $\phi$,
i.e.\ $k \models \phi$, into one value of $\mathbb{A}$, whereas the information
about all other valid configurations $k \in \Kk$ that do not satisfy $\phi$
 is disregarded.  We define
\begin{equation}
	\joina\phi = \joina{} \circ \proja\phi
	\quad \text{ and } \quad \joing{\phi} = \projg\phi \circ \joing{}\enspace
\end{equation}
where we have that \(\poset{\Aa^\Kk}\pwsqsubseteq
\galois{\proja\phi}{\projg\phi} \poset{\Aa^{\proja\phi(\Kk)}}\pwsqsubseteq\) and
$\poset{\Aa^{\proja\phi(\Kk)}}\pwsqsubseteq$
$\galois{\joina{}}{\joing{}}$
$\poset{\Aa^{(\joina{}\circ\proja\phi)(\Kk))}}\pwsqsubseteq$ are Galois
connections. Now the compositions in Example\,\ref{exmpl:sequential} can be
written simply as \joina{A} and \joina{\neg A}.

%
%

\paragraph{Ignoring features.} Consider a scenario, where a configurable
third-party component is integrated into a product line.  The code base is
large, and a static analysis does not scale to this size.  In a
compile-analyze-test cycle errors appear most often in the newly written code,
and are thus relatively little influenced by how the features of the third
party component are configured.  Lowering precision on analyzing external
components can allow finding errors faster.  
This scenario can be
realized using a feature projection, which simplifies feature domains by confounding
executions differing only on uninteresting features.

Before defining feature projection, let us consider a simpler case of ignoring
a single feature $A \in \Ff$ that is not directly relevant for current analysis. The
\emph{ignore feature} abstraction merges any configurations that only differ
with regard to $A$, and are identical with regard to remaining features,
$\Ff\!\setminus\!\{A\}$.
We write
\(\phi\setminus_A\) for a formula obtained by eliminating variable $A$ from
$\phi$.  The actual method of variable elimination is insignificant, as we
assume all equivalent formulas as identical in this paper. The new set of
configurations is given by $\fignorea{A}(\Kk) = \{ \bigvee_{k \in
		\Kk, k\setminus_A \equiv k' } k \mid k' \in \{ k\setminus_A \mid k\in\Kk \} \}$.
The abstraction $\fignorea{A}:\Aa^\Kk \to \Aa^{\fignorea{A}(\Kk)}$ and
concretization functions $\fignoreg{A}: \Aa^{\fignorea{A}(\Kk)} \to \Aa^\Kk$
are:
\begin{align}
	&\fignorea{A}(\abar) = \textstyle\prod_{k'\in\fignorea{A}(\Kk)}
	\textstyle\bigsqcup_{k\in\Kk, k \models k'}
	\pi_{k}(\abar)\label{eq:fignorea}\\
	&\fignoreg{A}(\abar') = \textstyle\prod_{k\in\Kk} \pi_{k'}(\abar') \ \text{ if } k \models k'
\end{align}

It turns out that ignoring features can be derived from the above basic abstractions as shown in the following theorem:
\begin{theorem}\label{thm:fproj-expansion}
Let $\fignorea{A}(\Kk) = \{k_1',\dots,k_n'\}$. Then:
\begin{equation*}
	\fignorea{A} = \joina{k_1'} \otimes \dots \otimes \joina{k_n'}
		 \quad\text{ and }\quad
	\fignoreg{A} = \joing{k_1'} \otimes \dots \otimes \joing{k_n'} \enspace.
\end{equation*}
\end{theorem}

 \begin{example}
 We consider the lifted store $\overline{a}_{S_2}$ with $\Kpsi=\{A \land B, A \land \neg B, \neg A \land B \}$.
 Then, we have $\fignorea{A}(\Kpsi)=\{(A \land B) \lor (\neg A \land B),  A \land \neg B\}$ and
 $\fignorea{A}(\overline{a}_{S_2}) = ( \confproj{A \land B}{\overline{a}_{S_2}} \sqcup \confproj{\neg A \land B}{\overline{a}_{S_2}},
  \confproj{A \land \neg B}{\overline{a}_{S_2}})
 =([\texttt{x} \! \mapsto \! \top], [\texttt{x} \! \mapsto \! 1])$.
 On the other hand, we have $\fignorea{B}(\Kpsi)=\{(A \land B) \lor (A \land \neg B), \neg A \land B\}$ and
 $\fignorea{A}(\overline{a}_{S_2}) = ( \confproj{A \land B}{\overline{a}_{S_2}} \sqcup \confproj{ A \land \neg B}{\overline{a}_{S_2}},
  \confproj{\neg A \land B}{\overline{a}_{S_2}})
 =([\texttt{x} \! \mapsto \! \top], [\texttt{x} \! \mapsto \! -1])$. \qed
 \end{example}

\paragraph{{Feature Projection.}} Now, if we need to ignore a larger number of
features (say features outside a certain component of interest), we can do it
using a feature projection operator which simply ignores a set of features 
$\{A_1,\dots,A_k\} \subseteq \Ff$:%
\begin{equation*}
\fproja{\{A_1,\dots,A_k\}} = \fignorea{A_1} \circ \dots \circ \fignorea{A_k}
\quad\text{ and }\quad
\fprojg{\{A_1,\dots,A_k\}} = \fignoreg{A_k} \circ \dots \circ \fignoreg{A_1}
\end{equation*}

\medskip

\noindent It follows from the theorems of Section\,\ref{sec:basic-abstractions}
that all the derived pairs of abstraction--concretization are Galois
connections.

%


%
%

%

\vspace{-2mm}

\section{Abstracting Lifted Analyses} \label{sec:analysis}

\begin{figure}[t]
	\small
   \begin{align*}
&
		(\alpha \circ \overline{\mathcal{A}}\sbr{\impifdef{(\theta)}{s}}
			\circ \gamma) (\overline{d})
		= \alpha (
				\overline{\mathcal{A}}\sbr{\impifdef{(\theta)}{s}}
					( \gamma (\overline{d}) ) ) =
		\tag{by def. of $\circ$}
\\[2mm] & =
		\alpha \bigg( \prod_{k \in \Kpsi}
         \begin{cases}
					\confproj{k}{\overline{\mathcal{A}}\sbr{s}\gamma(\dbar)}
							& \text{if } k \models \theta \\
					\confproj{k}{\gamma(\dbar)}
							& \text{if } k \not \models \theta
			\end{cases} ~\bigg)
      \tag{def.\ of $\overline{\mathcal{A}}$ in Fig.\,\ref{fig:liftedanalysis}}
\\ & \pwsqsubseteq
		\!\!\!\prod_{k'\in\alpha(\Kpsi)}\!\!
          \begin{cases}
				 \confproj{k'}{\alpha(\overline{\mathcal{A}}\sbr{s} \gamma(\dbar))}
				 		& \text{if } k' \models \theta \\[1mm]
         	\confproj{k'}{\alpha(\gamma (\overline{d}))} \sqcup \confproj{k'}{\alpha (\overline{\mathcal{A}}\sbr{s} \gamma (\overline{d}) )}
						& \text{if sat} (k' \!\land\! \theta) \land \text{sat} (k' \!\land\! \neg \theta) \\[1mm]
				\confproj{k'}{\alpha(\gamma (\overline{d}))}
						& \text{if } k' \models \neg \theta
			\end{cases}
      \tag{Lemma\,\ref{lemma:12}, App.\,\ref{AppHelp}}
\\ & \pwsqsubseteq
		\!\!\!\prod_{k' \in \alpha(\Kpsi)}\!\!
         \begin{cases}
				\confproj{k'}{ \overline{\mathcal{D}}_{\alpha}\sbr{s} \dbar}
					& \text{if } k' \models \theta \\[1mm]
         	\confproj{k'}{\dbar} \sqcup \confproj{k'}{\overline{\mathcal{D}}_{\alpha}\sbr{s} \dbar}
					& \text{if sat} (k' \!\land\! \theta) \land \textrm{sat} (k' \!\land\! \neg \theta) \\[1mm]
				\confproj{k'}{\overline{d}} & \text{if } k' \models \neg \theta
			\end{cases}
      \tag{by IH and $\alpha \circ \gamma$ reductive}
\\ & =
    \overline{\mathcal{D}}_{\alpha}\sbr{\impifdef{(\theta)}{s}}\,\overline{d}
        \notag
   \end{align*}
	\caption{Calculational derivation of
		$\overline{\mathcal{D}}_\alpha\sbr{\impifdef{\!(\theta)}s}$, the
		abstraction of
		$\overline{\mathcal{A}}\sbr{\impifdef{\!(\theta)}s}$.
		The `reductive' property of all Galois connections is $(\alpha \circ
		\gamma)(\dbar) \,
		\sqsubseteq \, \dbar$ for all \dbar.%
		\label{fig:deriving-if}}
\vspace{-3mm}
\end{figure}

We will now demonstrate how to derive abstracted lifted analyses using the
operators of Section\,\ref{sec:variability_abstractions}, using the case of
constant propagation for \VarIMP programs as an example.  Recall that this
analysis has been specified by: 1) the domain $\mathbb{A}^{\Kpsi}$; 2) the statement
transfer function $\overline{\mathcal{A}}\sbr{s}:(\mathbb{A} \to
\mathbb{A})^{\Kpsi}$; and 3) the expression evaluation function
$\overline{\mathcal{A'}}\sbr{e} : (\mathbb{A} \to \textit{Const})^{\Kpsi}$\!.
Let  $\poset{\mathbb{A}^{\Kpsi}}{\dot\sqsubseteq} \galois{\alpha}{\gamma}
\poset{\mathbb{A}^{\alpha(\Kpsi)} }{\dot\sqsubseteq}$ be a Galois connection
constructed using the abstractions presented in
Section\,\ref{sec:variability_abstractions}. 
We will also write $(\alpha,\gamma) \in Abs$
to denote a Galois connection obtained in such way.

Any function $f$ defined on the concrete domain of a Galois connection can be
abstracted to work on the abstract domain by applying concretization to its
argument and an abstraction to its value, i.e.\ by the function $F = \alpha
\circ f \circ \gamma$, where $\circ$ denotes the usual composition of
functions.  In fact, any monotone over-approximation of the composition $\alpha
\circ f \circ \gamma$ is sufficient for a sound analysis.  Even fixed   points
can be transferred from a concrete to an abstract domain of a Galois
connection. If both domains are complete lattices and $f$ is a monotone
function on the concrete domain, then by the fixed point transfer theorem (FPT
for short)~\cite{CousotCousot79}: $\alpha(\mathrm{lfp} f) ~\sqsubseteq~
\mathrm{lfp} F ~\sqsubseteq~ \mathrm{lfp} F^\#$.  Here $F = \alpha \circ f
\circ \gamma$ and $F^\#$ is some monotone, conservative
\emph{over}-approximation of $F$, i.e.\ $F \sqsubseteq F^\#$.  The
calculational approach to abstract interpretation~\cite{Cousot99} used in this
work, advocates simple algebraic manipulation to obtain a \emph{direct
	expression} for the function $F$ (if it exists) or for an over-approximation
$F^\#$.

In our case, for any lifted store $\abar \in \Aa^\Kpsi$, we calculate an
abstracted lifted store by $\alpha(\abar)=\overline{d} \in
\Aa^{\alpha(\Kpsi)}$.  Now, we use a Galois connection to derive an
over-approximation of $\alpha \circ \overline{\mathcal{A}}\sbr{s} \circ \gamma$
obtaining a new abstracted statement transfer function $\overline{\mathcal{D}}_{\alpha}\sbr{s}: (\Aa
\to \Aa)^{\alpha(\Kpsi)}$\!.  Similarly, one can derive an abstracted analysis
for expressions $\overline{\mathcal{D'}}_{\alpha}\sbr{e}$, approximating
$\alpha \circ \overline{\mathcal{A'}}\sbr{e} \circ \gamma$.
These approximations are derived using structural induction on statements
(respectively on expressions), in a process that resembles a simple algebraic
calculation, deceivingly akin to equation reasoning.

Let us consider the derivation steps for the static conditional
statement
($\impifdef{(\theta)}{s}$) in detail.  Our inductive hypothesis (IH) is that for
statements $\ensuremath{s'}$ that are structurally smaller than ($\impifdef{(\theta)}{s}$) the
(yet-to-be-calculated) $\overline{\mathcal{D}}_{\alpha}\sbr{s'}$ soundly
approximates $\alpha \circ \overline{\mathcal{A}}\sbr{s'} \circ \gamma$,
formally: $\alpha \circ \overline{\mathcal{A}}\sbr{s'} \circ \gamma
~\pwsqsubseteq~ \overline{\mathcal{D}}_{\alpha}\sbr{s'}$.  The derivation in
Fig.\,\ref{fig:deriving-if} begins with composing the concretization and
abstraction functions with the concrete transfer function and then proceeds by
expanding definitions. An (inner) induction on the structure of the abstraction
$\alpha$ follows, delegated to the Appendix for brevity. In the last step we
apply the inductive hypothesis, to obtain a closed representation independent
of $\overline{\mathcal{A}}$. This representation, just before the final
equality, is the newly obtained (calculated) definition of the abstracted
analysis $\overline{\mathcal{D}}_{\alpha}$.  Interestingly, the derivation is
independent of the structure of the abstraction $\alpha$, so this form works
for any abstraction specified using our operators.
We give a sketch of derivational steps for $\impbinop{e_0}{e_1}$ in Fig.~\ref{fig:deriving-op}.

\begin{figure}[t]
	\small
\begin{align*}
      & (\alpha \circ \overline{\mathcal{A'}}\sbr{\impbinop{e_0}{e_1}} \circ \gamma) (\overline{d})
\\ & =
        \alpha \big( \prod_{k \in \Kpsi}
                    \confproj{k}{ \overline{\mathcal{A'}}\sbr{e_0} \gamma (\overline{d})} \,\impparbinopname\,
                    \confproj{k}{ \overline{\mathcal{A'}}\sbr{e_1} \gamma (\overline{d})} \big)
        \tag{by def. of $\circ$, and $\overline{\mathcal{A'}}$ in
		 Fig.\,\ref{fig:liftedanalysis}}
\\ & =
        \prod_{k' \in \alpha(\Kpsi)}
                    \confproj{k'}{ \alpha \big( \overline{\mathcal{A'}}\sbr{e_0} \gamma (\overline{d}) \,\dot\impparbinopname\,
                     \overline{\mathcal{A'}}\sbr{e_1} \gamma (\overline{d}) \big) }
        \tag{by def. of $\pi_{k}$, $\dot\impparbinopname$, and $\alpha$}
\\ & \dot\sqsubseteq
        \prod_{k' \in \alpha(\Kpsi)}
                    \confproj{k'}{ \alpha ( \overline{\mathcal{A'}}\sbr{e_0} \gamma (\overline{d}) ) \,\dot\impparbinopname\,
                     \alpha ( \overline{\mathcal{A'}}\sbr{e_1} \gamma (\overline{d}) )}
        \tag{by Lemma 3 in App.~\ref{AppHelp}}
\\ & \dot\sqsubseteq
        \prod_{k' \in \alpha(\Kpsi)}
                    \confproj{k'}{ \overline{\mathcal{D'}}_{\alpha}\sbr{e_0} \overline{d}} \,\impparbinopname\,
                    \confproj{k'}{ \overline{\mathcal{D'}}_{\alpha}\sbr{e_1} \overline{d}}
        \tag{by IH, and def. of $\pi_{k'}$ and $\dot\impparbinopname$}
\\ & =
        \overline{\mathcal{D'}}_{\alpha}\sbr{\impbinop{e_0}{e_1}} \overline{d}
        \notag
\end{align*}
	\caption{Calculational derivation of $\overline{\mathcal{D}}_{\alpha}\sbr{\impbinop{e_0}{e_1}}$.
\label{fig:deriving-op}}
\vspace{-3mm}
\end{figure}

\begin{figure}[t]
\begin{scriptsize}
\begin{align*}
 \overline{\mathcal{D}}_{\alpha}\sbr{\impskip} &=
    \func{\overline{d}}{ \overline{d} }
\\
  \overline{\mathcal{D}}_{\alpha}\sbr{\impassign{\texttt{x}}{e}} &=
       \func{\overline{d}}{\prod_{k' \in \alpha(\Kpsi)}
           (\confproj{k'}{\overline{d}}) [ \texttt{x} \mapsto \confproj{k'}{\overline{\mathcal{D'}}_{\alpha}\sbr{e} \overline{d}} ] }
\\
  \overline{\mathcal{D}}_{\alpha}\sbr{\impseq{s_0}{s_1}} &=
    \overline{\mathcal{D}}_{\alpha}\sbr{s_1} \circ \overline{\mathcal{D}}_{\alpha}\sbr{s_0}
\\
  \overline{\mathcal{D}}_{\alpha}\sbr{\impif{e}{s_0}{s_1}} &=
    \func{\overline{d}}{\overline{\mathcal{D}}_{\alpha}\sbr{s_0} \overline{d} \,\dot\sqcup\, \overline{\mathcal{D}}_{\alpha}\sbr{s_1} \overline{d}}
\\
  \overline{\mathcal{D}}_{\alpha}\sbr{\impwhile{e}{s}} &=
    \mathrm{lfp}
      \func{\varsignfuncname}{\func{\overline{d}}{
           { \overline{d} ~\dot\sqcup~ \varsignfunc{\overline{\mathcal{D}}_{\alpha}\sbr{s} \, \overline{d}} }}}
\\[1.5ex]
  \overline{\mathcal{D}}_{\alpha}\sbr{\impifdef{(\theta)}{s}} &=
       \func{\overline{d}}{
\prod_{k' \in \alpha(\Kpsi)}
         \left\{ \begin{array}{ll} \confproj{k'}{ \overline{\mathcal{D}}_{\alpha}\sbr{s} \overline{d}} & \quad \textrm{if} \ k' \models \theta \\[1.5ex]
         \confproj{k'}{ \overline{d}} \sqcup \confproj{k'}{ \overline{\mathcal{D}}_{\alpha}\sbr{s} \overline{d}} & \quad \textrm{if sat} (k' \!\land\! \theta) \land \textrm{sat} (k' \!\land\! \neg \theta) \\[1.5ex]
		\confproj{k'}{\overline{d}} & \quad \textrm{if} \  k' \models \neg \theta \end{array} \right.
         }
\\[1.5ex]
 \overline{\mathcal{D'}}_{\alpha}\sbr{\mathit{n}} &=
        \func{\overline{d}}{
          \prod_{k' \in \alpha(\Kpsi)}
            \texttt{n} }
\\
  \overline{\mathcal{D'}}_{\alpha}\sbr{\texttt{x}} &=
        \func{\overline{d}}{
          \prod_{k' \in \alpha(\Kpsi)}{
            \confproj{k'}{\overline{d}} (\texttt{x}) }}
\\
  \overline{\mathcal{D'}}_{\alpha}\sbr{\impbinop{e_0}{e_1}} &=
        \func{\overline{d}}{
          \prod_{k' \in \alpha(\Kpsi)}{
            \impparbinop
              { \confproj{k'}{ \overline{\mathcal{D'}}_{\alpha}\sbr{e_0} \overline{d}}}
              { \confproj{k'}{ \overline{\mathcal{D'}}_{\alpha}\sbr{e_1} \overline{d}}} }}
\end{align*}%
\end{scriptsize}%
\vspace{-3mm}
\caption{Definitions of
$\overline{\mathcal{D}}_{\alpha}\sbr{\overline{s}} : (\mathbb{A} \to \mathbb{A})^{\alpha(\Kpsi)}$
and
$\overline{\mathcal{D'}}_{\alpha}\sbr{\overline{e}} : (\Aa \to \Const)^{\alpha(\Kpsi)}$.%
\label{fig:absanalysis}}
\end{figure}

The derivations for other cases are similar and can be found in
App.~\ref{App1}.  The process results in the definitions of
$\overline{\mathcal D}_\alpha\sbr{s}$ and
$\overline{\mathcal{D'}}_\alpha\sbr{e}$ presented in
Fig.\,\ref{fig:absanalysis}.
Monotonicity of $\overline{\mathcal{D}}_{\alpha}\sbr{s}$ and
$\overline{\mathcal{D'}}_{\alpha}\sbr{e}$ is shown in App.~\ref{App3}.
Soundness of the abstracted analysis follows by
construction; more precisely the complete calculation constitutes an
inductive proof of the following theorem:
\begin{theorem}[Soundness of Abstracted Analysis] \label{theorem:stmtsoundness}
\begin{enumerate}
\item[(i)]
  $\forall e \in \mathit{Exp}, (\alpha,\gamma) \in Abs, \overline{d} \in \mathbb{A}^{\alpha(\Kpsi)}:~ \alpha \circ \overline{\mathcal{A'}}\sbr{e} \circ \gamma (\overline{d}) ~\pwsqsubseteq~ \overline{\mathcal{D'}}_{\alpha}\sbr{e} \, \overline{d}$

\item[(ii)]  $\forall s \in \mathit{Stm}, (\alpha,\gamma) \in Abs, \overline{d} \in \mathbb{A}^{\alpha(\Kpsi)}:~ \alpha \circ \overline{\mathcal{A}\phantom'}\sbr{s} \circ \gamma (\overline{d}) ~\pwsqsubseteq~ \overline{\mathcal{D}\phantom'}_{\alpha}\sbr{s} \, \overline{d}$
\end{enumerate}
\end{theorem}

\begin{example} \label{ex-three}
	Consider the program $\ensuremath{S_1}$ from Example\,\ref{exp_lif2}, with
	$\Kpsi = \{ A \land B, A \land \neg B, \neg A \land B \}$.  We calculate
	$\overline{\mathcal{D}}_{\alpha_1}\sbr{S_1}$ for
	$\alpha_1=\joina{A}$. Following the rules of Fig.\,\ref{fig:absanalysis}, we
	obtain the following confounded abstract execution off all configurations
	containing the feature $A$:
   \begin{align}
     & \big([\texttt{x} \! \mapsto \! \top] \big)
 \stackrel{\overline{\mathcal{D}}_{\alpha_1}\sbr{\texttt{x} := 0}}{\longmapsto}
        \big([\texttt{x} \! \mapsto \! 0] \big)
  \stackrel{\overline{\mathcal{D}}_{\alpha_1}\sbr{\mathsf{\#if} \, (A) \, \texttt{x} := \texttt{x} + 1}}{\longmapsto}
        \big( [\texttt{x} \! \mapsto 1] \big)
 \stackrel{\overline{\mathcal{D}}_{\alpha_1}\sbr{\mathsf{\#if} \, (B) \, \texttt{x} := 1}}{\longmapsto}
        \big( [\texttt{x} \! \mapsto \! 1] \big)
      \notag
   \end{align}
	In the last step we used
	$\overline{\mathcal{D}}_{\alpha_1}\sbr{\impifdef{\!\!(B)}{\texttt{x} :=
			1}}([\texttt{x} \mapsto 1]) = ([\texttt{x} \mapsto 1]) ~\dot\sqcup~
	\overline{\mathcal{D}}_{\alpha_1}\sbr{\texttt{x} := 1}([\texttt{x} \mapsto
	1])$ since $((A \land B) \lor (A \land \neg B)) \land B$ and $((A \land B) \lor (A \land \neg B)) \land \neg B$ are both satisfiable.  The
	final result shows that the value of $\texttt{x}$ is the constant 1 for
	every configuration that satisfies $A$.
	%
	On the other hand, for the program $\ensuremath{S_2}$ and the same abstraction we obtain
	$\overline{\mathcal{D}}_{\alpha_1}\sbr{S_2}([\texttt{x} \mapsto
	\!\!\top]) = ([\texttt{x} \mapsto \!\!\top])$, so the value of $x$ is lost
	(approximated) by $\overline{\mathcal D}_{\alpha_1}$.
	\qed
\end{example}

We may implement the abstracted analysis in
Fig.~\ref{fig:absanalysis} directly by using Kleene's fixed point theorem
to calculate fixed points of loops iteratively.
But, we can also extract corresponding data-flow equations, and then apply the known
iterative algorithms to calculate fixed-point solutions.
We assume that the individual
statements are uniquely labelled with labels $\ell$.
Given an abstraction $\alpha$,
for each statement $s^\ell$ we generate two abstracted stores $\varstmtin{s^\ell}^{\alpha},
\varstmtout{s^\ell}^{\alpha} : \mathbb{A}^{\alpha(\Kpsi)}$, which describe the input and output abstract store
for all configurations before and after executing the statement $s^\ell$.
They are related with the definitions for abstracted analysis $\overline{\mathcal{D}}_{\alpha}$
given in Fig.~\ref{fig:absanalysis} as follows:
for each statement $s$ the input store $\varstmtin{s^\ell}^{\alpha}$
is substituted for the parameter $\overline{d}$, and the output store
$\varstmtout{s^\ell}^{\alpha}$ for the value of the corresponding function.
Some variability dependent data-flow equations are given in Fig~\ref{fig:lifteddataflow}.
The complete list of data-flow equations along with the proof of their soundness can be found in App.~\ref{App4}.

\begin{figure}[t]
\begin{scriptsize}
\vspace{-2mm}
\begin{align*}
  \underline{\forall k' \in \alpha(\Kpsi)}{:}~
    \confproj{k'}{ \varstmtout{\impassign[\ell]{\texttt{x}}{e^{\ell_0}}}^{\alpha} }
                           &= \confproj{k'}{\varstmtin{\impassign[\ell]{\texttt{x}}{e^{\ell_0}}}}^{\alpha}
                              [ \texttt{x} \mapsto \confproj{k'}{
                                   \overline{\mathcal{D}'}_{\alpha}\sbr{e^{\ell_0}}
                                     \varstmtin{\impassign[\ell]{\texttt{x}}{e^{\ell_0}}}^{\alpha}
                              } ]
\\[0.6em]
\underline{\forall k' \in  \alpha(\Kpsi)}{:}~
   \confproj{k'}{\varstmtout{\impifdef[\ell]{(\theta)}{s^{\ell_0}}}^{\alpha}}
                           &= \left\{ \begin{array}{ll} \confproj{k'}{\varstmtout{s^{\ell_0}}^{\alpha}} & \ \textrm{if} \ k' \models \theta \\[1.5ex]
         \confproj{k'}{\varstmtin{\impifdef[\ell]{(\theta)}{s^{\ell_0}}}^{\alpha}}  \sqcup \confproj{k'}{\varstmtout{s^{\ell_0}}^{\alpha}} & \ \textrm{if sat} (k' \!\!\land\!\! \theta) \! \land \! \textrm{sat} (k' \!\!\land\!\! \neg \theta) \\[1.5ex]
		\confproj{k'}{\varstmtin{\impifdef[\ell]{(\theta)}{s^{\ell_0}}}^{\alpha}} & \ \textrm{if} \  k' \models \neg \theta \end{array} \right.
\\ \underline{\forall k' \in  \alpha(\Kpsi)}{:}~
     \confproj{k'}{\varstmtin{s^{\ell_0}}^{\alpha}}
                           &= \confproj{k'}{\varstmtin{\impifdef[\ell]{(\theta)}{s^{\ell_0}}}^{\alpha}}
                                                     \quad\text{if sat} (k' \land \theta)
\end{align*}
\end{scriptsize}
\vspace{-2mm}
\caption{Selected data-flow\,equations\,for\,abstracted\,constant\,propagation.}
\label{fig:lifteddataflow}
\end{figure}

\section{Variability Abstraction with Syntactic Transformation}
\label{sec:reconfigurator}
The analyses $\overline{\mathcal{A}}$ and $\overline{\mathcal{D}}_{\alpha}$ can
be implemented either directly by using definitions of
Figs.\,\ref{fig:liftedanalysis} and \ref{fig:absanalysis}, or by extracting the
corresponding data-flow equations.  An entirely different way to implement
$\overline{\mathcal D}_\alpha$ is to execute the abstraction on the source
program, before running the analysis, and then running the previously existing
analysis $\overline{\mathcal{A}}$ on this transformed program.  We take this
route as it allows to completely reuse the effort invested in designing and
implementing $\overline{\mathcal A}$.

Any \VarIMP\ program $s$ with sets of features $\Ff$ and valid configurations \Kk\ is translated into a
corresponding abstract program $\alpha(s)$ with  corresponding set of features $\alpha(\Ff)$
and set of valid configurations $\alpha(\Kk)$.  We define the translation recursively over the
structure of $\alpha$.  All statements other than \texttt{\#if}\ are copied.
For example, $\alpha(\impskip)=\impskip$ and $\alpha(\impseq{s_0}{s_1}) =
\impseq{\alpha(s_0)}{\alpha(s_1)}$.   We discuss the rewrites for \texttt{\#if}\ statements below.

In the rewrite, we associate a fresh feature name $Z \notin \Ff$, with
every join abstraction $\joina{}$ (consequently written $\joina{\mquoted{Z}}$).  The new
feature $Z$ is an abstract name (renaming) of the compound formula $\bigvee_{k \in \Kk} k$.
It denotes the single valid configuration obtained from $\joina{}$.  The new
feature name is used to simplify conditions in the transformed code.  The
$\joina{\mquoted{Z}}$ rewrite is defined as follows:%
\begin{align*}
	&  \joina{\mquoted{Z}}( \mathbb F  ) = \{ Z \},
	\quad \joina{\mquoted{Z}}( \Kk ) = \{ Z \}
\\[1mm] &
	\joina{\mquoted{Z}}( \impifdef{(\theta)}s ) =
		\begin{cases}
			\impifdef{(Z)}{\joina{\mquoted{Z}}(s)}
				& \text{if } \ \bigvee_{k \in \Kk} k \models \theta \\[2.5mm]
			\impifdef{(Z)}{\texttt{lub}(\joina{\mquoted{Z}}(s),\impskip)} &
					\text{if sat}(\bigvee_{k \in \Kk } k \!\land\! \theta) \land \\
				& \qquad \text{sat}(\bigvee_{k \in \Kk } k \!\land\! \neg\theta) \\
			\impifdef{(\neg Z)}{\joina{\mquoted{Z}}(s)}
				& \text{if } \bigvee_{k \in \Kk } k \models \neg \theta
		\end{cases}
\end{align*}
In effect of applying the $\joina{\mquoted{Z}}$ transformation to any program $s$ we
obtain a single variant program, i.e.\ a SPL with only one valid product where
the feature $Z$ is enabled.  It can be analyzed with existing single-program
analyses. Note that it enables performing family-based analyses with
implementations of single-program analyses, albeit with loss of precision.  The
newly introduced statement $\texttt{lub}(s_0,s_1)$ represents the least upper
bound (join) of the results obtained by executing $s_0$ and $s_1$. This is the
only language-dependent aspect of \texttt{reconfigurator}. It can have
different implementations depending on the programming language and the
analysis we work with.
In our case, we exploit the fact that
$\overline{\mathcal{A}}\sbr{\impif{e}{s_0}{s_1}}$ ignores the branching
condition (cf.\ Fig.\,\ref{fig:liftedanalysis}) and use
$\texttt{lub}(s_0,s_1)=\impif{(n)}{s_0}{s_1}$ for some fixed integer $n$.
Finally, observe that  $\impifdef{(\neg Z)}{\joina{\mquoted{Z}}(s)}$ is equivalent to
$\impskip$, however it is useful to keep the statement in the program, which
makes it easy to merge programs when we use compound abstractions (below).


The rewrite for projection only changes the set of legal configurations:
\begin{equation*}
	\proja\phi (\Ff) = \Ff, \quad
	\proja\phi (\Kk) = \{ k \!\in\! \Kk \mid k \!\models\! \phi \}, \quad
	\proja\phi (\impifdef{\!(\theta)}{s})
		= \impifdef{\!(\theta)}{\proja\phi(s)}
\end{equation*}
%
%
%
%
Note that the general scheme for the basic rewrites of $\texttt{\#if}$
statements can be summarized as $\alpha(\impifdef{(\theta)}{s}) =
\impifdef{(\overline{\alpha}(\theta))}{\overline{\alpha}(s,\theta)}$, where
$\overline{\alpha}$ are functions transforming the condition $\theta$ and the
statement $s$.  It is easy to extract $\overline{\alpha}(\theta)$ and
$\overline{\alpha}(s,\theta)$ from the above rewrites for \joina{\mquoted{Z}} and
\proja{\phi}.  We will use them in defining transformations for binary
operators.

Now, for the case of parallel composition $\alpha_1 \otimes \alpha_2$, recall
that the set $\alpha_{1} \otimes \alpha_{2} ( \Kk  )$ is the union of
$\alpha_{1}( \Kk )$ and $\alpha_{2}( \Kk )$.  However in the rewrite semantics,
we are sometimes modifying the set of features. If $\alpha_1(\mathbb F)
\neq \alpha_2(\mathbb F)$ then some of valid configurations in $\alpha_1 (\Kk)
\cup \alpha_2 (\Kk)$ will not assign truth values to all features in
$\alpha_1(\mathbb F) \cup \alpha_2(\mathbb F)$.  To take a meaningful union of
configurations, we need to first unify their alphabets. To achieve this aim,
each valid configuration can be extended by information that the missing
features are excluded from it (negated).  Now the rewrite rules for parallel composition
are given by:
\begin{align*}
   &  \alpha_1 \!\otimes\! \alpha_2 (\Ff) \!=\! \alpha_1 (\Ff) \cup \alpha_2 (\Ff)
\\[0.5mm] &
      \alpha_1 \!\otimes\! \alpha_2 (\Kk) \!=\!
			\{ k_1\!\land\! {\scriptstyle{\bigwedge_{f\in\alpha_2(\Ff)
				\setminus \alpha_1(\Ff) }}}\!\!\neg f
				\!\mid\! k_1 \!\in\! \alpha_1(\Kk) \}
			\cup \{ k_2 \!\land\! {\scriptstyle{\bigwedge_{f\in\alpha_1(\Ff)
				\setminus \alpha_2(\Ff) }}} \!\!\neg f
				\!\mid\! k_2 \!\in\! \alpha_{2}(\Kk) \}
\\[1.3mm] &
      \alpha_{1} \otimes \alpha_{2} ( \impifdef{(\theta)}{s}  ) =
            \left\{ \begin{array}{ll}  \impifdef{\big(\overline{\alpha_1}(\theta) \lor \overline{\alpha_2}(\theta)\big)\,}{\overline{\alpha_1}(s,\theta)} & \quad  \textrm{if} \ \overline{\alpha_{1}} (s,\theta) = \overline{\alpha_{2}} (s,\theta) \\[1.5ex]
         \alpha_{1} ( \impifdef{(\theta)}{s}  ); \alpha_{2} ( \impifdef{(\theta)}{s}  ) & \quad  \textrm{otherwise } \end{array} \right.
\end{align*}
Observe that the second case of the parallel composition transformation can
only appear if the second case of a join transformation has been used somewhere
in recursive rewriting of $s$ (perhaps deep). All the other rewrites leave $s$
intact. However, in such case the branches have disjoint feature alphabets, as
every join is using a fresh feature name as parameter.  This ensures that only
one of the sequenced copies of $s$, $\overline{\alpha_{1}} (s,\theta)$ and $\overline{\alpha_{2}} (s,\theta)$,
will actually be executed (and the other
will amount to skip) in any given configuration of the product.

For sequential composition of abstractions $\alpha_2 \circ \alpha_1$ we use the
following  rewrites: $\alpha_2 \circ \alpha_1 (\Ff) = \alpha_2
(\alpha_1(\Ff))$,
      $\alpha_{2} \circ \alpha_{1} ( \Kk  ) =  \alpha_{2} ( \alpha_{1} ( \Kk  )  )$ and
$\alpha_{2} \circ \alpha_{1} ( \impifdef{(\theta)}{s}  ) =
             \impifdef{\big( \overline{\alpha_{2}} ( \overline{\alpha_{1}}(\theta))\big) \,}{\overline{\alpha_{2}} ( \overline{\alpha_{1}}(s,\theta), \overline{\alpha_{1}}(\theta))}$.

\begin{example} \label{ex-transform2}
Consider the program $S'_1$:
\(
\impifdef{\!\!(A)}{\texttt{x} := \texttt{x}+1}; \impifdef{\!(B)}{\texttt{x} :=
	1} \)\\
with $\Ff = \{ A, B\}$, $\psi = A \lor B$, and $\Kpsi = \{ A \land B, A \land \neg B, \neg A \land B \}$. Then
\begin{equation}
\joina{\mquoted{Z}}\circ\proja A(S'_1)=\impifdef{(Z)}{\texttt{x} := \texttt{x}+1}; \impifdef{(Z)}{\texttt{lub}(\texttt{x} := 1,\impskip)}\label{eq:join-proj-rewrite}
\end{equation}
The set of valid configurations after projection is changed to $\{A\land B,
A\land\neg B\}$, and after join again to just $\{Z\}$.  The obtained program
has only one configuration, the one that satisfies $Z$. The projection does not
change the statements of the program.  The join rewrite however, simplifies the
first \texttt{\#if} (it is statically determined; cf.\ the first case of
\joina{\mquoted{Z}} transformation), and joins the second statement with \texttt{skip} as
it is unknown whether it will be executed or not, in the lack of information
about the assignment to $B$ in the abstracted program. Note that since $Z$ is the only
one valid configuration, the obtained program is equivalent to:
$\texttt{x} := \texttt{x}+1; \texttt{lub}(\texttt{x} := 1,\impskip)$. Similarly,
we can calculate:
$\joina{\mquoted{Z}}\circ\proja B(S'_1) = \impifdef{(Z)}{ \texttt{lub}(\texttt{x} := \texttt{x}+1,\impskip)}; \impifdef{(Z)}{\texttt{x} := 1}$.

Now consider $((\joina{\mquoted{Z}} \circ \proja A) \otimes \proja B)(S_1')$.
The new set
of features is $\{Z, A, B\}$.  The subset \{A,B\} is retained from the right
projection component, and \{Z\} comes from the left join-project component.
After extending the configurations of both components with negations of absent
feature names we get the following set of valid configurations:
$\Kk'=\{ Z \land \neg A \land \neg B, \neg Z \land A \land B,
	\neg Z \land \neg A \land B\}$.
The result of the left join-project operand is the program
\eqref{eq:join-proj-rewrite}, and the right rewrite (projection) never changes
the statements, so its result is identical to $S_1'$. Thus we are composing
programs \eqref{eq:join-proj-rewrite} and $S_1'$ using the parallel composition
rewrites.
Then $((\proja Z\circ\joina A)
\otimes \proja B)(S'_1)$ is:
\begin{equation*}
\impifdef{(Z \lor A)}{\texttt{x} := \texttt{x}+1};
 \impifdef{(Z)}{\texttt{lub}(\texttt{x} := 1,\impskip)};
 \impifdef{(B)}{\texttt{x} := 1}
\end{equation*}
 The first \texttt{\#if}\ has been unified using the first case of the
transformation for $\otimes$, and the second \texttt{\#if} is transformed into two copies of
the statement with different guards, using the second case of the rewrite
definition.  For any legal configuration in $\Kk'$ at most one of them does not reduce to
\texttt{skip}.\qed

\end{example}

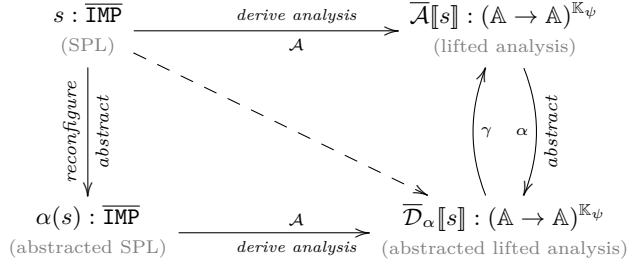
\begin{figure}[t]
\vspace{-3mm}
$$
\xymatrix@C=6pc@R=4pc{
\text{\begin{tabular}{c}$s : \overline{\texttt{IMP}}$\\{\scriptsize\textcolor{gray}{(SPL)}}\end{tabular}} \ar@{-->}[rd] \ar[r]_{\mathcal{A}}^{\textit{derive analysis}} \ar[d]^{\rotatebox{90}{\scriptsize\textit{abstract}}}_{\rotatebox{90}{\scriptsize\textit{reconfigure}}} & \ar@/^1.0pc/[d]_{\alpha}^{\rotatebox{90}{\scriptsize\textit{abstract}}} \text{\begin{tabular}{c}$\overline{\mathcal{A}}[\![s]\!] : (\mathbb{A} \rightarrow \mathbb{A})^{\Kk_\psi}$\\{\scriptsize\textcolor{gray}{(lifted analysis)}}\end{tabular}} \\
\text{\begin{tabular}{c}$\alpha(s) :\overline{\texttt{IMP}}$\\{\scriptsize\textcolor{gray}{(abstracted SPL)}}\end{tabular}} \ar[r]^{\mathcal{A}}_{\textit{derive analysis}} & \ar@/^1.0pc/[u]_{\gamma} \text{\begin{tabular}{c}$\overline{\mathcal{D}}_\alpha[\![s]\!] : (\mathbb{A} \rightarrow \mathbb{A})^{\Kk_\psi}$\\{\scriptsize\textcolor{gray}{(abstracted lifted analysis)}}\end{tabular}} \\
}
$$
\vspace{-3mm}
\caption{Illustration of \emph{derive} vs \emph{abstract}: $\overline{\mathcal{D}}_{\alpha}\sbr{s} = \overline{\mathcal{A}}\sbr{\alpha(s)}$.}
\label{fig:commutative}
\end{figure}

Now the analysis $\overline{\mathcal{A}}\sbr{\alpha(s)}$ and
$\overline{\mathcal{D}}_{\alpha}\sbr{s}$ coincide up to renaming of valid configurations.
So the \texttt{reconfigurator} together with an existing implementation of
$\overline{\mathcal{A}}$ gives us the abstracted analysis
$\overline{\mathcal{D}}_{\alpha}$.  The above equality is illustrated
by Fig.\,\ref{fig:diagram}.

\begin{theorem} \label{theorem:correctness} $\\$
 $\forall s \in Stm, \alpha: \mathbb{A}^{\Kpsi} \to \mathbb{A}^{\alpha(\Kpsi)}  \in Abs, \overline{d} \in \mathbb{A}^{\alpha(\Kpsi)}:~
    \overline{\mathcal{D}}_{\alpha}\sbr{s} \, \overline{d} = \overline{\mathcal{A}}\sbr{\alpha(s)} \, \overline{d}$~\footnote{The proof of this theorem is in App.~\ref{App5}.}.
\end{theorem}

\begin{example} \label{exp_comm1}
Consider the program $\ensuremath{S_1}$ from Example~\ref{exp_lif2} with
$\Kpsi = \{ A \land B, A \land \neg B, \neg A \land B \}$.
We have calculated in Example~\ref{ex-three} that $\overline{\mathcal{D}}_{\joina{A}}\sbr{S_1}([\texttt{x} \! \mapsto \! \top])=([\texttt{x} \! \mapsto \! 1])$.
We now calculate $\overline{\mathcal{A}}\sbr{\joina{A,\mquoted{Z}}(S_1)}([\texttt{x} \! \mapsto \! \top])$
(here $\joina{A,\mquoted{Z}}=\joina{\mquoted{Z}} \circ \proja A$):
   \begin{align}
     & \big([\texttt{x} \! \mapsto \! \top] \big)
 \stackrel{\overline{\mathcal{A}}\sbr{\texttt{x} := 0}}{\longmapsto}
        \big([\texttt{x} \! \mapsto \! 0] \big)
  \stackrel{\overline{\mathcal{A}}\sbr{\mathsf{\#if} \, (Z) \, \texttt{x} := \texttt{x} + 1}}{\longmapsto}
        \big( [\texttt{x} \! \mapsto 1] \big)
 \stackrel{\overline{\mathcal{A}}\sbr{\mathsf{\#if} \, (Z) \, \texttt{lub}(\texttt{x} := 1,\texttt{skip})}}{\longmapsto}
        \big( [\texttt{x} \! \mapsto \! 1] \big)
      \notag
   \end{align}
\end{example}

\section{Evaluation} \label{sec:evaluation}
Recall that there are two ways to speed up lifted analyses: improving
\emph{representation} and increasing \emph{abstraction}.
First, we will compare the performance of the two using an unoptimized
lifted analysis as a baseline.  Then, we demonstrate that
abstraction may be used to turn previously infeasible analysis into
feasible ones.  Finally, we consider example scenarios that use
projection and join and show that abstraction may be applied to an
entire product line or when just analyzing a single method.

For our experiments, we use an existing implementation of lifted data-flow
analyses for Java Object-Oriented SPLs~\cite{TAOSD}.  The implementation is based on
SOOT's intra-procedural data-flow analysis framework~\cite{SOOT} for
analyzing Java programs.  It uses CIDE (Colored IDE)~\cite{KastnerPhD} to
annotate statements using background colors rather than
$\texttt{\#ifdef}$ directives.  Every feature is thus associated with
a unique color.

We will consider an unoptimized lifted intra-procedural analysis,
known as $\mathcal{A}2$ (from \cite{TAOSD}), that uses
$|\Kk_\psi|$-tuples of analysis information, one analysis value per
configuration.  Also, we consider $\mathcal{A}3$
(from \cite{TAOSD}) which is the same analysis, but with improved
representation via sharing of analysis-equivalent configurations using
a high-performance bit vector library.  Note that
$\mathcal{A}2$ corresponds to $\overline{\mathcal A}$ in
Fig.~\ref{fig:liftedanalysis} and we will thus refer to it as
$\overline{\mathcal A}$, while we will use $\overline{\mathcal S}$ for
the analysis with sharing ($\mathcal{A}3$ in \cite{TAOSD}).
The performance of abstracted analyses depends on the size of tuples
they work on. Therefore
as variability abstractions, we have chosen $\overline{\mathcal
  D}_{\joina{}}$ which joins together (confounds) information from all
configurations down to just one abstracted analysis value, and
$\overline{\mathcal  D}_{\proja{N/2} \otimes \joina{N/2}}$ (where $N=|\Kk_\psi|$) which
is a parallel composition of a projection of $1/2$ (randomly selected) configurations and
a \emph{join} of the remaining $1/2$ configurations.
We abbreviate them as $\Done_1$ and $\Done_{N/2}$ in the following.  We have chosen those
variability abstractions because they represent the coarsest abstraction
$\Done_1$ that works on 1-sized tuples, and the medium abstraction
$\Done_{N/2}$ that works on $N/2$-sized tuples.
Any other abstraction will have a speed up anywhere
between $\Atwo$ (no abstraction), $\Done_{N/2}$ (medium abstraction) and $\Done_1$ (maximum abstraction).
It thus quantifies the potential of abstractions.

\begin{figure}[t]
\begin{center}
\begin{footnotesize}
\begin{tabular}{||l|c|c|c||l|c|c|c||}                                                                           \hhline{|t:====:t:=t=t=t=:t|}
~Benchmark~              & ~avg.\,$|\Kk_\psi|$~ & ~$|\Ff|$~ & ~LOC~      & ~max variability mth~ & ~$|\Kk_\psi|$~ & ~$|\Ff|$~ & ~LOC~ \\[-0.2ex] \hhline{|:====::====:|}
 ~\texttt{Prevayler}~   & ~N=1.3~                 & 5             & ~8,000~    & ~\texttt{P'F'.publisher()}~ & ~N=8~ &  ~3~ & 10   \\[-0.2ex] \hhline{||----||----||}
 ~\texttt{BerkelyDB}~  & ~N=1.6~                 & 42           &  ~84,000~ & ~\texttt{DBRunAction.main()}~   & ~N=40~ &             ~7~ & 165  \\[-0.2ex] \hhline{||----||----||}
 ~\texttt{GPL}~            & ~N=3.9~                & 18           &  ~1,350~   & ~\texttt{Vertex.display()}~          & ~N=106~ &             ~9~ & 31 \\[-0.2ex] \hhline{|b:====:b:=b=b=b=:b|}
\hline
\end{tabular}
\caption{Characteristics of our three SPL benchmarks (average \#configurations in all methods in SPL, total \#features, and LOC) along with, for each SPL, its method with maximum variability (\#configurations, local \#features, and LOC).}\label{fig:benchmarks}
\end{footnotesize}
\end{center}
\end{figure}

For our experiment\footnote{
The implementation, benchmarks, and all results obtained from our experiments
are available in the supplemental material submitted with this paper.}, we have chosen two analyses:
\emph{reaching definitions} and \emph{uninitialized variables};  and three SPL
benchmarks~\cite{KastnerPhD}.  Graph PL (GPL) is a small desktop
application with intensive feature usage, Prevayler is a slightly
larger product line with low feature usage, and BerkelyDB is a larger
database library with moderate feature usage.
Fig.~\ref{fig:benchmarks} summarises relevant characteristics for
each benchmark: the average number of valid configurations in all methods in the SPL,
the total number of features in the entire SPL, the total number of
lines of code (LOC).  Also, for each SPL, the figure details
information about the method with the highest variability (most
configurations): its number of valid configurations, features, and
lines of code.

\begin{figure*}[t]
\centering
\subfigure{
\begin{tabular}{c}
\begin{minipage}{3.6cm}
\centering
\begin{scriptsize}
  \begin{barenv}
    \setwidth{7}
    \sethspace{0.2}
    \setstretch{0.25} 
    \setyaxis{0}{200}{50}
    \setyname{$\mu$s}
    \setnumberpos{empty}
    \hspace{1.5}
   \shadebar{184}{20}[{\begin{tabular}{c}\\[1.3ex]{}$\Atwo$\\{\scriptsize{}}\end{tabular}}]
    \hspace{1.5}
   \shadebar{178}{30}[{\begin{tabular}{c}\\[1.3ex]{}$\Athree$\\{\scriptsize{}1.0x}\end{tabular}}]
    \hspace{1.5}
   \shadebar{126}{60}[{\begin{tabular}{c}\\[1.3ex]{}$\Done_{N/2}$\\{\scriptsize{}1.5x}\end{tabular}}]
    \hspace{1.5}
   \shadebar{67}{100}[{\begin{tabular}{c}\\[1.3ex]{}$\Done_1$\\{\scriptsize{}2.7x}\end{tabular}}]
    \hspace{1.5}
  \end{barenv}
\end{scriptsize}
\vspace*{0.3cm}\phantom{x}
~\\
\begin{scriptsize}
  \begin{barenv}
    \setwidth{7}
    \sethspace{0.2}
    \setstretch{0.2} 
    \setyaxis{0}{250}{50}
    \setyname{}
    \setnumberpos{empty}
    \hspace{1.5}
   \shadebar{238}{20}[{\begin{tabular}{c}\\[1.3ex]{}$\Atwo$\\{\scriptsize{}}\end{tabular}}]
    \hspace{1.5}
   \shadebar{193}{30}[{\begin{tabular}{c}\\[1.3ex]{}$\Athree$\\{\scriptsize{}1.2x}\end{tabular}}]
    \hspace{1.5}
   \shadebar{153}{60}[{\begin{tabular}{c}\\[1.3ex]{}$\Done_{N/2}$\\{\scriptsize{}1.6x}\end{tabular}}]
    \hspace{1.5}
   \shadebar{85}{100}[{\begin{tabular}{c}\\[1.3ex]{}$\Done_1$\\{\scriptsize{}2.8x}\end{tabular}}]
    \hspace{1.5}
  \end{barenv}
\end{scriptsize}\\
~\\
~\\
\small\texttt{Prevayler::publisher()}\\
N=8\\
\end{minipage}\\
\end{tabular}
}
\subfigure{
\begin{tabular}{c}
\begin{minipage}{3.6cm}
\centering
\begin{scriptsize}
  \begin{barenv}
    \setwidth{7}
    \sethspace{0.2}
    \setstretch{0.333333333} 
    \setyaxis{0}{150}{50}
    \setyname{ms}
    \setnumberpos{empty}
    \hspace{1.5}
    \shadebar{123}{20}[{\begin{tabular}{c}\\[1.3ex]{}$\Atwo$\\{\scriptsize{}}\end{tabular}}]
    \hspace{1.5}
    \shadebar{104}{30}[{\begin{tabular}{c}\\[1.3ex]{}$\Athree$\\{\scriptsize{}1.2x}\end{tabular}}]
    \hspace{1.5}
    \shadebar{61}{60}[{\begin{tabular}{c}\\[1.3ex]{}$\Done_{N/2}$\\{\scriptsize{}2.0x}\end{tabular}}]
    \hspace{1.5}
    \shadebar{5.05}{100}[{\begin{tabular}{c}\\[1.3ex]{}$\Done_1$\\{\scriptsize{}24x}\end{tabular}}]
    \hspace{1.5}
 \end{barenv}
\end{scriptsize}
\vspace*{0.3cm}\phantom{x}
~\\
\begin{scriptsize}
  \begin{barenv}
    \setwidth{7}
    \sethspace{0.2}
    \setstretch{2.5} 
    \setyaxis{0}{20}{5}
    \setyname{}
    \setnumberpos{empty}
    \hspace{1.5}
   \shadebar{18.0}{20}[{\begin{tabular}{c}\\[1.3ex]{}$\Atwo$\\{\scriptsize{}}\end{tabular}}]
    \hspace{1.5}
   \shadebar{14.8}{30}[{\begin{tabular}{c}\\[1.3ex]{}$\Athree$\\{\scriptsize{}1.2x}\end{tabular}}]
    \hspace{1.5}
   \shadebar{9.4}{60}[{\begin{tabular}{c}\\[1.3ex]{}$\Done_{N/2}$\\{\scriptsize{}1.9x}\end{tabular}}]
    \hspace{1.5}
   \shadebar{1.54}{100}[{\begin{tabular}{c}\\[1.3ex]{}$\Done_1$\\{\scriptsize{}12x}\end{tabular}}]
    \hspace{1.5}
  \end{barenv}
\end{scriptsize}\\
~\\
~\\
\small\texttt{BerkeleyDB::main()}\\
N=40\\
\end{minipage}\\
\end{tabular}
}
\subfigure{
\begin{tabular}{c}
\begin{minipage}{3.6cm}
\centering
\begin{scriptsize}
  \begin{barenv}
    \setwidth{7}
    \sethspace{0.2}
    \setstretch{1} 
    \setyaxis{0}{50}{10}
    \setyname{ms}
    \setnumberpos{empty}
    \hspace{1.5}
    \shadebar{45}{20}[{\begin{tabular}{c}\\[1.3ex]{}$\Atwo$\\{\scriptsize{}}\end{tabular}}]
    \hspace{1.5}
    \shadebar{25}{30}[{\begin{tabular}{c}\\[1.3ex]{}$\Athree$\\{\scriptsize{}1.8x}\end{tabular}}]
    \hspace{1.5}
    \shadebar{27}{60}[{\begin{tabular}{c}\\[1.3ex]{}$\Done_{N/2}$\\{\scriptsize{}1.7x}\end{tabular}}]
    \hspace{1.5}
    \shadebar{0.95}{100}[{\begin{tabular}{c}\\[1.3ex]{}$\Done_1$\\{\scriptsize{}47x}\end{tabular}}]
    \hspace{1.5}
 \end{barenv}
\end{scriptsize}
\vspace*{0.3cm}\phantom{x}
~\\
\begin{scriptsize}
  \begin{barenv}
    \setwidth{7}
    \sethspace{0.2}
    \setstretch{2.5} 
    \setyaxis{0}{20}{5}
    \setyname{}
    \setnumberpos{empty}
    \hspace{1.5}
   \shadebar{17.2}{20}[{\begin{tabular}{c}\\[1.3ex]{}$\Atwo$\\{\scriptsize{}}\end{tabular}}]
    \hspace{1.5}
   \shadebar{7.86}{30}[{\begin{tabular}{c}\\[1.3ex]{}$\Athree$\\{\scriptsize{}2.2x}\end{tabular}}]
    \hspace{1.5}
   \shadebar{8.91}{60}[{\begin{tabular}{c}\\[1.3ex]{}$\Done_{N/2}$\\{\scriptsize{}1.9x}\end{tabular}}]
    \hspace{1.5}
   \shadebar{0.62}{100}[{\begin{tabular}{c}\\[1.3ex]{}$\Done_1$\\{\scriptsize{}28x}\end{tabular}}]
    \hspace{1.5}
  \end{barenv}
\end{scriptsize}\\
~\\
~\\
\small\texttt{GPL::display()}\\
N=106\\
\end{minipage}\\
\end{tabular}
}
\caption{Analysis time for \emph{reaching definitions} (above) and
  \emph{uninitialized variables} (below): $\Atwo$ (baseline) and
  $\Athree$ (sharing) vs.\ $\Done_{N/2}$ (medium abstraction) and
  $\Done_1$ (maximum abstraction).}
\label{fig:performance}
\end{figure*}

\paragraph{{Performance.}}
Fig.~\ref{fig:performance} shows the time it takes to run each of our
three maximum variability methods, as a relative comparion between
$\Atwo$ (baseline) and $\Athree$ (sharing) \emph{vs} $\Done_{N/2}$
(medium abstraction) and $\Done_1$ (maximum abstraction).
The experiments are executed on a 64-bit Intel$^\circledR$Core$^{TM}$
i5 CPU with 8 GB memory.
All times are reported as averages over ten runs with the highest and
lowest number removed.  For each benchmark method, we give the speed
up factor relative to the baseline (normalized with factor 1) and
recall the number of configurations, N.

Our experiment confirms previous results that sharing is indeed
effective and especially so for larger values of N~\cite{TAOSD}.  On our
methods, it translates to speed ups
(i.e., $\Atwo$~\textit{vs}~$\Athree$) anywhere between 3\% faster (for
N=8) and slightly more than twice as fast (for N=106).
We also observe that abstraction is not surprisingly significantly
faster than unabstracted analyses (i.e., $\Done$~\textit{vs}~$\Atwo$
and $\Athree$); i.e., abstraction yields significant performance
gains, especially for benchmarks with higher variability.  For
\texttt{GPL} with N=106, we see a dramatic 47 and 28 times speed up
depending on the analysis (i.e., $\Done_1~\textit{vs}~\Atwo$).  Also, we
note that increased abstraction is up to 26 times faster than improved
representation (i.e., $\Done_1~\textit{vs}~\Athree$).  In general, it is
obviously possible to combine the benefits from representation and
abstraction to yield even more efficient analyses.

\paragraph{{From Infeasible to Feasible Analysis.}}
Of course, for very large values of N, analyses may become
impractically slow or infeasible.  As an experiment, we took a large
method (\texttt{processFile()} from \texttt{BerkeleyDB}) and kept
adding unconstrained variability.  For N=$2^{13}$=8,192
configurations, the analysis $\Atwo$ took 138 seconds.  For
N=$2^{14}$=16,384, it ran more than ten minutes until it eventually
produced an out-of-memory error.  In contrast, variability abstraction
$\Done_1$ analyses the same high variability method in less than 8 ms
(albeit less precisely).  Hence, abstraction can not only speed up
analyses, but also turn previously infeasible analyses feasible.

\paragraph{{Projection on Entire SPL.}}
GPL is a family of classical graph applications with variability on
its representation and algorithms.  For instance, the features
\texttt{Directed} and \texttt{Undirected} control whether or not
graphs are \emph{directed}; \texttt{Weighted} and \texttt{Unweighted}
control whether or not the graphs are \emph{weighted}; and, the
features \texttt{BFS} and \texttt{DFS} control the search algorithm
used (\emph{breadth-first search} or \emph{depth-first search}).  It is
common industrial practice, to ship products with a subset of
configurations, and thereby functionality.  Here, we may use
projection to \emph{disable} features \texttt{BFS} and
\texttt{Undirected}, along with any features that only work on
undirected graphs: (\texttt{Connected}, \texttt{MSTKruskal}, and
\texttt{MSTPrim} for  implementing \emph{connected components} and \emph{minimum
  spanning trees} algorithms) which can be obtained from GPL's feature model,
detailing such feature dependencies.  With this projection
(abstraction), the configuration space of GPL is reduced from 528 to
370 valid configurations.  This, in turn, cuts analysis time of
reaching definitions in half (from 90ms to 49ms).  For 123 out of 135
methods, the abstracted analysis computes the exact same analysis
information.  For larger product lines and projections, lots of time
may be saved in this way.

\paragraph{{Join on One Method.}}
Figure~\ref{fig:BKDBcode} shows a fragment extracted from
\texttt{BerkeleyDB}'s \texttt{main()} method with N=40 valid
configurations.  A local variable, \texttt{doAction} is defined and
initialized to zero, after which it is conditionally assigned three
times in statements guarded by \texttt{\#ifdef}s.  (Actually, there
are two more similar \texttt{\#ifdef}s involving features
\texttt{Evictor} and \texttt{DeleteOp}, but we have omitted those for
brevity in the code fragment.)  We can use a join abstraction of the
reaching definitions analysis to compute what are the possible values
(definitions) that \emph{reach} the condition of the \texttt{switch}
statement in line 12.  An abstracted analysis would be able to
determine that these are the assignments in lines 1, 3, 6, 8, and 10, by
analyzing only \emph{one} crudely over-approximated configuration
instead of all (N=40) configurations.  In general, by inspecting the
structure of the code and the features used, we can tailor abstactions
that can analyze individual methods much faster than analyzing all
configurations.

\begin{figure}
\centering
\begin{scriptsize}
\begin{tabular}{ll}
\texttt{void} & \texttt{main(..) \{}\\
\textcolor{gray}{1} & \texttt{.. int doAction = 0; ..}\\ \hhline{~-}
\textcolor{gray}{2} & \cellcolor{lightgray!25}{\texttt{\textcolor{darkblue}{\underline{\#ifdef} Cleaner}}}\\
\textcolor{gray}{3} & \cellcolor{lightgray!25}{\texttt{if (..) doAction = CLEAN;}}\\
\textcolor{gray}{4} & \cellcolor{lightgray!25}{\texttt{\textcolor{darkblue}{\underline{\#endif}}}}\\ \hhline{~=}
\textcolor{gray}{5} & \cellcolor{lightgray!25}{\texttt{\textcolor{darkblue}{\underline{\#ifdef} INCompresser}}}\\
\textcolor{gray}{6} & \cellcolor{lightgray!25}{\texttt{if (..) doAction = COMPRESS;}}\\
\textcolor{gray}{7} & \cellcolor{lightgray!25}{\texttt{\textcolor{darkblue}{\underline{\#endif}}}}\\ \hhline{~-}
\textcolor{gray}{8} & \texttt{if (..) doAction = CHECKPOINT;}\\ \hhline{~-}
\textcolor{gray}{9} & \cellcolor{lightgray!25}{\texttt{\textcolor{darkblue}{\underline{\#ifdef} Statistics}}}\\
\textcolor{gray}{10} & \cellcolor{lightgray!25}{\texttt{if (..) doAction = DBSTATS;}}\\
\textcolor{gray}{11} & \cellcolor{lightgray!25}{\texttt{\textcolor{darkblue}{\underline{\#endif}}}}\\ \hhline{~-}
\textcolor{gray}{12} & \texttt{.. switch (doAction) \{ .. \} ..}\\
\texttt{\}} & \\
\end{tabular}
\end{scriptsize}
\caption{Code fragment extracted from \texttt{BerkeleyDB::main()} with N=40.}
\label{fig:BKDBcode}
\end{figure}

\section{Related Work}\label{sec:related}

Static analyses can be accelerated by devicing more efficient
representations or by introducing abstraction.  In family-based
analysis for software product lines the representation improvements
primarily rely on sharing state information for variants with
analysis-equivalent information (which implies reducing redundant
computation).  This can optimize the analyses
considerably~\cite{TAOSD,model-checking-spls,TypeCheckingSPL}.
However, in the worst case, the number of variants that a lifted
analysis has to consider is still inherently exponential in the number
of features, $|\Ff|$.  Thus with a large number
of features lifted analyses may become impractical or even infeasible.
In this work we have taken the alternate route of using
abstraction. Our experiments show that abstraction introduces
speed-ups independently of representation gains.  Thus our results can
be beneficially combined with efficient representations.

An efficient implementation of lifted analysis formulated within the
IFDS framework\,\cite{ifds} for inter-procedural distributive
environments was proposed in SPL$^\text{LIFT}$~\cite{SPLLIFT}.  It
uses binary decision diagrams to represent shared feature constraints.
The authors have found that the running time of analysing all variants
in a family is close to the analysis of a single-program. In such
case, further benefit of applying abstraction, as presented in this
paper, is unlikely to bring any significant improvement.  However,
notice that the method of SPL$^\text{LIFT}$ is limited only to
distributive data-flow analysis encoded within the IFDS framework.  Many
analyses, including constant propagation, are not distributive and
hence cannot be expressed in IFDS.  Let alone static analyses that are
not expressible as data-flow analyses (including type checking,
model-checking, etc).

The formal developments in this paper are based on \emph{variational
abstract interpretation}, a formal methodology for systematic
derivation of lifted analyses for $\texttt{\#ifdef}$-based product
lines, proposed in \,\cite{MBW}.  The method is based
on the calculational approach to abstract interpretation of
Cousot\,\cite{Cousot99}, applied and contextualized to product lines.
In that work, Galois connections are not used for lifting, but only
for derivation of single program analyses as shown in
\,\cite{Cousot99}, so they are variability-unaware.  Calculations are
used to derive a directly operational \emph{abstracted lifted
analysis} which is \emph{correct} by construction.  
In the present paper,
we assume that lifted analyses exist (possibly obtained using the
methodology of \cite{MBW}), and focus on abstracting variability using
them.  We devise an expressive calculus for specifying abstraction
operators. 
Also, thanks to our tool, all
abstractions specifiable in our calculus, are now automatically
executable.

A good collection of analyses that have been lifted manually is
presented in the survey \cite{TRfin2012}.  We should remark, that the
join operation \joina{} allows
applying single program analyses to program families, even if with
precision loss. In that sense, the our approach is the
first ever method that can \emph{automatically} lift single program
analyses to work on program families.
Besides the family-based strategy, the survey\,\cite{TRfin2012}
identifies a \emph{sampling strategy} as a suitable way of analyzing
product lines (see also\,\cite{apel.ea-sampling}). In the sampling strategy only a random subset
of products is analyzed. We remark that once the sample is selected,
our projection operator \fproja{\phi}\ can be used to realize the
sampling strategy in a simultanous way by exploiting an existing
family-based analysis. 

In fact, the agebraic specification framework of
Section\,\ref{sec:variability_abstractions} allows specifying any
analysis in the spectrum between a fully family-based analyses, and a
single variant, \emph{product-based}, analysis.  We can specify
abstractions that select (sample) any subsets of configurations and
then analyze this subset with selected choice of precision, either all
variants precisely, like in sampling, or confounding some executions
for efficiency.  In this sense, we show how to design analyses placed
anywhere in the design spectrum painted in\,\cite{TRfin2012}.
Consider, the \emph{feature-based} analysis strategy as an example. In
this strategy an analysis explores the program code feature-by-feature
(as opposed to configuration-by-configuration).  Analyses following
this strategy can now be systematically obtained using our
abstractions, by projecting away (ignoring) all but one feature and
running a single program analysis on the result.  This is quite
remarkable. It has been well recognized that designing such analyses
is very difficult, yet now there exists a systematic way of doing
that, so it is no longer an impenetrable art.  


\section{Conclusion}\label{sec:conclusion}

We have defined variability-aware abstractions given as Galois
connections, and used them to derive efficient and
correct-by-construction abstract analyses of program families.  We
have designed a calculus for the abstractions, and shown how
abstractions specified in this language can be applied not only on
analyses, but also on programs, obtaining a convenient implementation
strategy of the abstractions in form of a source-to-source \texttt{reconfigurator}
transformation.

The \texttt{reconfigurator} transformation presently requires that the
programming language is able to express \emph{sequential composition}
(e.g., ``\texttt{;}'' in \texttt{IMP}) and \emph{join of statements}
(i.e., \texttt{lub} as in ``$\sqcup$'') with respect to the analysis
in question.  It would be interesting to consider lifting those
assumptions in future, and apply this method to more modeling and programming languages.

We evaluated the method 
on three Java-based product lines.  We found that the abstractions
improve performance of analyses independently of improvements in
the data representations used in the implementations of these
analyses.  This indicates that the proposed abstraction strategies
will be instrumental in tackling error finding analysis in large
configurable software systems, like the Linux kernel.  Indeed we have
developed these techniques with the intention of scaling error finding
tools to such challenging cases in future.  Besides this, we would
like to experiment with applying these abstraction techniques to
alternative quality assurance methods including model checking, and
testing.

\bibliographystyle{splncs}

\newpage
\appendix



\section{Properties of Abstraction Operators} \label{App:oper}

We recall properties of Galois connections for completeness.

A pair $\poset{L}{\leq_L} \galois{\alpha}{\gamma}
\poset{M}{\leq_M}$ is a \emph{Galois connection} between complete lattices $L$
and $M$ iff $\alpha$ and $\gamma$ are total functions that satisfy: $\alpha(l)
\leq_M m \iff l \leq_L \gamma(m)$ for all $l \in L, m \in M$.

Some important properties of Galois connections \cite{CousotCousot92} include:
1) $\gamma \circ \alpha$ is \emph{extensive}, i.e.\ $l \, \leq_L
	\, (\gamma \circ \alpha)(l)$ for all $l \in L$;
2) $\alpha \circ \gamma$ is \emph{reductive}, i.e.\ $(\alpha \circ
	\gamma)(m) \, \leq_M \, m$ for all $m \in M$;
3) $\alpha$	is a \emph{complete join morphism} (CJM), i.e.\
   $\alpha(\bigsqcup_{l\in L'} l) ~~=~~ \bigsqcup_{l\in L'} \alpha(l)$
   for all $L' \subseteq L$.

\bigskip

\noindent Now we turn to proving theorems of
Sect.\ref{sec:basic-abstractions}.

\begin{proof}[Thm.\,\ref{thm:join-galois}]
	Let $\abar\in\Aa^\Kk$ and $a \in \Aa^{\joina{}(\Kk)} \equiv \Aa$;
	recall that $\joina{}(\Kk)$ is always a singleton.  We have:
\begin{align*}
   &  \joina{}(\abar) ~\pwsqsubseteq~ ( a )
\\ &\iff \left(\textstyle\bigsqcup_{k\in\Kk }
				\confproj{k}{\overline{a}}\right) \sqsubseteq (a)
      \tag{by def. of \joina{}}
\\ &\iff
       \forall k \in \Kk. \, \confproj{k}{\overline{a}} \sqsubseteq a
      \tag{by def. of $\sqcup$}
\\ &\iff
       \overline{a} \, \pwsqsubseteq\, \joing{}(a)
      \tag{by def. of \joing{}}
    \end{align*}
\qed
\end{proof}

\begin{proof}[Thm.\,\ref{thm:proj-galois}]
	Let $\abar\in\Aa^\Kk$ and $\abar'\in\Aa^{\{k \in \Kk \mid k \models \phi\}}$.
	We have:
	\begin{align*}
	&  \proja\phi(\abar) \, \pwsqsubseteq \, \abar'
\\ &\iff \forall k\in\Kk,  k\models\phi.\,
		\confproj k\abar \sqsubseteq \confproj k {\abar'}
		\tag{by def.\ of $\proja\phi$}
\\ &\iff \forall k \in \Kk, k \models \phi.\,
		\confproj{k}{\abar} \sqsubseteq \confproj{k}{\overline{a'}} \, \land \,
		\forall k\in\Kk, k \not\models \phi.\,
			\confproj k\abar \sqsubseteq \top
      \tag{by def.\ of $\top$}
\\ &\iff
       \abar \,\pwsqsubseteq\, \projg\phi(\abar')
      \tag{by def. of $\gamma^{proj}_{\phi}$}
    \end{align*} \qed
\end{proof}

\bigskip
For sequential composition Galois connection properties follow
directly from the definition and the standard results about
compositions of Galois connections. Let's consider the parallel
composition:
\begin{proof}[Thm.\,\ref{thm:product-galois}]
To verify that this defines a Galois connection, we calculate:
  \begin{align*}
   &  \alpha_{1} \otimes \alpha_{2} (\overline{a}) \, \dot\sqsubseteq \, \overline{a'}
\\ &\iff \alpha_{1} (\overline{a}) \, \times \, \alpha_{2} (\overline{a}) \dot\sqsubseteq \, \overline{a'}
      \tag{by def. of $\alpha_{1} \otimes \alpha_{2}$}
\\ &\iff \alpha_{1} (\overline{a}) \, \dot\sqsubseteq \, \pi_{\alpha_1(\Kk)}(\overline{a'}) \land \alpha_{2} (\overline{a}) \, \dot\sqsubseteq \, \pi_{\alpha_2(\Kk)}(\overline{a'})
      \tag{by def. of $\overline{a_1} \times \overline{a_{2}}$, $\pi_{\alpha_1(\Kk)}$, and $\pi_{\alpha_2(\Kk)}$}
\\ &\iff
       \overline{a} \, \dot\sqsubseteq \, \gamma_{1} (\pi_{\alpha_1(\Kk)}(\overline{a'})) \land \overline{a} \, \dot\sqsubseteq \,
       \gamma_{2} (\pi_{\alpha_2(\Kk)}(\overline{a'}))
      \tag{by def. of Galois conn.}
\\ &\iff
       \overline{a} \, \dot\sqsubseteq \, \gamma_{1}(\pi_{\alpha_1(\Kk)}(\overline{a'})) \sqcap \gamma_{2} (\pi_{\alpha_2(\Kk)}(\overline{a'}))
        \tag{by def. of $\sqcap$}
\\ &\iff
       \overline{a} \, \dot\sqsubseteq \, \gamma_{1} \otimes \gamma_{2} (\overline{a'})
      \tag{by def. of $\gamma_{1} \otimes \gamma_{2}$}
    \end{align*}
 \end{proof}

\bigskip

\noindent We now turn our attention to proving properties of the
derived abstraction operators introduced in
Sect.\,\ref{sec:derived-abstractions}. Observe that all derived
abstractions are Galois connections thanks to theorems of
Sect.\,\ref{sec:basic-abstractions}.

We proceed to show that $\fignorea{A}$ can be expressed using the basic
abstractions. This will allow us to disregard it in further proofs,
which are mostly done by structural induction on the structure of the
abstractions.  In this proof, it is convenient to name the
configuration formulas of the abstract domain, so let
$\{k_1',\dots,k_n'\} = \fignorea{A}(\Kk)$, indexed in the order of
components in vectors indexed by \(\fignorea{A}(\Kk)\).
Also, recall that \(\joina{\phi} = \joina{}\circ\proja{\phi}\) is
another derived operator, which we use in this theorem.

\begin{proof}[Thm.\,\ref{thm:fproj-expansion}]
	We first look into the expansion of \fignorea A and establish that
	the types of both sides are correct.   By definition
	(equation~\eqref{eq:fignorea}) the type of $\fignorea A$ is
	$\Aa^\Kk \to \Aa^{\{k_1',\dots,k_n'\}}$. The type of each
	\smash{$\joina{k_l'}$} in the right hand side of the equality is $\Aa^{\Kk}
	\to \Aa^{\{k_l'\}}$, and consequently the type of the entire
	product in the left-hand-side is $\strut\Aa^{\Kk} \to
	\Aa^{\{k_1',\dots,k_n'\}}$ as required; cf.\ the definition of
	parallel composition for configuration sets.

	\noindent The proof proceeds by mathematical induction with the
	following hypothesis:
	\begin{equation}
		\left(\joina{k_1'} \otimes \dots \otimes \joina{k_i'}\right)(\abar)=
		\textstyle\prod_{l=1}^i\bigsqcup_{k\in\Kk,
			k \models\, k_l'} \pi_k(\abar)
	\end{equation}

	\noindent \emph{Base case.}  Consider a single $\joina{k_i'}$ and let $\abar
	\in\Aa^{\Kk}$. We proceed by equational reasoning from left to
	right:
	\begin{align}
		\joina{k_l'}(\abar) =~
		& (\joina{} \circ \proja{k_l'}) (\abar)
		\tag{def.\ of \joina{\phi}}
\\ =~ &
		\joina{} \left(
				\textstyle\prod_{\{k\in\Kk \mid k\entails k_l'\}}
					\pi_k(\abar)\right)
		\tag{def.\ of \proja{\phi}}
\\ =~ &
	\left(\textstyle\bigsqcup_{\{k\in\Kk \mid k\entails k_l'\}}
		\pi_k(\abar)\right)
		\tag{def.\ of \joina{\phantom\phi}}
	\end{align}

	\noindent\emph{Inductive step} (again by equational reasoning from left to
	right):%
	\begin{align*}
		&
		\left(\joina{k_1'} \otimes \dots \otimes
			\joina{k_{i+1}'}\right)(\abar)=
	\\ =~ &
		\left(\left(\lambda \abar.\, \textstyle\prod_{l=1}^i
				\bigsqcup_{k\in\Kk, k\models k_l'} \pi_k(\abar)\right)
				\otimes \joina{k_{i+1}'}\right)(\abar)
		\tag{by IH}
	\\ =~ &
		\left(\left(\lambda \abar.\, \textstyle\prod_{l=1}^i
		\bigsqcup_{k\in\Kk,  k\models k_l'} \pi_k(\abar)\right) \otimes
		\left(\lambda\abar.\,\left(\textstyle\bigsqcup_{k\in\Kk,
					k\entails k_{i+1}'} \!\pi_k(\abar)\right)\right)\right)(\abar)
		\tag{the base case above}
	\\ =~ &
			\left(\lambda \abar.\, \textstyle\prod_{l=1}^{i+1}
				\bigsqcup_{k\in\Kk, k\models k_l'}
				\pi_k(\abar)\right)(\abar)
		\tag{def.\ of $\otimes$; $k_{i+1}'$ is a different formula from any of $k_l'$s}
	\\ =~ &
			\textstyle\prod_{l=1}^{i+1}
				\bigsqcup_{k\in\Kk, k\models k_l'} \pi_k(\abar)
		\tag{beta reduction}
	\end{align*}

	The above completes the inductive proof. The inductive hypothesis for $i=n$
	concludes the proof of correctness for expansion of \fignorea
	A\kern-5mm.

	\bigskip

	\noindent The proof for the expansion of \fignoreg A is similar.
	The type of $\fignoreg A$ is by definition $\Aa^{\{k_1',\dots,k_n'\}} \to
	\Aa^\Kk$.  The type of each of
	the factors in the right-hand-side is $\joing{k_l'} :
	\Aa^{\{k_l'\}} \to \Aa^\Kk$.  Now, by definition of the product the
	type of the entire term is: \(\Aa^{\{k_1',\dots,k_n'\}} \to
		\Aa^\Kk\) (since $k_l'$ are different formulea).

	\noindent The inductive hypothesis is ($\abar''\in\Aa^{\{k_1',...,k_i'\}}$):%
	\begin{equation*}
		\left(\joing{k_1'} \otimes \dots \otimes
			\joing{k_i'}\right)(\abar'')=
		\prod_{k\in\Kk}
		\begin{cases}
			\pi_{k_l'} (\abar'') & \text{if }
				k\models k_l' \text{ for some } l\in1..i\\
			\top & \text{otherwise}
		\end{cases}
	\end{equation*}

	\noindent \emph{Base case.}
	\begin{align*}
&\fignoreg{k_l'} = \projg{k_l'} \circ \joing{} =  \projg{k_l'} \circ
(\lambda \abar''.\, \textstyle\prod_{k\in\Kk, k\entails k_l'}
\pi_{k_l'}(\abar''))
\\ =~ &
\left(\lambda\abar'.\, \prod_{k\in\Kk}
		\begin{cases}
	 		\pi_k(\abar') & \text{if } k\entails k_l'\\
	 		\top			   & \text{otherwise }
		\end{cases}
					 \right) \circ
					 \left(\lambda \abar''.\,
						 \textstyle\prod_{k\in\Kk, k\entails k_l'}
\pi_{k_l'}(\abar'')\right)
	\tag{def.\ of projection}
\\ =~ &
\lambda\abar''.\, \prod_{k\in\Kk}
		\begin{cases}
	 		\pi_{k_l'}(\abar'') & \text{if } k\entails k_l'\\
	 		\top			   & \text{otherwise }
		\end{cases}
\tag{composition}
\end{align*}

\noindent\emph{Inductive step.}%
\begin{align*}
& \left(\joing{k_1'} \otimes \dots \otimes
			\joing{k_{i+1}'}\right)(\abar'') = \left((\joing{k_1'} \otimes \dots \otimes
			\joing{k_i'}) \otimes \joing{k_{i+1}'}\right)(\abar'')
\\ =~ &
 \left(\lambda \abar''.\, \left(
			\prod_{k\in\Kk}
				\begin{cases}
					\pi_{k_l'} (\abar'')
							& \text{if } k\models k_l' \text{ for some } l\in1..i\\
					\top
							& \text{otherwise}
				\end{cases}\right)
		\otimes\right. \\
		& \left.\hspace{1cm}\left(
			\prod_{k\in\Kk}
				\begin{cases}
					\pi_{k_{i+1}}(\abar'')	
							& \text{if } k\models k_{i+1}'\\
					\top
							& \text{otherwise }
				\end{cases}
		\right)
 \right) (\abar'')
		\tag{IH and the base case}
\\ =~ &
\left(\lambda \abar''.\,
		\prod_{k\in\Kk}
				\begin{cases}
					\pi_{k_l'} (\abar'')
							& \text{if } k\models k_l' \text{ for some }
							l\in1..i+1\\
					\top
							& \text{otherwise}
				\end{cases}\right)
		(\abar'')
	\tag{$k_l'$ formulas are not equivalent and $\otimes$ uses $\sqcup$}
\\ =~ &
		\prod_{k\in\Kk}
				\begin{cases}
					\pi_{k_l'} (\abar'')
							& \text{if } k\models k_l' \text{ for some }
							l\in1..i+1\\
					\top
							& \text{otherwise}
				\end{cases}
	\tag{beta reduction}
	\end{align*}

	\bigskip

	Now, instantiate the inductive hypothesis for $i=n$, and observe that
	for any $k\in\Kk$ there exists a $k_l'\in\fignorea{\Kk}$, such that $k \models k_l'$, so the second
	case is never exercised and we end up concluding that:
\begin{equation*}
\left(\joing{k_1'} \otimes \dots \otimes
			\joing{k_n'}\right)(\abar'') =
					\fignoreg{A} (\abar'')
\end{equation*}
\qed
\end{proof}


\section{Appendix: Proof of Soundness of Abstracted Analyses} \label{App1}

We denote with $(*)$ the equation:
\[
\alpha(\overline{a})=\prod_{k' \in \alpha(\Kk)} \confproj{k'}{\alpha(\overline{a})}
\]
where $\overline{a} \in \mathbb{A}^{\Kk}$ and $\alpha : \mathbb{A}^{\Kk} \to \mathbb{A}^{\alpha(\Kk)} \in Abs$.

\begin{proposition} \label{proposition:1}
  $\forall e \in \mathit{Exp}, (\alpha,\gamma) \in Abs, \overline{d} \in \mathbb{A}^{\alpha(\Kpsi)}:~ \alpha \circ \overline{\mathcal{A'}}\sbr{e} \circ \gamma (\overline{d})
      ~\dot\sqsubseteq~
    \overline{\mathcal{D'}}_{\alpha}\sbr{e} \, \overline{d}$
\end{proposition}
\begin{proof}
By induction on the structure of expressions.

  \begin{description}
  \item[Case $\mathit{n}$:]
    \begin{align*}
   &  (\alpha \circ \overline{\mathcal{A'}}\sbr{\mathit{n}} \circ \gamma) (\overline{d})
\\ &=
      \alpha ( \overline{\mathcal{A'}}\sbr{\mathit{n}} ( \gamma (\overline{d}) ) )
    \tag{by def. of $\circ$}
\\ &=
      \alpha ( \prod_{k \in \Kpsi} \mathit{n} )
    \tag{by def. of $\overline{\mathcal{A'}}$ in Fig.~\ref{fig:liftedanalysis}}
\\ &=
      \prod_{k' \in \alpha(\Kpsi)} \mathit{n}
    \tag{by helper Lemma 1 in App.~\ref{AppHelp}}
\\ &=
      \overline{\mathcal{D'}}_{\alpha}\sbr{\mathit{n}}
    \notag
    \end{align*}

\item[Case $\texttt{x}$:]
    \begin{align*}
      & (\alpha \circ \overline{\mathcal{A'}}\sbr{\texttt{x}} \circ \gamma) (\overline{d})
\\ & =
        \alpha ( \overline{\mathcal{A'}}\sbr{\texttt{x}}
                ( \gamma (\overline{d}) ) )
        \tag{by def. of $\circ$}
\\ & =
        \alpha \big( \prod_{k \in \Kpsi}
                    \confproj{k}{ \gamma (\overline{d}) } (\texttt{x})  \big)
        \tag{by def. of $\overline{\mathcal{A'}}$ in
		 Fig.\,\ref{fig:liftedanalysis}}
\\ & =
        \prod_{k' \in \alpha(\Kpsi)}
                    \confproj{k'}{ \alpha ( \gamma (\overline{d}) )} (\texttt{x})
        \tag{by def. of ($*$)}
\\ & ~\dot\sqsubseteq~
        \prod_{k' \in \alpha(\Kpsi)}
                    \confproj{k'}{ \overline{d}} (\texttt{x})
        \tag{$\alpha \circ \gamma$ is reductive}
\\ & =
        \overline{\mathcal{D'}}_{\alpha}\sbr{\texttt{x}} \overline{d}
        \notag
    \end{align*}

\item[Case $\impbinop{e_0}{e_1}$:]
    \begin{align*}
      & (\alpha \circ \overline{\mathcal{A'}}\sbr{\impbinop{e_0}{e_1}} \circ \gamma) (\overline{d})
\\ & =
        \alpha ( \overline{\mathcal{A'}}\sbr{\impbinop{e_0}{e_1}}
                ( \gamma (\overline{d}) ) )
        \tag{by def. of $\circ$}
\\ & =
        \alpha \big( \prod_{k \in \Kpsi}
                    \confproj{k}{ \overline{\mathcal{A'}}\sbr{e_0} \gamma (\overline{d})} \,\impparbinopname\,
                    \confproj{k}{ \overline{\mathcal{A'}}\sbr{e_1} \gamma (\overline{d})} \big)
        \tag{by def. of $\overline{\mathcal{A'}}$ in
		 Fig.\,\ref{fig:liftedanalysis}}
\\ & =
        \alpha \big( \prod_{k \in \Kpsi}
                    \confproj{k}{ \overline{\mathcal{A'}}\sbr{e_0} \gamma (\overline{d}) \,\dot\impparbinopname\,
                     \overline{\mathcal{A'}}\sbr{e_1} \gamma (\overline{d})} \big)
        \tag{by def. of $\pi_{k}$ and $\dot\impparbinopname$}
\\ & =
        \prod_{k' \in \alpha(\Kpsi)}
                    \confproj{k'}{ \alpha \big( \overline{\mathcal{A'}}\sbr{e_0} \gamma (\overline{d}) \,\dot\impparbinopname\,
                     \overline{\mathcal{A'}}\sbr{e_1} \gamma (\overline{d}) \big) }
        \tag{by def. of ($*$)}
\\ & \dot\sqsubseteq
        \prod_{k' \in \alpha(\Kpsi)}
                    \confproj{k'}{ \alpha ( \overline{\mathcal{A'}}\sbr{e_0} \gamma (\overline{d}) ) \,\dot\impparbinopname\,
                     \alpha ( \overline{\mathcal{A'}}\sbr{e_1} \gamma (\overline{d}) )}
        \tag{by helper Lemma 3 in App.~\ref{AppHelp}}
\\ & \dot\sqsubseteq
        \prod_{k' \in \alpha(\Kpsi)}
                    \confproj{k'}{ \overline{\mathcal{D'}}_{\alpha}\sbr{e_0} \overline{d} \,\dot\impparbinopname\,
                     \overline{\mathcal{D'}}_{\alpha}\sbr{e_1} \overline{d}}
        \tag{by IH, twice}
\\ & =
        \prod_{k' \in \alpha(\Kpsi)}
                    \confproj{k'}{ \overline{\mathcal{D'}}_{\alpha}\sbr{e_0} \overline{d}} \,\impparbinopname\,
                    \confproj{k'}{ \overline{\mathcal{D'}}_{\alpha}\sbr{e_1} \overline{d}}
        \tag{by def. of $\pi_{k'}$ and $\dot\impparbinopname$}
\\ & =
        \overline{\mathcal{D'}}_{\alpha}\sbr{\impbinop{e_0}{e_1}} \overline{d}
        \notag
    \end{align*}
 \end{description}

\end{proof}

\begin{proposition} \label{proposition:2}
 $\forall s \in \mathit{Stm}, (\alpha,\gamma) \in Abs, \overline{d} \in \mathbb{A}^{\alpha(\Kpsi)}:~ \alpha \circ \overline{\mathcal{A}}\sbr{s} \circ \gamma (\overline{d})
      ~\dot\sqsubseteq~
    \overline{\mathcal{D}}_{\alpha}\sbr{s} \, \overline{d}$
\end{proposition}
\begin{proof}
By induction on the structure of statements.
First, we define
$\overline{d}[\texttt{x} \dot\mapsto \overline{v}] = \prod_{k' \in \alpha(\Kpsi)}
                    \confproj{k'}{\overline{d}} [\texttt{x} \mapsto \confproj{k'}{\overline{v}} ]$,
for all $\overline{d} \in \mathbb{A}^{\alpha(\Kpsi)}$ and $\overline{v} \in Const^{\alpha(\Kpsi)}$.
Thus, $\overline{d}[\texttt{x} \dot\mapsto \overline{v}]$
is a tuple that is as $\overline{d}$ except that in each component of $\overline{d}$
the variable $\texttt{x}$ is mapped to the corresponding component of the tuple $\overline{v}$.

  \begin{description}
  \item[Case $\impskip$:]
    \begin{align*}
   &  (\alpha \circ \overline{\mathcal{A}}\sbr{\impskip} \circ \gamma) (\overline{d})
\\ &=
      \alpha ( \overline{\mathcal{A}}\sbr{\impskip} ( \gamma (\overline{d}) ) )
    \tag{by def. of $\circ$}
\\ &=
      \alpha ( \gamma (\overline{d}) )
    \tag{by def. of $\overline{\mathcal{A}}$ in Fig.~\ref{fig:liftedanalysis}}
\\ & ~\dot\sqsubseteq~
      \overline{d}
    \tag{$\alpha \circ \gamma$ is reductive}
\\ &=
      \overline{\mathcal{D}}_{\alpha}\sbr{\impskip} \overline{d}
    \notag
    \end{align*}

\item[Case $\impassign{\texttt{x}}{e}$:]
    \begin{align*}
      & (\alpha \circ \overline{\mathcal{A}}\sbr{\impassign{\texttt{x}}{e}} \circ \gamma) (\overline{d})
\\ & =
        \alpha \big( \prod_{k \in \Kpsi}
                    \confproj{k}{ \gamma (\overline{d}) } [\texttt{x} \mapsto \confproj{k}{ \overline{\mathcal{A'}}\sbr{e} \gamma (\overline{d}) }]  \big)
        \tag{by def. of $\overline{\mathcal{A}}$ in
		 Fig.\,\ref{fig:liftedanalysis}}
\\ & =
        \alpha \big( \prod_{k \in \Kpsi}
                    \confproj{k}{ \gamma (\overline{d}) [\texttt{x} \dot\mapsto \overline{\mathcal{A'}}\sbr{e} \gamma (\overline{d}) ] }   \big)
        \tag{by def. of $\dot\mapsto$}
\\ & =
        \prod_{k' \in \alpha(\Kpsi)}
                    \confproj{k'}{ \alpha \big( \gamma (\overline{d}) [\texttt{x} \dot\mapsto \overline{\mathcal{A'}}\sbr{e} \gamma (\overline{d}) ] \big)}
        \tag{by def. of ($*$)}
\\ & =
        \prod_{k' \in \alpha(\Kpsi)}
                    \confproj{k'}{ \alpha ( \gamma (\overline{d}) ) [\texttt{x} \dot\mapsto \alpha ( \overline{\mathcal{A'}}\sbr{e} \gamma (\overline{d}) ) ] }
        \tag{by helper Lemma 4 in App.~\ref{AppHelp}}
\\ & \,\dot\sqsubseteq\,
       \prod_{k' \in \alpha(\Kpsi)}
                    \confproj{k'}{ \overline{d} [\texttt{x} \dot\mapsto \overline{\mathcal{D'}}_{\alpha}\sbr{e} \overline{d} ) ]  }
        \tag{IH, and $\alpha \circ \gamma$ is reductive}
\\ & =
        \prod_{k' \in \alpha(\Kpsi)}
                    \confproj{k'}{ \overline{d} } [\texttt{x} \mapsto \confproj{k'}{\overline{\mathcal{D'}}_{\alpha}\sbr{e} \overline{d} } ]
        \tag{by def. of $\dot\mapsto$}
\\ & =
        \overline{\mathcal{D}}_{\alpha}\sbr{\impassign{\texttt{x}}{e}} \overline{d}
        \notag
    \end{align*}

\item[Case $\impif{e}{s_0}{s_1}$:]
   \begin{align}
     & (\alpha \circ \overline{\mathcal{A}}\sbr{\impif{e}{s_0}{s_1}} \circ \gamma) (\overline{d})
        \notag
\\ & =
        \alpha (
                  \overline{\mathcal{A}}\sbr{\impif{e}{s_0}{s_1}}
                  ( \gamma (\overline{d}) ) )
      \tag{by def. of $\circ$}
\\ & =
        \alpha \big(  \overline{\mathcal{A}}\sbr{s_0} \gamma (\overline{d}) \ \dot \sqcup \ \overline{\mathcal{A}}\sbr{s_1} \gamma (\overline{d}) \big)
      \tag{by def. of $\overline{\mathcal{A}}$ in
		 Fig.\,\ref{fig:liftedanalysis}}
\\ & =
         \alpha \big(  \overline{\mathcal{A}}\sbr{s_0} \gamma (\overline{d}) \big) \ \dot \sqcup \ \alpha \big( \overline{\mathcal{A}}\sbr{s_1} \gamma (\overline{d}) \big)
      \tag{by $\alpha$ is a CJM }
\\ & \,\dot\sqsubseteq\,
          \overline{\mathcal{D}}_{\alpha}\sbr{s_0} \overline{d} \ \dot \sqcup \  \overline{\mathcal{D}}_{\alpha}\sbr{s_1} \overline{d}
      \tag{by IH, twice}
\\ & =
    \overline{\mathcal{D}}_{\alpha}\sbr{\impif{e}{s_0}{s_1}}\,\overline{d}
        \notag
   \end{align}

\item[Case $\impifdef{(\theta)}{s}$: ]
   \begin{align}
     & (\alpha \circ \overline{\mathcal{A}}\sbr{\impifdef{(\theta)}{s}} \circ \gamma) (\overline{d})
        \notag
\\ & =
        \alpha (
                  \overline{\mathcal{A}}\sbr{\impifdef{(\theta)}{s}}
                  ( \gamma (\overline{d}) ) )
      \tag{by def. of $\circ$}
\\ & =
        \alpha \big(  \prod_{k \in \Kpsi}
         \left\{ \begin{array}{ll} \confproj{k}{ \overline{\mathcal{A}}\sbr{s} \gamma (\overline{d})} & \quad \textrm{if} \ k \models \theta \\[1.5ex]
		\confproj{k}{\gamma (\overline{d})} & \quad \textrm{if} \  k \not \models \theta \end{array} \right. \big)
      \tag{by def. of $\overline{\mathcal{A}}$ in
		 Fig.\,\ref{fig:liftedanalysis}}
\\ & \dot\sqsubseteq
          \prod_{k' \in \alpha(\Kpsi)}
         \left\{ \begin{array}{ll} \confproj{k'}{ \alpha (\overline{\mathcal{A}}\sbr{s} \gamma (\overline{d}) )} & \quad \textrm{if} \ k' \models \theta \\[1.5ex]
         \confproj{k'}{\alpha(\gamma (\overline{d}))} \sqcup \confproj{k'}{\alpha (\overline{\mathcal{A}}\sbr{s} \gamma (\overline{d}) )} & \quad \textrm{if sat} (k' \!\land\! \theta) \land \textrm{sat} (k' \!\land\! \neg \theta) \\[1.5ex]
		\confproj{k'}{\alpha(\gamma (\overline{d}))} & \quad \textrm{if} \  k' \models \neg \theta \end{array} \right.
      \tag{by helper Lemma~\ref{lemma:12} in App.~\ref{AppHelp}}
\\ & \dot\sqsubseteq
          \prod_{k' \in \alpha(\Kpsi)}
         \left\{ \begin{array}{ll} \confproj{k'}{ \overline{\mathcal{D}}_{\alpha}\sbr{s} \overline{d}} & \quad \textrm{if} \ k' \models \theta \\[1.5ex]
         \confproj{k'}{ \overline{d}} \sqcup \confproj{k'}{ \overline{\mathcal{D}}_{\alpha}\sbr{s} \overline{d}} & \quad \textrm{if sat} (k' \!\land\! \theta) \land \textrm{sat} (k' \!\land\! \neg \theta) \\[1.5ex]
		\confproj{k'}{\overline{d}} & \quad \textrm{if} \  k' \models \neg \theta \end{array} \right.
      \tag{by IH, and $\alpha \circ \gamma$ is reductive}
\\ & =
    \overline{\mathcal{D}}_{\alpha}\sbr{\impifdef{(\theta)}{s}}\,\overline{d}
        \notag
   \end{align}

\item[Case ${\impwhile{e}{s}}$:]
We introduce a higher-order Galois connection between
$\mathbb{A}^{\Kpsi} \rightarrow \mathbb{A}^{\Kpsi}$ and $\mathbb{A}^{\alpha(\Kpsi)} \rightarrow \mathbb{A}^{\alpha(\Kpsi)}$
defined as:
\begin{align*}
  \hoalpha(\overline{\Phi}) &= \func{\overline{d}}
                                        {\alpha (\overline{\Phi}(\gamma (\overline{d}))) }, \
  \textrm{for} \ \overline{\Phi} :  \mathbb{A}^{\Kpsi} \rightarrow \mathbb{A}^{\Kpsi}
\\
  \hogamma(\overline{\Phi'}) &= \func{\overline{a}}
                                        {\gamma(\overline{\Phi'}(\alpha(\overline{a}))) }, \
  \textrm{for} \ \overline{\Phi'} :  \mathbb{A}^{\alpha(\Kpsi)} \rightarrow \mathbb{A}^{\alpha(\Kpsi)}
\end{align*}
Let $f=\func{\overline{\Phi}}{\func{\overline{a}}{}}
          {\overline{a}} \, \dot\sqcup \, \overline{\Phi} (\overline{\mathcal{A}}\sbr{s} \overline{a})$
be the functional in $\overline{\mathcal{A}}\sbr{\impwhile{e}{s}}$. We calculate
an over-approximation of $\hoalpha \circ f \circ \hogamma$, denoted as $F$, and then apply
the fixed point transfer (FPT) theorem ~\cite{CousotCousot79}  on the result.
Given a monotone function $\overline{\Phi'}$, we have:
\begin{align*}
   &
             (\hoalpha \circ f \circ \hogamma) \overline{\Phi'}
\\ & =
             \hoalpha (f (\func{\overline{a}}
                               {\gamma(\overline{\Phi'}(\alpha(\overline{a})))}))
             \tag{by def. of $\circ$ and $\hogamma$}
\\ & =
             \hoalpha (
              \func{\overline{a}}{}
                    \overline{a} \dot\sqcup
                          {\gamma(\overline{\Phi'}(\alpha(\overline{\mathcal{A}}\sbr{s}\overline{a})))} )
             \tag{by def. of $f$ and $\beta$-reduction}
\\ & ~=~
             \func{\overline{d}}{}
                    \alpha \big( \gamma(\overline{d}) \dot\sqcup
                          {\gamma(\overline{\Phi'}(\alpha(\overline{\mathcal{A}}\sbr{s}\gamma(\overline{d}))))}
             \tag{by def. of $\hoalpha$}
\\ & =
           \func{\overline{d}}{}
                    \alpha( \gamma(\overline{d}) ) \dot\sqcup
                         \alpha ( \gamma(\overline{\Phi'}(\alpha(\overline{\mathcal{A}}\sbr{s}\gamma(\overline{d})))))
             \tag{by $\alpha$ is a CJM}
\\ & \,\dot\sqsubseteq\,
              \func{\overline{d}}{}
                    \overline{d} \,\dot\sqcup\,
                         \overline{\Phi'}(\overline{\mathcal{D}}_{\alpha}\sbr{s}\overline{d})
             \tag{by IH; and $\alpha \circ \gamma$ is reductive, twice}
\\ & =      F \overline{\Phi'}
            \notag
\end{align*}
Thus, we obtain $F = \func{\overline{\Phi'}}{\func{\overline{d}}{}}
                    \overline{d} \dot\sqcup
                         \overline{\Phi'}(\overline{\mathcal{D}}_{\alpha}\sbr{s}\overline{d})$.
Since $\overline{\Phi'}$ and $\overline{\mathcal{D}}_{\alpha}$ are monotone, $F$ is also monotone.
We now have:
    \begin{align}
      & (\alpha \circ \overline{\mathcal{A}}\sbr{\impwhile{e}{s}} \circ \gamma) (\overline{d})
        \notag
\\ & =
        \alpha ( \mathrm{lfp}\, f  (\gamma (\overline{d})))
        \tag{by def. of $\overline{\mathcal{A}}$ in Fig.~\ref{fig:liftedanalysis}}
\\ & =
         \hoalpha
         (\mathrm{lfp} \, f)
         (\overline{d})
         \tag{by def. of $\hoalpha$}
\\ & ~\dot\sqsubseteq~
         (\mathrm{lfp} F)
         \overline{d}
         \tag{by fixed point transfer (FPT) theorem}
\\ & ~=~
         \overline{\mathcal{D}}_{\alpha}\sbr{\impwhile{e}{s}} \overline{d}
         \tag{by def. of $\overline{\mathcal{D}}$ in Fig.~\ref{fig:absanalysis}}
    \end{align}
\end{description}
\end{proof}

 \section{Appendix: Helper Lemmas} \label{AppHelp}

\begin{lemma} \label{lemma:1}
 $\forall \alpha \in Abs:~ \alpha (\prod_{k \in \Kk} n ) = \prod_{k' \in \alpha(\Kk)} n$
\end{lemma}
\begin{proof}
By induction on the structure of $\alpha$.
  \begin{description}
  \item[Case $\joina{} $:]
    \begin{align*}
   &  \joina{} (\prod_{k \in \Kk} n )
\\ &=
      (\textstyle\bigsqcup_{k \in \Kk} n)
    \tag{by def. of $\joina{}$}
\\ &=
      \prod_{k' \in \joina{}(\Kk)} n
    \tag{by def. of $\joina{}(\Kk)$}
    \end{align*}

      \item[Case $\proja{\phi}$:]
    \begin{align*}
   &  \proja{\phi} (\prod_{k \in \Kk} n )
\\ &=
      \prod_{\{k \in \Kk \mid k \models \phi \}} n
    \tag{by def. of $\proja{\phi}$}
\\ &=
      \prod_{k' \in \proja{\phi}(\Kk)} n
    \tag{by def. of $\proja{\phi}(\Kk)$}
    \end{align*}

      \item[Case $\alpha_{1} \otimes \alpha_{2}$:]
    \begin{align*}
   &  \alpha_{1} \otimes \alpha_{2} (\prod_{k \in \Kk} n )
\\ &=
      \alpha_{1} (\prod_{k \in \Kk} n ) \times \alpha_{2} (\prod_{k \in \Kk} n )
    \tag{by def. of $\alpha_{1} \otimes \alpha_{2}$}
\\ &=
      ( \prod_{k' \in \alpha_1(\Kk)} n ) \times ( \prod_{k' \in \alpha_2(\Kk)} n )
    \tag{by IH, twice}
\\ &=
      \prod_{k' \in \alpha_1 \cup \alpha_2(\Kk)} n
    \tag{by def. of $\overline{x_1} \times \overline{x_2}$}
\\ &=
      \prod_{k' \in \alpha_1 \otimes \alpha_2(\Kk)} n
    \tag{by def. of $\alpha_1 \otimes \alpha_2(\Kk)$}
    \end{align*}

      \item[Case $\alpha_{2} \circ \alpha_{1}$:]
    \begin{align*}
   &  \alpha_{2} \circ \alpha_{1} (\prod_{k \in \Kk} n )
\\ &=
      \alpha_{2} ( \alpha_{1} (\prod_{k \in \Kk} n ) )
    \tag{by def. of $\circ$}
\\ &=
      \alpha_2 ( \prod_{k' \in \alpha_1(\Kk)} n )
    \tag{by IH}
\\ &=
       \prod_{k'' \in \alpha_2(\alpha_1(\Kk))} n
    \tag{by IH}
\\ &=
      \prod_{k'' \in \alpha_2 \circ \alpha_1(\Kk)} n
    \tag{by def. of $\alpha_2 \circ \alpha_1(\Kk)$}
    \end{align*}
\end{description}
\end{proof}

\begin{lemma} \label{lemma:12}
\[
\begin{array}{l}
\forall \alpha \in Abs, \psi, \theta \in FeatExp, \overline{a_1}, \overline{a_2} \in \mathbb{A}^{\Kk}:~ \\
\quad    \alpha \big( \prod_{k \in \Kk}
        \left\{ \begin{array}{ll} \confproj{k}{ \overline{a_1}} & \quad \textrm{if} \ k \models \theta \\[1.5ex]
		\confproj{k}{\overline{a_2}} & \quad \textrm{if} \  k \not \models \theta \end{array} \right. \big) \\
\qquad \qquad     \dot\sqsubseteq \prod_{k' \in \alpha(\Kk)}
    \left\{ \begin{array}{ll} \confproj{k'}{ \alpha(\overline{a_1})} & \quad \textrm{if} \ k' \models \theta \\[1.5ex]
        \confproj{k'}{ \alpha(\overline{a_1})} \sqcup \confproj{k'}{\alpha(\overline{a_2})} & \quad \textrm{if sat} (k' \!\land\! \theta) \land \textrm{sat} (k' \!\land\! \neg \theta) \\[1.5ex]
		\confproj{k'}{ \alpha(\overline{a_2})} & \quad \textrm{if} \ k' \models \neg \theta \end{array} \right.
\end{array}
\]
\end{lemma}
\begin{proof}
By induction on the structure of $\alpha$.
  \begin{description}
  \item[Case $\joina{}$:]
    \begin{align*}
   &  \joina{} \big( \prod_{k \in \Kk}
        \left\{ \begin{array}{ll} \confproj{k}{ \overline{a_1}} & \quad \textrm{if} \ k \models \theta \\[1.5ex]
		\confproj{k}{\overline{a_2}} & \quad \textrm{if} \  k \not \models \theta \end{array} \right. \big)
\\ &=
      \big( \textstyle\bigsqcup_{ k \in \Kk}
        \left\{ \begin{array}{ll} \confproj{k}{ \overline{a_1}} & \quad \textrm{if} \ k \models \theta \\[1.5ex]
		\confproj{k}{\overline{a_2}} & \quad \textrm{if} \  k \not \models \theta \end{array} \right. \big)
    \tag{by def. of $\joina{}$}
\\ &=
      \big( \left\{ \begin{array}{l} \textstyle\bigsqcup_{ k \in \Kk \}} \confproj{k}{ \overline{a_1}} \qquad \quad \textrm{if} \ \bigvee_{ k \in \Kk } k \models \theta \\[1.5ex]
        \textstyle\bigsqcup_{\{ k \in \Kk \mid k \models \theta \}} \confproj{k}{ \overline{a_1}} \sqcup  \textstyle\bigsqcup_{\{ k \in \Kk \mid k \models \neg \theta\}} \confproj{k}{ \overline{a_2}} \\
        \qquad \qquad \qquad \qquad \qquad \textrm{if sat} \ ({\scriptstyle\bigvee_{ k \in \Kk }} k \!\land\! \theta) \, \land \, \textrm{ sat} ({\scriptstyle\bigvee_{k \in \Kk }} k \!\land\! \neg \theta) \\[1.5ex]
		\textstyle\bigsqcup_{ k \in \Kk } \confproj{k}{ \overline{a_2}} \qquad \quad \textrm{if} \ \bigvee_{k \in \Kk } k \models \neg \theta \end{array} \right. \big)
    \tag{by def. of $\pi_{k}$ and $\sqcup$}
\\ & \dot\sqsubseteq
      \big( \left\{ \begin{array}{l} \textstyle\bigsqcup_{k \in \Kk } \confproj{k}{ \overline{a_1}} \qquad \quad \textrm{if} \ \bigvee_{ k \in \Kk } k \models \theta \\[1.5ex]
        \textstyle\bigsqcup_{ k \in \Kk } \confproj{k}{ \overline{a_1}} \sqcup  \textstyle\bigsqcup_{ k \in \Kk } \confproj{k}{ \overline{a_2}} \\
        \qquad \qquad \qquad \qquad \qquad \textrm{if sat} \ ({\scriptstyle\bigvee_{ k \in \Kk }} k \!\land\! \theta) \, \land \, \textrm{ sat} ({\scriptstyle\bigvee_{ k \in \Kk }} k \!\land\! \neg \theta) \\[1.5ex]
		\textstyle\bigsqcup_{ k \in \Kk } \confproj{k}{ \overline{a_2}} \qquad \quad \textrm{if} \ \bigvee_{ k \in \Kk} k \models \neg \theta \end{array} \right. \big)
    \tag{by def. of $\pi_{k}$ and $\sqcup$}
\\ &=
      \big( \left\{ \begin{array}{ll} \joina{} ( \overline{a_1} ) & \quad \textrm{if} \ \bigvee_{ k \in \Kk } k \models \theta \\[1.5ex]
        \joina{} ( \overline{a_1} ) \sqcup  \joina{} ( \overline{a_2} )
        & \quad \textrm{if sat} \ ({\scriptstyle\bigvee_{ k \in \Kk }} k \!\land\! \theta) \, \land \, \textrm{ sat} ({\scriptstyle\bigvee_{ k \in \Kk }} k \!\land\! \neg \theta) \\[1.5ex]
		\joina{} ( \overline{a_2} ) & \quad \textrm{if} \ \bigvee_{ k \in \Kk } k \models \neg \theta \end{array} \right. \big)
    \tag{by def. of $\joina{}$}
    \end{align*}
    We provide an example confirming that the above relation is not equality.
    Let $\Kk=\{A \land B, A \land \neg B \}$, $\overline{a_1}=([x \mapsto 2],[x \mapsto 4])$, and $\overline{a_2}=([x \mapsto 6],[x \mapsto 2])$.
    For $\overline{a}=\prod_{k \in \Kk}
        \left\{ \begin{array}{ll} \confproj{k}{ \overline{a_1}} & \quad \textrm{if} \ k \models B \\[1.5ex]
		\confproj{k}{\overline{a_2}} & \quad \textrm{if} \  k \not \models B \end{array} \right.$, we have $\overline{a}=([x \mapsto 2],[x \mapsto 2])$.
    Then $\joina{}(\overline{a} ) = ([x \mapsto 2] )$.
    On the other hand, $\joina{} (\overline{a_1}) \sqcup \joina{} (\overline{a_2}) = ([x \mapsto \top])$.
      \item[Case $\proja{\phi}$:]
    \begin{align*}
   &  \proja{\phi} \big( \prod_{k \in \Kk}
        \left\{ \begin{array}{ll} \confproj{k}{ \overline{a_1}} & \quad \textrm{if} \ k \models \theta \\[1.5ex]
		\confproj{k}{\overline{a_2}} & \quad \textrm{if} \  k \not \models \theta \end{array} \right. \big)
\\ &=
      \prod_{\{ k \in \Kk \mid k \models \phi \}}
      \left\{ \begin{array}{ll} \confproj{k}{ \overline{a_1}} & \quad \textrm{if} \ k \models \theta \\[1.5ex]
		\confproj{k}{\overline{a_2}} & \quad \textrm{if} \  k \not \models \theta \end{array} \right.
    \tag{by def. of $\proja{\phi}$}
\\ &=
      \prod_{\{ k \in \Kk \mid k \models \phi \}}
      \left\{ \begin{array}{ll} \confproj{k}{ \proja{\phi}(\overline{a_1})} & \quad \textrm{if} \ k \models \theta \\[1.5ex]
		\confproj{k}{\proja{\phi}(\overline{a_2})} & \quad \textrm{if} \  k \not \models \theta \end{array} \right.
    \tag{by def. of $\pi_{k}$ and $\proja{\phi}$}
\\ &=
      \prod_{k' \in \proja{\phi}(\Kk)}
      \left\{ \begin{array}{ll} \confproj{k'}{ \proja{\phi}(\overline{a_1})} & \quad \textrm{if} \ k' \models \theta \\[1.5ex]
		\confproj{k'}{\proja{\phi}(\overline{a_2})} & \quad \textrm{if} \  k' \not \models \theta \end{array} \right.
    \tag{by def. of $\proja{\phi}$}
\\ &\dot\sqsubseteq
      \prod_{k' \in \proja{\phi}(\Kk)}
      \left\{ \begin{array}{ll} \confproj{k'}{ \proja{\phi}(\overline{a_1})} & \quad \textrm{if} \ k' \models \theta \\[1.5ex]
      \confproj{k'}{\proja{\phi} ( \overline{a_1} )}  \sqcup  \confproj{k'}{\proja{\phi} ( \overline{a_2} ) }
        & \quad \textrm{if sat} (k' \!\land\! \theta) \land \textrm{sat} (k' \!\land\! \neg \theta) \\[1.5ex]
		\confproj{k'}{\proja{\phi}(\overline{a_2})} & \quad \textrm{if} \  k'  \models \neg \theta \end{array} \right.
    \tag{by def. of $\dot\sqsubseteq$ and $\sqcup$}
    \end{align*}

      \item[Case $\alpha_{1} \otimes \alpha_{2}$:]
    \begin{align*}
   &  \alpha_{1} \otimes \alpha_{2} \big( \prod_{k \in \Kk}
        \left\{ \begin{array}{ll} \confproj{k}{ \overline{a_1}} & \quad \textrm{if} \ k \models \theta \\[1.5ex]
		\confproj{k}{\overline{a_2}} & \quad \textrm{if} \  k \not \models \theta \end{array} \right. \big)
\\ &=
      \alpha_{1} \big( \prod_{k \in \Kk}
        \left\{ \begin{array}{ll} \confproj{k}{ \overline{a_1}} & \quad \textrm{if} \ k \models \theta \\[1.5ex]
		\confproj{k}{\overline{a_2}} & \quad \textrm{if} \  k \not \models \theta \end{array} \right. \big) \times
        \alpha_{2} \big( \prod_{k \models \psi}
        \left\{ \begin{array}{ll} \confproj{k}{ \overline{a_1}} & \quad \textrm{if} \ k \models \theta \\[1.5ex]
		\confproj{k}{\overline{a_2}} & \quad \textrm{if} \  k \not \models \theta \end{array} \right. \big)
    \tag{by def. of $\alpha_{1} \otimes \alpha_{2}$}
\\ & \dot\sqsubseteq
      \prod_{k' \in \alpha_1(\Kk)}
    \left\{ \begin{array}{ll} \confproj{k'}{ \alpha_1(\overline{a_1})} & \quad \textrm{if} \ k' \models \theta \\[1.5ex]
        \confproj{k'}{ \alpha_1(\overline{a_1})} \sqcup \confproj{k'}{\alpha_1(\overline{a_2})} & \quad \textrm{if sat} (k' \!\land\! \theta) \land \textrm{sat} (k' \!\land\! \neg \theta) \\[1.5ex]
		\confproj{k'}{ \alpha_1(\overline{a_2})} & \quad \textrm{if} \ k' \models \neg \theta \end{array} \right.
\\ & \quad \times
      \prod_{k' \in \alpha_2(\Kk)}
    \left\{ \begin{array}{ll} \confproj{k'}{ \alpha_2(\overline{a_1})} & \quad \textrm{if} \ k' \models \theta \\[1.5ex]
        \confproj{k'}{ \alpha_2(\overline{a_1})} \sqcup \confproj{k'}{\alpha_2(\overline{a_2})} & \quad \textrm{if sat} (k' \!\land\! \theta) \land \textrm{sat} (k' \!\land\! \neg \theta) \\[1.5ex]
		\confproj{k'}{ \alpha_2(\overline{a_2})} & \quad \textrm{if} \ k' \models \neg \theta \end{array} \right.
    \tag{by IH, twice}
\\ &=
            \prod_{k' \in \alpha_1 \cup \alpha_2(\Kk)}
    \left\{ \begin{array}{l} \confproj{k'}{ \alpha_1(\overline{a_1}) \times \alpha_2(\overline{a_1})} \quad \quad \textrm{if} \ k' \models \theta \\[1.5ex]
        \confproj{k'}{ ( \alpha_1(\overline{a_1}) \times \alpha_2(\overline{a_1}) )} \sqcup \confproj{k'}{(\alpha_1(\overline{a_2}) \times \alpha_2(\overline{a_2}))} \\
         \qquad \qquad \qquad \qquad \qquad \quad \textrm{if sat} (k' \!\land\! \theta) \land \textrm{sat} (k' \!\land\! \neg \theta) \\[1.5ex]
		\confproj{k'}{ \alpha_1(\overline{a_2}) \times \alpha_2 (\overline{a_2})} \quad \quad \textrm{if} \ k' \models \neg \theta \end{array} \right.
    \tag{by def. of $\overline{x_1} \times \overline{x_2}$}
\\ &=
            \prod_{k' \in \alpha_1 \otimes \alpha_2(\Kk)}
    \left\{ \begin{array}{ll} \confproj{k'}{ \alpha_1 \otimes \alpha_2(\overline{a_1})} & \textrm{ if } k' \models \theta \\[1.5ex]
        \confproj{k'}{ \alpha_1 \otimes \alpha_2(\overline{a_1})} \sqcup \confproj{k'}{\alpha_1 \otimes \alpha_2(\overline{a_2})} & \textrm{ if sat} (k' \!\land\! \theta) \land \textrm{sat} (k' \!\land\! \neg \theta) \\[1.5ex]
		\confproj{k'}{ \alpha_1 \otimes \alpha_2 (\overline{a_2})} & \textrm{ if } k' \models \neg \theta \end{array} \right.
    \tag{by def. of $\overline{x_1} \times \overline{x_2}$, and $\alpha_{1} \otimes \alpha_{2}$}
    \end{align*}

      \item[Case $\alpha_{2} \circ \alpha_{1}$:]
    \begin{align*}
   &  \alpha_{2} \circ \alpha_{1} \big( \prod_{k \in \Kk}
        \left\{ \begin{array}{ll} \confproj{k}{ \overline{a_1}} & \quad \textrm{if} \ k \models \theta \\[1.5ex]
		\confproj{k}{\overline{a_2}} & \quad \textrm{if} \  k \not \models \theta \end{array} \right. \big)
\\ &\dot\sqsubseteq
      \alpha_{2} \big(
      \prod_{k' \in \alpha_1(\Kk)}
    \left\{ \begin{array}{ll} \confproj{k'}{ \alpha_1(\overline{a_1})} & \quad \textrm{if} \ k' \models \theta \\[1.5ex]
        \confproj{k'}{ \alpha_1(\overline{a_1})} \sqcup \confproj{k'}{\alpha_1(\overline{a_2})} & \quad \textrm{if sat} (k' \!\land\! \theta) \land \textrm{sat} (k' \!\land\! \neg \theta) \\[1.5ex]
		\confproj{k'}{ \alpha_1(\overline{a_2})} & \quad \textrm{if} \ k' \models \neg \theta \end{array} \right. \big)
    \tag{by IH}
\\ &=
      \alpha_{2} \big(
      ( \prod_{k' \in \alpha_1(\Kk)}
    \left\{ \begin{array}{ll} \confproj{k'}{ \alpha_1(\overline{a_1})} & \textrm{if} \ k' \models \theta \\[1.5ex]
        \confproj{k'}{\alpha_1(\overline{a_2})} & \textrm{if} \ k' \not \models \theta \end{array} \right. ) \sqcup
        ( \prod_{k' \in \alpha_1(\Kk)}
    \left\{ \begin{array}{ll} \confproj{k'}{ \alpha_1(\overline{a_1})} & \textrm{if} \ k' \not \models \neg \theta \\[1.5ex]
        \confproj{k'}{\alpha_1(\overline{a_2})} & \textrm{if} \ k' \models \neg \theta \end{array} \right. )  \big)
    \tag{by def. of $\sqcup$ and $\not \models$}
\\ &=
      \alpha_{2} \big(
       \prod_{k' \in \alpha_1(\Kk)}
    \left\{ \begin{array}{ll} \confproj{k'}{ \alpha_1(\overline{a_1})} & \textrm{if } k' \models \theta \\[1.5ex]
        \confproj{k'}{\alpha_1(\overline{a_2})} & \textrm{if } k' \not \models \theta \end{array} \right. \big) \sqcup
        \alpha_{2} \big( \prod_{k' \in \alpha_1(\Kk)}
    \left\{ \begin{array}{ll} \confproj{k'}{ \alpha_1(\overline{a_1})} & \textrm{if } k' \not \models \neg \theta \\[1.5ex]
        \confproj{k'}{\alpha_1(\overline{a_2})} & \textrm{if } k' \models \neg \theta \end{array} \right.  \big)
    \tag{by $\alpha_2$ is CJM}
\\ &\dot\sqsubseteq
      \prod_{k'' \in \alpha_2(\alpha_1(\Kk))}
    \left\{ \begin{array}{ll} \confproj{k''}{ \alpha_{2} (\alpha_1(\overline{a_1}))} & \textrm{ if } k'' \models \theta \\[1.5ex]
        \confproj{k''}{ \alpha_{2}(\alpha_1(\overline{a_1}))} \sqcup \confproj{k'}{\alpha_{2}(\alpha_1(\overline{a_2}))} & \textrm{ if sat} (k'' \!\land\! \theta) \land \textrm{sat} (k'' \!\land\! \neg \theta) \\[1.5ex]
		\confproj{k''}{ \alpha_{2}(\alpha_1(\overline{a_2}))} & \textrm{ if } k'' \models \neg \theta \end{array} \right.  \\[1.5ex]
  &\bigsqcup \prod_{k'' \in \alpha_2(\alpha_1(\Kk))}
    \left\{ \begin{array}{ll} \confproj{k''}{ \alpha_{2} (\alpha_1(\overline{a_1}))} & \textrm{ if } k'' \models \theta \\[1.5ex]
        \confproj{k''}{ \alpha_{2}(\alpha_1(\overline{a_1}))} \sqcup \confproj{k'}{\alpha_{2}(\alpha_1(\overline{a_2}))} & \textrm{ if sat} (k'' \!\land\! \theta) \land \textrm{sat} (k'' \!\land\! \neg \theta) \\[1.5ex]
		\confproj{k''}{ \alpha_{2}(\alpha_1(\overline{a_2}))} & \textrm{ if} \ k'' \models \neg \theta \end{array} \right.
    \tag{by IH; twice}
\\ &=
            \prod_{k'' \in \alpha_2 \circ \alpha_1(\Kk)}
    \left\{ \begin{array}{ll} \confproj{k''}{ \alpha_2 \circ \alpha_1(\overline{a_1})} & \quad \textrm{if} \ k'' \models \theta \\[1.5ex]
        \confproj{k''}{ \alpha_2 \circ \alpha_1(\overline{a_1}) \sqcup \alpha_2 \circ \alpha_1(\overline{a_2})} & \quad \textrm{if sat} (k' \!\land\! \theta) \land \textrm{sat} (k' \!\land\! \neg \theta) \\[1.5ex]
		\confproj{k'}{ \alpha_2 \circ \alpha_1 (\overline{a_2})} & \quad \textrm{if} \ k'' \models \neg \theta \end{array} \right.
    \tag{by def. of $\alpha_{2} \circ \alpha_{1}$}
    \end{align*}
\end{description}

\end{proof}

\begin{lemma} \label{lemma:2}
 $\forall \alpha \in Abs, \overline{v_1},\overline{v_2} \in Const^{\Kk}:~
    \alpha (\overline{v_1} ~\dot\impparbinopname~ \overline{v_2}) ~\dot\sqsubseteq~ \alpha (\overline{v_1}) ~\dot\impparbinopname~ \alpha (\overline{v_2})$
\end{lemma}
\begin{proof}
By induction on the structure of $\alpha$.
  \begin{description}
  \item[Case $\joina{}$:]
    \begin{align*}
   &  \joina{} (\overline{v_1} ~\dot\impparbinopname~ \overline{v_2} )
\\ &=
      \textstyle\bigsqcup_{ k \in \Kk} \confproj{k}{\overline{v_1} ~\dot\impparbinopname~ \overline{v_2}}
    \tag{by def. of $\joina{}$}
\\ &=
      \textstyle\bigsqcup_{ k \in \Kk } \big( \confproj{k}{\overline{v_1}} ~\impparbinopname~ \confproj{k}{\overline{v_2}} \big)
    \tag{by def. of $\pi_{k}$ and $\dot\impparbinopname$}
\\ &~\dot\sqsubseteq~
      \big( \textstyle\bigsqcup_{ k \in \Kk } \confproj{k}{\overline{v_1}} \big) ~\impparbinopname~ \big( \textstyle\bigsqcup_{ k \in \Kk} \confproj{k}{\overline{v_2}} \big)
    \tag{by def. of $\textstyle\bigsqcup$ and $\impparbinopname$}
\\ &=
      \joina{} ( \overline{v_1} ) ~\impparbinopname~ \joina{} ( \overline{v_2} )
    \tag{by def. of $\joina{}$}
    \end{align*}
    We provide an example confirming that the above relation is not equality.
    Let $\overline{v_1}=(5,2)$, $\overline{v_2}=(2,5)$, and $\oplus=+$.
    Then $\joina{} ((5,2) \dot\impparbinopname (2,5) ) = \joina{} ((7,7) ) = 7$.
    On the other hand, $\joina{} ((5,2) ) = \top$, $\joina{} ((2,5) ) = \top$,
    and $\top \widehat{+} \top = \top$.
      \item[Case $\proja{\phi}$:]
    \begin{align*}
   &  \proja{\phi} ( \overline{v_1} ~\dot\impparbinopname~ \overline{v_2} )
\\ &=
      \prod_{\{ k \in \Kk \mid k \models \phi\}} \confproj{k}{\overline{v_1} ~\dot\impparbinopname~ \overline{v_2}}
    \tag{by def. of $\proja{\phi}$}
\\ &=
      \prod_{\{ k \in \Kk \mid k \models \phi\}} \big( \confproj{k}{\overline{v_1}} ~\impparbinopname~ \confproj{k}{\overline{v_2}} \big)
    \tag{by def. of $\pi_{k}$ and $\dot\impparbinopname$}
\\ &=
      \big( \prod_{\{ k \in \Kk \mid k \models \phi\}} \confproj{k}{\overline{v_1}} \big) ~\dot\impparbinopname~ \big( \prod_{\{ k \in \Kk \mid k \models \phi\}} \confproj{k}{\overline{v_2}} \big)
    \tag{by def. of $\prod$ and $\dot\impparbinopname$}
\\ &=
      \proja{\phi} ( \overline{v_1} ) ~\dot\impparbinopname~ \proja{\phi} ( \overline{v_2} )
    \tag{by def. of $\proja{\phi}$}
    \end{align*}

      \item[Case $\alpha_{1} \otimes \alpha_{2}$:]
    \begin{align*}
   &  \alpha_{1} \otimes \alpha_{2} (\overline{v_1} ~\dot\impparbinopname~ \overline{v_2} )
\\ &=
      \alpha_{1} (\overline{v_1} \dot\impparbinopname \overline{v_2} ) \times \alpha_{2} (\overline{v_1} ~\dot\impparbinopname~ \overline{v_2} )
    \tag{by def. of $\alpha_{1} \otimes \alpha_{2}$}
\\ &~\dot\sqsubseteq~
      \big( \alpha_{1} (\overline{v_1}) ~\dot\impparbinopname~ \alpha_{1}(\overline{v_2} ) \big) \times \big( \alpha_{2} (\overline{v_1}) ~\dot\impparbinopname~ \alpha_{2} (\overline{v_2}) \big)
    \tag{by IH, twice}
\\ &=
      \big( \alpha_{1} (\overline{v_1}) \times \alpha_{2}(\overline{v_1} ) \big) ~\dot\impparbinopname~ \big( \alpha_{1} (\overline{v_2}) \times \alpha_{2} (\overline{v_2}) \big)
    \tag{by def. of $\times$ and $\dot\impparbinopname$}
\\ &=
      \alpha_{1} \otimes \alpha_{2} (\overline{v_1}) ~\dot\impparbinopname~  \alpha_{1} \otimes \alpha_{2} (\overline{v_2})
    \tag{by def. of $\alpha_{1} \otimes \alpha_{2}$}
    \end{align*}

      \item[Case $\alpha_{2} \circ \alpha_{1}$:]
    \begin{align*}
   &  \alpha_{2} \circ \alpha_{1} (\overline{v_1} ~\dot\impparbinopname~ \overline{v_2} )
\\ &=
      \alpha_{2} ( \alpha_{1} (\overline{v_1} ~\dot\impparbinopname~ \overline{v_2} ) )
    \tag{by def. of $\circ$}
\\ &~\dot\sqsubseteq~
     \alpha_{2} \big( \alpha_{1} (\overline{v_1}) ~\dot\impparbinopname~ \alpha_{1}(\overline{v_2} ) \big)
    \tag{by IH}
\\ &~\dot\sqsubseteq~
      \alpha_{2} \big( \alpha_{1} (\overline{v_1}) \big) ~\dot\impparbinopname~ \alpha_{2} \big( \alpha_{1} (\overline{v_2}) \big)
    \tag{by IH}
\\ &=
      \alpha_{2} \circ \alpha_{1} (\overline{v_1}) ~\dot\impparbinopname~  \alpha_{2} \circ \alpha_{1} (\overline{v_2})
    \tag{by def. of $\alpha_{2} \circ \alpha_{1}$}
    \end{align*}
\end{description}
\end{proof}

We define $\overline{a}[\texttt{x} \dot\mapsto \overline{v}]$ to mean
a tuple that is as $\overline{a}$ except that in each its component
the variable $\texttt{x}$ is mapped to the corresponding component of the tuple $\overline{v}$.

\begin{lemma} \label{lemma:3}
 $\forall \alpha \in Abs, \overline{a} \in \mathbb{A}^{\Kk}, \overline{v} \in Const^{\Kk}:~
    \alpha ( \overline{a} [\texttt{x} \dot\mapsto \overline{v}] ) = \alpha ( \overline{a}) [\texttt{x} \dot\mapsto \alpha(\overline{v})]$
\end{lemma}
\begin{proof}
By induction on the structure of $\alpha$.
  \begin{description}
  \item[Case $\joina{}$:]
    \begin{align*}
   &  \joina{} ( \overline{a} [\texttt{x} \dot\mapsto \overline{v}]  )
\\ &=
      \textstyle\bigsqcup_{ k \in \Kk } \confproj{k}{ \overline{a} [\texttt{x} \dot\mapsto \overline{v}] }
    \tag{by def. of $\joina{}$}
\\ &=
      \textstyle\bigsqcup_{ k \in \Kk } \big( \confproj{k}{\overline{a}} [\texttt{x} \mapsto \confproj{k}{\overline{v}}] \big)
    \tag{by def. of $\pi_{k}$ and $\dot\mapsto$}
\\ &=
      \big( \textstyle\bigsqcup_{ k \in \Kk } \confproj{k}{\overline{a}} \big) [\texttt{x} \mapsto  \textstyle\bigsqcup_{ k \in \Kk} \confproj{k}{\overline{v}} ]
    \tag{by def. of $\textstyle\bigsqcup$ and $\mapsto$}
\\ &=
      \joina{} ( \overline{a} ) [ \texttt{x} \mapsto \joina{} ( \overline{v} )]
    \tag{by def. of $\joina{}$}
    \end{align*}

      \item[Case $\proja{\phi}$:]
    \begin{align*}
   &  \proja{\phi} ( \overline{a} [\texttt{x} \dot\mapsto \overline{v}] )
\\ &=
      \prod_{\{ k \in \Kk \mid k \models \phi\}} \confproj{k}{\overline{a} [\texttt{x} \dot\mapsto \overline{v}]}
    \tag{by def. of $\proja{\phi}$}
\\ &=
      \prod_{\{ k \in \Kk \mid k \models \phi\}} \big( \confproj{k}{\overline{a}} [ \texttt{x} \mapsto \confproj{k}{\overline{v}} ] \big)
    \tag{by def. of $\pi_{k}$ and $\dot\mapsto$}
\\ &=
      \big( \prod_{\{ k \in \Kk \mid k \models \phi\}} \confproj{k}{\overline{a}} \big) [\texttt{x} \dot\mapsto \prod_{\{ k \in \Kk \mid k \models \phi\}} \confproj{k}{\overline{v}} ]
    \tag{by def. of $\prod$ and $\dot\mapsto$}
\\ &=
      \proja{\phi} ( \overline{a} ) [\texttt{x} \dot\mapsto \proja{\phi} ( \overline{v} ) ]
    \tag{by def. of $\proja{\phi}$}
    \end{align*}

      \item[Case $\alpha_{1} \otimes \alpha_{2}$:]
    \begin{align*}
   &  \alpha_{1} \otimes \alpha_{2} (\overline{a} [\texttt{x} \dot\mapsto \overline{v}] )
\\ &=
      \alpha_{1} (\overline{a} [\texttt{x} \dot\mapsto \overline{v}] ) \times \alpha_{2} (\overline{a} [\texttt{x} \dot\mapsto \overline{v}] )
    \tag{by def. of $\alpha_{1} \otimes \alpha_{2}$}
\\ &=
       \alpha_{1} (\overline{a}) [\texttt{x} \dot\mapsto \alpha_{1}(\overline{v}) ] \times \alpha_{2} (\overline{a}) [\texttt{x} \dot\mapsto \alpha_{2} (\overline{v}) ]
    \tag{by IH, twice}
\\ &=
      \big( \alpha_{1} (\overline{a}) \times \alpha_{2}(\overline{a} ) \big) [\texttt{x} \dot\mapsto \alpha_{1} (\overline{v}) \times \alpha_{2} (\overline{v}) ]
    \tag{by def. of $\times$ and $\dot\mapsto$}
\\ &=
      \alpha_{1} \otimes \alpha_{2} (\overline{a}) [\texttt{x} \dot\mapsto  \alpha_{1} \otimes \alpha_{2} (\overline{v})]
    \tag{by def. of $\alpha_{1} \otimes \alpha_{2}$}
    \end{align*}

      \item[Case $\alpha_{2} \circ \alpha_{1}$:]
    \begin{align*}
   &  \alpha_{2} \circ \alpha_{1} (\overline{a} [\texttt{x} \dot\mapsto \overline{v}] )
\\ &=
      \alpha_{2} ( \alpha_{1} (\overline{a} [\texttt{x} \dot\mapsto \overline{v}] ) )
    \tag{by def. of $\alpha_{2} \circ \alpha_{1}$}
\\ &=
      \alpha_{2} ( \alpha_{1} (\overline{a}) [\texttt{x} \dot\mapsto \alpha_{1}(\overline{v}) ] )
    \tag{by IH}
\\ &=
      \alpha_{2} \big( \alpha_{1} (\overline{a}) \big) [\texttt{x} \dot\mapsto \alpha_{2} \big( \alpha_{1} (\overline{v}) \big) ]
    \tag{by IH}
\\ &=
      \alpha_{2} \circ \alpha_{1} (\overline{a}) [\texttt{x} \dot\mapsto  \alpha_{2} \circ \alpha_{1} (\overline{v})]
    \tag{by def. of $\alpha_{2} \circ \alpha_{1}$}
    \end{align*}
\end{description}
\end{proof}

\section{Appendix: Monotonicity of Abstracted Analyses} \label{App3}

\begin{lemma}[$\overline{\mathcal{D'}}_{\alpha}\sbr{e}$ is monotone] \label{lemma:5}
  \begin{align*}
    \forall e \in Exp, \alpha \in Abs, \overline{d}, \overline{d'} \in \mathbb{A}^{\alpha(\Kpsi)}.~
    \overline{d} ~\dot\sqsubseteq~ \overline{d'}
    \implies
    \overline{\mathcal{D'}}_{\alpha}\sbr{e} \overline{d}
    ~\dot\sqsubseteq~
    \overline{\mathcal{D'}}_{\alpha}\sbr{e} \overline{d'}
  \end{align*}
\end{lemma}
\begin{proof}
  Let $e$, $\alpha$, and $\overline{d} ~\dot\sqsubseteq~ \overline{d'}$ be given. We
  proceed by structural induction on $e$.
  \begin{description}
  \item[Case $n$:]
    \begin{align}
      \overline{\mathcal{D'}}_{\alpha}\sbr{n} \overline{d}
    =
      \prod_{k' \in \alpha(\Kpsi)} n
    =
      \overline{\mathcal{D'}}_{\alpha}\sbr{n} \overline{d'}
    \notag
    \end{align}

  \item[Case $\texttt{x}$:]
    \begin{align}
    &  \overline{\mathcal{D'}}_{\alpha}\sbr{\texttt{x}} \overline{d}
    =
      \prod_{k' \in \alpha(\Kpsi)}
                    \confproj{k'}{ \overline{d}} (\texttt{x})
    \tag{by def. of $\overline{\mathcal{D'}}_{\alpha}$} \\
    &\sqsubseteq
      \prod_{k' \in \alpha(\Kpsi)}
                    \confproj{k'}{ \overline{d'}} (\texttt{x})
    \tag{by $\overline{d} ~\dot\sqsubseteq~ \overline{d'}$} \\
    &=
      \overline{\mathcal{D'}}_{\alpha}\sbr{\texttt{x}} \overline{d'}
      \tag{by def. of $\overline{\mathcal{D'}}_{\alpha}$}
    \end{align}

  \item[Case $\impbinop{e_0}{e_1}$:]
    \begin{align}
  &   \overline{\mathcal{D'}}_{\alpha}\sbr{\impbinop{e_0}{e_1}} \overline{d}
      \notag
\\& =
      \prod_{k' \in \alpha(\Kpsi)}
                    \confproj{k'}{ \overline{\mathcal{D'}}_{\alpha}\sbr{e_0} \overline{d}} \impparbinopname
                    \confproj{k'}{ \overline{\mathcal{D'}}_{\alpha}\sbr{e_1} \overline{d}}
      \tag{by def. of $\overline{\mathcal{D'}}_{\alpha}$}
\\& \sqsubseteq
       \prod_{k' \in \alpha(\Kpsi)}
                    \confproj{k'}{ \overline{\mathcal{D'}}_{\alpha}\sbr{e_0} \overline{d'}} \impparbinopname
                    \confproj{k'}{ \overline{\mathcal{D'}}_{\alpha}\sbr{e_1} \overline{d'}}
      \tag{by IH; and $\overline{d} ~\dot\sqsubseteq~ \overline{d'}$}
\\& = \overline{\mathcal{D'}}_{\alpha}\sbr{\impbinop{e_0}{e_1}} \overline{d'}
      \tag{by def. of $\overline{\mathcal{D'}}_{\alpha}$}
    \end{align}
  \end{description}
\end{proof}

\begin{lemma}[$\overline{\mathcal{D}}_{\alpha}\sbr{s}$ is monotone] \label{lemma:6}
  \begin{align*}
    \forall s \in Stm, \alpha \in Abs, \overline{d}, \overline{d'} \in \mathbb{A}^{\alpha(\Kpsi)}.~
    \overline{d} ~\dot\sqsubseteq~ \overline{d'}
    \implies
    \overline{\mathcal{D}}_{\alpha}\sbr{s} \overline{d}
    ~\dot\sqsubseteq~
    \overline{\mathcal{D}}_{\alpha}\sbr{s} \overline{d'}
  \end{align*}
\end{lemma}
\begin{proof}
  Let $s$, $\alpha$, and $\overline{d} ~\dot\sqsubseteq~ \overline{d'}$ be given. We
  proceed by structural induction on $s$.
  \begin{description}
  \item[Case $\impskip$:]
    \begin{align}
      \overline{\mathcal{D}}_{\alpha}\sbr{\impskip} \overline{d}
     =
      \overline{d}
     ~\dot\sqsubseteq~
      \overline{d'}
     =
      \overline{\mathcal{D}}_{\alpha}\sbr{\impskip} \overline{d'}
     \tag{by def. of $\overline{\mathcal{D}}_{\alpha}$}
    \end{align}

  \item[Case $\impassign{\texttt{x}}{e}$:]
    \begin{align}
   &  \overline{\mathcal{D}}_{\alpha}\sbr{\impassign{\texttt{x}}{e}} \overline{d}
      \notag
\\ & =
       \prod_{k' \in \alpha(\Kpsi)}
           (\confproj{k'}{\overline{d}}) [ \texttt{x} \mapsto \confproj{k'}{\overline{\mathcal{D'}}_{\alpha}\sbr{e} \overline{d}} ]
      \tag{by def. of $\overline{\mathcal{D}}_{\alpha}$}
\\ & ~\dot\sqsubseteq~
      \prod_{k' \in \alpha(\Kpsi)}
           (\confproj{k'}{\overline{d'}}) [ \texttt{x} \mapsto \confproj{k'}{\overline{\mathcal{D'}}_{\alpha}\sbr{e} \overline{d'}} ]
      \tag{by $\overline{d} ~\dot\sqsubseteq~ \overline{d'}$ and Lemma~\ref{lemma:5}}
\\ & =
      \overline{\mathcal{D}}_{\alpha}\sbr{\impassign{\texttt{x}}{e}} \overline{d'}
      \tag{by def. of $\overline{\mathcal{D}}_{\alpha}$}
    \end{align}

  \item[Case $\impseq{s_0}{s_1}$:]
    \begin{align}
   &  \overline{\mathcal{D}}_{\alpha}\sbr{\impseq{s_0}{s_1}} \overline{d}
      \notag
\\ & =
      \overline{\mathcal{D}}_{\alpha}\sbr{s_1} (
        \overline{\mathcal{D}}_{\alpha}\sbr{s_0} \overline{d} )
      \tag{by def. of $\overline{\mathcal{D}}_{\alpha}$}
\\ & ~\dot\sqsubseteq~
      \overline{\mathcal{D}}_{\alpha}\sbr{s_1} (
        \overline{\mathcal{D}}_{\alpha}\sbr{s_0} \overline{d'} )
      \tag{by IH, twice; and $\overline{d} ~\dot\sqsubseteq~ \overline{d'}$}
\\ & =
      \overline{\mathcal{D}}_{\alpha}\sbr{\impseq{s_0}{s_1}} \overline{d'}
      \tag{by def. of $\overline{\mathcal{D}}_{\alpha}$}
    \end{align}

  \item[Case $\impif{e}{s_0}{s_1}$:]
    \begin{align}
   &  \overline{\mathcal{D}}_{\alpha}\sbr{\impif{e}{s_0}{s_1}} \overline{d}
      \notag
\\ & =
     \overline{\mathcal{D}}_{\alpha}\sbr{s_0} \overline{d} \,\dot\sqcup\, \overline{\mathcal{D}}_{\alpha}\sbr{s_1} \overline{d}
      \tag{by def. of $\overline{\mathcal{D}}_{\alpha}$}
\\ & ~\dot\sqsubseteq~
      \overline{\mathcal{D}}_{\alpha}\sbr{s_0} \overline{d'} \,\dot\sqcup\, \overline{\mathcal{D}}_{\alpha}\sbr{s_1} \overline{d'}
      \tag{by IH, twice; and $\overline{d} ~\dot\sqsubseteq~ \overline{d'}$}
\\ & =
      \overline{\mathcal{D}}_{\alpha}\sbr{\impif{e}{s_0}{s_1}} \overline{d'}
      \tag{by def. of $\overline{\mathcal{D}}_{\alpha}$}
    \end{align}

    \item[Case $\impifdef{(\theta)}{s}$:]
    \begin{align}
   &  \overline{\mathcal{D}}_{\alpha}\sbr{\impifdef{(\theta)}{s}} \overline{d}
      \notag
\\ & =
         \prod_{k' \in \alpha(\Kpsi)}
         \left\{ \begin{array}{ll} \confproj{k'}{ \overline{\mathcal{D}}_{\alpha}\sbr{s} \overline{d}} & \quad \textrm{if} \ k' \models \theta \\[1.5ex]
         \confproj{k'}{ \overline{d}} \sqcup \confproj{k'}{\overline{\mathcal{D}}_{\alpha}\sbr{s} \overline{d}} & \quad \textrm{if sat} (k' \!\land\! \theta) \land \textrm{sat} (k' \!\land\! \neg \theta) \\[1.5ex]
		\confproj{k'}{\overline{d}} & \quad \textrm{if} \  k' \models \neg\theta \end{array} \right.
      \tag{by def. of $\overline{\mathcal{D}}_{\alpha}$}
\\ & ~\dot\sqsubseteq~
         \prod_{k' \in \alpha(\Kpsi)}
         \left\{ \begin{array}{ll} \confproj{k'}{ \overline{\mathcal{D}}_{\alpha}\sbr{s} \overline{d'}} & \quad \textrm{if} \ k' \models \theta \\[1.5ex]
         \confproj{k'}{ \overline{d'}} \sqcup \confproj{k'}{\overline{\mathcal{D}}_{\alpha}\sbr{s} \overline{d'}} & \quad \textrm{if sat} (k' \!\land\! \theta) \land \textrm{sat} (k' \!\land\! \neg \theta) \\[1.5ex]
		\confproj{k'}{\overline{d'}} & \quad \textrm{if} \  k' \models \neg\theta \end{array} \right.
      \tag{by IH, and $\overline{d} ~\dot\sqsubseteq~ \overline{d'}$}
\\ & =
      \overline{\mathcal{D}}_{\alpha}\sbr{\impifdef{(\theta)}{s}} \overline{d'}
      \tag{by def. of $\overline{\mathcal{D}}_{\alpha}$}
    \end{align}

  \item[Case $\impwhile{e}{s}$:]
    Let $f = \func{\overline{\Phi}}{\func{\overline{d}}{
                                                { \overline{d}
                                                  \,\dot\sqcup\,
                                                  \overline{\Phi}(\overline{\mathcal{D}}_{\alpha}\sbr{s} \overline{d})
                                                }}}$
    be the functional in the rule for $\impwhile{e}{s}$.
    First we prove that applying the functional $f$ to a monotone
    function $\overline{\Phi}$ yields a monotone function.
    Thus, we obtain that the functional $f$ operates over
    the complete lattice of monotone functions.
    Let $\overline{d} \dot\sqsubseteq \overline{d'}$ and a monotone function $\overline{\Phi}$ be given.
    We have:
    \begin{align}
   &  (f \overline{\Phi}) \overline{d}
      \notag
\\ & =
      \overline{d} \,\dot\sqcup\, \overline{\Phi}(\overline{\mathcal{D}}_{\alpha}\sbr{s} \overline{d})
      \tag{by def. of $f$}
\\ & ~\dot\sqsubseteq~
      \overline{d'} \,\dot\sqcup\, \overline{\Phi}(\overline{\mathcal{D}}_{\alpha}\sbr{s} \overline{d'})
      \tag{by IH, monotonicity of $\overline{\Phi}$, and $\overline{d} \dot\sqsubseteq \overline{d'}$}
\\ & =
      (f \overline{\Phi}) \overline{d'}
      \tag{by def. of $f$}
    \end{align}

    Second we prove that the functional $f$ itself is monotone, which
    guarantees that the while rule is well defined by Tarski's fixed
    point theorem. We extend the operator $\dot\sqsubseteq$ to operate
    over tuples of functions: $\overline{f} \overline{\dot\sqsubseteq} \overline{g}
    = \forall \overline{x}. \overline{f}(\overline{x}) \dot\sqsubseteq \overline{g}(\overline{x})$.
    Let monotone functions $\overline{\Phi}$ and $\overline{\Phi'}$ be given and
    $\overline{\Phi} \overline{\dot\sqsubseteq} \overline{\Phi'}$.
    \begin{align}
    & F \overline{\Phi}
      \notag
\\ & =
      \func{\overline{d}}{ \overline{d}
                           \sqcup
                           \overline{\Phi}(\overline{\mathcal{D}}_{\alpha}\sbr{s} \overline{d})}
      \tag{by def. of $f$}
\\ & ~\overline{\dot\sqsubseteq}~
      \func{\overline{d}}{ \overline{d}
                           \sqcup
                           \overline{\Phi'}(\overline{\mathcal{D}}_{\alpha}\sbr{s} \overline{d})}
      \tag{by def. of $\overline{\dot\sqsubseteq}$, and  $\overline{\Phi} \overline{\dot\sqsubseteq} \overline{\Phi'}$}
\\ & =
      F \overline{\Phi'}
      \tag{by def. of $f$}
    \end{align}

    Since the least fixed point is an element of the complete
    lattice of monotone functions, it is itself monotone.
    Given $\overline{d} \dot\sqsubseteq \overline{d'}$, we have:
    \begin{align*}
      \overline{\mathcal{D}}_{\alpha}\sbr{ \impwhile{e}{s} } \overline{d}
    =
      (\lfp f) \overline{d}
    \,\dot\sqsubseteq\,
      (\lfp f) \overline{d'}
    =
      \overline{\mathcal{D}}_{\alpha}\sbr{ \impwhile{e}{s} } \overline{d'}
    \end{align*}
    which concludes this case.

\end{description}
\end{proof}

\section{Appendix: Abstracted Data-flow Equations} \label{App4}

The complete list of data-flow equations for abstracted constant propagation:
\begin{scriptsize}
\vspace{-2mm}
\begin{align*}
  \varstmtout{\impskip[\ell]}^{\alpha} &= \varstmtin{\impskip[\ell]}^{\alpha}
\\[0.6em]
  \underline{\forall k' \in \alpha(\Kpsi)}{:}~
    \confproj{k'}{ \varstmtout{\impassign[\ell]{\texttt{x}}{e^{\ell_0}}}^{\alpha} }
                           &= \confproj{k'}{\varstmtin{\impassign[\ell]{\texttt{x}}{e^{\ell_0}}}}^{\alpha}
                              [ \texttt{x} \mapsto \confproj{k'}{
                                   \overline{\mathcal{D}'}_{\alpha}\sbr{e^{\ell_0}}
                                     \varstmtin{\impassign[\ell]{\texttt{x}}{e^{\ell_0}}}^{\alpha}
                              } ]
\\[0.6em]
  \varstmtout{\impseq[\ell]{s_0^{\ell_0}}{s_1^{\ell_1}}}^{\alpha} &= \varstmtout{s_1^{\ell_1}}^{\alpha}
\\            \varstmtin{s_1^{\ell_1}}^{\alpha} &= \varstmtout{s_0^{\ell_0}}^{\alpha}
\\            \varstmtin{s_0^{\ell_0}}^{\alpha} &= \varstmtin{\impseq[\ell]{s_0^{\ell_0}}{s_1^{\ell_1}}}^{\alpha}
\\[0.6em]
  \varstmtout{\impif[\ell]{\!e}{\!s_0^{\ell_0}}{\!s_1^{\ell_1}}}^{\alpha}
                           &= \varstmtout{s_0^{\ell_0}}^{\alpha} ~\dot\sqcup~ \varstmtout{s_1^{\ell_1}}^{\alpha}
\\           \varstmtin{s_0^{\ell_0}}^{\alpha} &= \varstmtin{\impif[\ell]{e}{s_0^{\ell_0}}{s_1^{\ell_1}}}^{\alpha}
\\           \varstmtin{s_1^{\ell_1}}^{\alpha} &= \varstmtin{\impif[\ell]{e}{s_0^{\ell_0}}{s_1^{\ell_1}}}^{\alpha}
\\[0.6em]
  \varstmtout{\impwhile[\ell]{e}{s^{\ell_0}}}^{\alpha} &= \varstmtin{s^{\ell_0}}^{\alpha}
\\                     \varstmtin{s^{\ell_0}}^{\alpha} &= \varstmtin{\impwhile[\ell]{e}{s^{\ell_0}}}^{\alpha}
                                                 ~\dot\sqcup~
                                                 \varstmtout{s^{\ell_0}}^{\alpha}
\\[0.6em]
\underline{\forall k' \in  \alpha(\Kpsi)}{:}~
   \confproj{k'}{\varstmtout{\impifdef[\ell]{(\theta)}{s^{\ell_0}}}^{\alpha}}
                           &= \left\{ \begin{array}{ll} \confproj{k'}{\varstmtout{s^{\ell_0}}^{\alpha}} & \ \textrm{if} \ k' \models \theta \\[1.5ex]
         \confproj{k'}{\varstmtin{\impifdef[\ell]{(\theta)}{s^{\ell_0}}}^{\alpha}}  \sqcup \confproj{k'}{\varstmtout{s^{\ell_0}}^{\alpha}} & \ \textrm{if sat} (k' \!\!\land\!\! \theta) \! \land \! \textrm{sat} (k' \!\!\land\!\! \neg \theta) \\[1.5ex]
		\confproj{k'}{\varstmtin{\impifdef[\ell]{(\theta)}{s^{\ell_0}}}^{\alpha}} & \ \textrm{if} \  k' \models \neg \theta \end{array} \right.
\\ \underline{\forall k' \in  \alpha(\Kpsi)}{:}~
     \confproj{k'}{\varstmtin{s^{\ell_0}}^{\alpha}}
                           &= \confproj{k'}{\varstmtin{\impifdef[\ell]{(\theta)}{s^{\ell_0}}}^{\alpha}}
                                                     \quad\text{if} \ (k' \land \theta \textrm{ is sat})
\end{align*}
\end{scriptsize}

We can derive data-flow equations for expressions as well, but for brevity we refer directly to $\overline{\mathcal{D'}}_{\alpha}\sbr{e}$ function.

\begin{theorem}[Soundness of Abstracted Data-Flow Equations]  \label{theorem:vardataflowequivalence}
For all $s \in Stm$ and $\alpha \in Abs$, such that
$\varstmtin{s^\ell}^{\alpha}$ and $\varstmtout{s^\ell}^{\alpha}$ satisfy the
data-flow equations in Fig.~\ref{fig:lifteddataflow}, it holds:
\[
    \overline{\mathcal{D}}_{\alpha}\sbr{s^\ell} (\varstmtin{s^\ell}^{\alpha}) ~~\dot\sqsubseteq~~ \varstmtout{s^\ell}^{\alpha}
\]
\end{theorem}
\begin{proof}
The proof is by structural induction on $s^\ell$.

\begin{description}
  \item[Case ${\impskip[\ell]}$:]
    \begin{align}
      & \overline{\mathcal{D}}_{\alpha}\sbr{\impskip[\ell]} (\varstmtin{\impskip[\ell]}^{\alpha})
      \notag
     \\& =
      \varstmtin{\impskip[\ell]}^{\alpha}
      \tag{by def. of $\overline{\mathcal{D}}_{\alpha}$}
    \\& = \varstmtout{\impskip[\ell]}^{\alpha}
      \tag{by def. of $\varstmtin{-}^{\alpha}$ and $\varstmtout{-}^{\alpha}$}
    \end{align}

  \item[Case ${\impassign[\ell]{\texttt{x}}{e}}$:]
    \begin{align}
   &  \overline{\mathcal{D}}_{\alpha}\sbr{\impassign[\ell]{\texttt{x}}{e^{\ell_0}}} (\varstmtin{\impassign[\ell]{\texttt{x}}{e^{\ell_0}}}^{\alpha})
      \notag
\\ & =
       \prod_{k' \in \alpha(\Kpsi)}
           (\confproj{k'}{\varstmtin{\impassign[\ell]{\texttt{x}}{e^{\ell_0}}}^{\alpha}}) [ \texttt{x} \mapsto \confproj{k'}{\overline{\mathcal{D'}}_{\alpha}\sbr{e} \varstmtin{\impassign[\ell]{\texttt{x}}{e^{\ell_0}}}^{\alpha}} ]
      \tag{by def. of $\overline{\mathcal{D}}_{\alpha}$}
\\ & ~=~
      \prod_{k' \in \alpha(\Kpsi)}
           \confproj{k'}{\varstmtout{\impassign[\ell]{x}{e^{\ell_0}}}^{\alpha}}
      \tag{by def. of $\varstmtin{-}^{\alpha}$ and $\varstmtout{-}^{\alpha}$}
\\ & =
      \varstmtout{\impassign[\ell]{x}{e^{\ell_0}}}^{\alpha}
      \notag
    \end{align}

  \item[Case ${\impseq[\ell]{s_0^{\ell_0}}{s_1^{\ell_1}}}$:]
      \begin{align}
   &  \overline{\mathcal{D}}_{\alpha}\sbr{\impseq[\ell]{s_0^{\ell_0}}{s_1^{\ell_1}}} (\varstmtin{\impseq[\ell]{s_0^{\ell_0}}{s_1^{\ell_1}}}^{\alpha})
      \notag
\\ & ~=~
       \overline{\mathcal{D}}_{\alpha}\sbr{s_1^{\ell_1}} ( \overline{\mathcal{D}}_{\alpha}\sbr{s_0^{\ell_0}} \varstmtin{\impseq[\ell]{s_0^{\ell_0}}{s_1^{\ell_1}}}^{\alpha})
      \tag{by def. of $\overline{\mathcal{D}}_{\alpha}$}
\\ & ~=~
      \overline{\mathcal{D}}_{\alpha}\sbr{s_1^{\ell_1}} ( \overline{\mathcal{D}}_{\alpha}\sbr{s_0^{\ell_0}} (\varstmtin{s_0^{\ell_0}}^{\alpha}))
      \tag{by def. of $\varstmtin{-}^{\alpha}$ and $\varstmtout{-}^{\alpha}$}
\\ & ~\dot\sqsubseteq~
      \overline{\mathcal{D}}_{\alpha}\sbr{s_1^{\ell_1}} (\varstmtout{s_0^{\ell_0}}^{\alpha})
      \tag{by IH}
\\ & ~=~
      \overline{\mathcal{D}}_{\alpha}\sbr{s_1^{\ell_1}} (\varstmtin{s_1^{\ell_1}}^{\alpha})
      \tag{by def. of $\varstmtin{-}^{\alpha}$ and $\varstmtout{-}^{\alpha}$}
\\ & ~\dot\sqsubseteq~
      \varstmtout{s_1^{\ell_1}}^{\alpha}
      \tag{by IH}
\\ & ~=~
      (\varstmtout{\impseq[\ell]{s_0^{\ell_0}}{s_1^{\ell_1}}}^{\alpha})
      \tag{by def. of $\varstmtin{-}^{\alpha}$ and $\varstmtout{-}^{\alpha}$}
    \end{align}

      \item[Case ${\impif[\ell]{e}{s_0^{\ell_0}}{s_1^{\ell_1}}}$:]
        \begin{align}
      &   \overline{\mathcal{D}}_{\alpha}\sbr{\impif[\ell]{e}{s_0^{\ell_0}}{s_1^{\ell_1}}} (\varstmtin{\impif[\ell]{e}{s_0^{\ell_0}}{s_1^{\ell_1}}}^{\alpha})
          \notag
\\
      & = \overline{\mathcal{D}}_{\alpha}\sbr{s_0^{\ell_0}} (\varstmtin{\impif[\ell]{e}{s_0^{\ell_0}}{s_1^{\ell_1}}}^{\alpha}) ~\pwsqcup~
\notag
\\ &\qquad \overline{\mathcal{D}}_{\alpha}\sbr{s_1^{\ell_1}} (\varstmtin{\impif[\ell]{e}{s_0^{\ell_0}}{s_1^{\ell_1}}}^{\alpha})
          \tag{by def. of $\overline{\mathcal{D}}_{\alpha}$}
\\
      & = \overline{\mathcal{D}}_{\alpha}\sbr{s_0^{\ell_0}} (\varstmtin{s_0^{\ell_0}}^{\alpha})
          ~\pwsqcup~
          \overline{\mathcal{D}}_{\alpha}\sbr{s_1^{\ell_1}} (\varstmtin{s_1^{\ell_1}}^{\alpha})
          \tag{by def. of $\varstmtin{-}^{\alpha}$, $\varstmtout{-}^{\alpha}$}
\\
      & ~\pwsqsubseteq~ \varstmtout{s_0^{\ell_0}}^{\alpha} ~\dot\sqcup~ \varstmtout{s_1^{\ell_1}}^{\alpha}
          \tag{by IH, twice}
\\
      & = \varstmtout{\impif[\ell]{e}{s_0^{\ell_0}}{s_1^{\ell_1}}}^{\alpha}
          \tag{by def. of $\varstmtin{-}^{\alpha}$, $\varstmtout{-}^{\alpha}$}
        \end{align}

      \item[Case ${\impifdef[\ell]{(\theta)}{s^{\ell_0}}}$:]
        \begin{align}
      &   \overline{\mathcal{D}}_{\alpha}\sbr{\impifdef[\ell]{(\theta)}{s^{\ell_0}}} (\varstmtin{\impifdef[\ell]{(\theta)}{s^{\ell_0}}}^{\alpha})
          \notag
\\
      & =          \prod_{k' \in \alpha(\Kpsi)}
         \left\{ \begin{array}{ll} \confproj{k'}{ \overline{\mathcal{D}}_{\alpha}\sbr{s^{\ell_0}} (\varstmtin{\impifdef[\ell]{(\theta)}{s^{\ell_0}}}^{\alpha})} & \ \textrm{if} \ k' \models \theta \\[1.5ex]
         \confproj{k'}{ \varstmtin{\impifdef[\ell]{(\theta)}{s^{\ell_0}}}^{\alpha}} \sqcup \confproj{k'}{ \overline{\mathcal{D}}_{\alpha}\sbr{s^{\ell_0}} (\varstmtin{\impifdef[\ell]{(\theta)}{s^{\ell_0}}}^{\alpha})} & \ \textrm{if sat} (k' \!\land\! \theta) \!\land\! \textrm{sat} (k' \!\land\! \neg \theta) \\[1.5ex]
		\confproj{k'}{\varstmtin{\impifdef[\ell]{(\theta)}{s^{\ell_0}}}^{\alpha}} & \ \textrm{if} \  k' \models \neg\theta \end{array} \right.
          \tag{by def. of $\overline{\mathcal{D}}_{\alpha}$}
\\
      & = \prod_{k' \in \alpha(\Kpsi)}
         \left\{ \begin{array}{ll} \confproj{k'}{ \overline{\mathcal{D}}_{\alpha}\sbr{s^{\ell_0}} (\varstmtin{s^{\ell_0}}^{\alpha})} & \ \textrm{if} \ k' \models \theta \\[1.5ex]
         \confproj{k'}{ \varstmtin{\impifdef[\ell]{(\theta)}{s^{\ell_0}}}^{\alpha}} \sqcup \confproj{k'}{\overline{\mathcal{D}}_{\alpha}\sbr{s^{\ell_0}} (\varstmtin{s^{\ell_0}}^{\alpha})} & \ \textrm{if sat} (k' \!\land\! \theta) \!\land\! \textrm{sat} (k' \!\land\! \neg \theta) \\[1.5ex]
		\confproj{k'}{\varstmtin{\impifdef[\ell]{(\theta)}{s^{\ell_0}}}^{\alpha}} & \ \textrm{if} \  k' \models \neg \theta \end{array} \right.
          \tag{by def. of $\varstmtin{-}^{\alpha}$, $\varstmtout{-}^{\alpha}$}
\\
      & ~\pwsqsubseteq~ \prod_{k' \in \alpha(\Kpsi)}
         \left\{ \begin{array}{ll} \confproj{k'}{ \varstmtout{s^{\ell_0}}^{\alpha}} & \ \textrm{if} \ k' \models \theta \\[1.5ex]
         \confproj{k'}{ \varstmtin{\impifdef[\ell]{(\theta)}{s^{\ell_0}}}^{\alpha}} \sqcup \confproj{k'}{ \varstmtout{s^{\ell_0}}^{\alpha}} & \ \textrm{if sat} (k' \!\land\! \theta) \!\land\! \textrm{sat} (k' \!\land\! \neg \theta) \\[1.5ex]
		\confproj{k'}{\varstmtin{\impifdef[\ell]{(\theta)}{s^{\ell_0}}}^{\alpha}} & \ \textrm{if} \  k' \models \neg \theta \end{array} \right.
          \tag{by IH}
\\
      & = \varstmtout{\impifdef[\ell]{(\theta)}{s^{\ell_0}}}^{\alpha}
          \tag{by def. of $\varstmtin{-}^{\alpha}$, $\varstmtout{-}^{\alpha}$}
        \end{align}

      \item[Case ${\impwhile[\ell]{e}{s^{\ell_0}}}$:]

Let $f = \func{\overline{\Phi}}{\func{\overline{d}}{
                                                { \overline{d}
                                                  \,\dot\sqcup\,
                                                  \overline{\Phi}(\overline{\mathcal{D}}_{\alpha}\sbr{s} \overline{d})
                                                }}}$
    be the functional in the rule for $\impwhile{e}{s}$.
We first prove by inner induction on $n$, that:
\begin{equation} \label{flow:while}
  f^n(\overline{\pwbot})
   (\varstmtin{\impwhile[\ell]{e}{s^{\ell_0}}}^{\alpha}
    ~\pwsqcup~
    \varstmtout{s^{\ell_0}}^{\alpha})
   ~\pwsqsubseteq~
  \varstmtout{\impwhile[\ell]{e}{s^{\ell_0}}}^{\alpha}
\end{equation}
for all $n \ge 0$, where
 $\overline{\pwbot}=\func{\overline{d}}{\pwbot}$.
The base case for $n=0$ is straightforward.

For the inductive case $n=k+1$, we assume that:
\[
  f^k(\overline{\pwbot})
   (\varstmtin{\impwhile[\ell]{e}{s^{\ell_0}}}^{\alpha}
    ~\pwsqcup~
    \varstmtout{s^{\ell_0}}^{\alpha})
  ~\pwsqsubseteq~
  \varstmtout{\impwhile[\ell]{e}{s^{\ell_0}}}^{\alpha}
\]
Then we have:
  \begin{align*}
    & f^{k+1}(\overline{\pwbot})
       (\varstmtin{\impwhile[\ell]{e}{s^{\ell_0}}}^{\alpha}
        ~\pwsqcup~
        \varstmtout{s^{\ell_0}}^{\alpha})
\\  & =
      f^{k+1}(\overline{\pwbot})
       (\varstmtin{s^{\ell_0}}^{\alpha})
      \tag{by def. of $\varstmtin{-}^{\alpha}$, $\varstmtout{-}^{\alpha}$}
\\  & =
      f (f^{k}(\overline{\pwbot})) (\varstmtin{s^{\ell_0}}^{\alpha})
      \tag{by def. of $f^{k+1}$}
\\  & =
      \varstmtin{s^{\ell_0}}^{\alpha}
                   ~\pwsqcup~
                   f^{k}(\overline{\pwbot}) \big( \overline{\mathcal{D}}_{\alpha}\sbr{s^{\ell_0}}  (\varstmtin{s^{\ell_0}}^{\alpha}) \big)
      \tag{by def. of $f$}
\\  & ~\pwsqsubseteq~
      \varstmtin{s^{\ell_0}}^{\alpha}
                   ~\pwsqcup~
                   f^{k}(\overline{\pwbot}) \big( \varstmtout{s^{\ell_0}}^{\alpha} \big)
      \tag{by outer IH, monotonicity of $f^k(\overline{\pwbot})$}
\\  & ~\pwsqsubseteq~
      \varstmtin{s^{\ell_0}}^{\alpha}
      ~\pwsqcup~
      f^{k}(\overline{\pwbot}) \big( \varstmtin{\impwhile[\ell]{e}{s^{\ell_0}}}^{\alpha}
          ~\pwsqcup~  \varstmtout{s^{\ell_0}}^{\alpha} \big)
      \tag{by monotonicity of $f^k(\overline{\pwbot})$}
\\  & ~\pwsqsubseteq~
      \varstmtin{s^{\ell_0}}^{\alpha}
      ~\pwsqcup~
      \varstmtout{\impwhile[\ell]{e}{s^{\ell_0}}}^{\alpha}
      \tag{by inner IH}
\\  & =
      \varstmtout{\impwhile[\ell]{e}{s^{\ell_0}}}^{\alpha}
      \tag{by def. of $\varstmtin{-}^{\alpha}$, $\varstmtout{-}^{\alpha}$}
  \end{align*}

Finally, we have:
\begin{align*}
      &  \overline{\mathcal{D}}_{\alpha}\sbr{\impwhile[\ell]{e}{s^{\ell_0}}} (\varstmtin{\impwhile[\ell]{e}{s^{\ell_0}}}^{\alpha})
         \notag
\\    & = (\lfp f)
             (\varstmtin{\impwhile[\ell]{e}{s^{\ell_0}}}^{\alpha})
          \tag{by def. of $\overline{\mathcal{D}}_{\alpha}$}
\\    & = (\func{\overline{d}}{ \pwsqcup_i f^i (\overline{\pwbot}) \overline{d} })
             (\varstmtin{\impwhile[\ell]{e}{s^{\ell_0}}}^{\alpha})
          \tag{by Kleene's fixed point theorem}
\\    & = \pwsqcup_i f^i (\overline{\pwbot})
                        ( \varstmtin{\impwhile[\ell]{e}{s^{\ell_0}}}^{\alpha} )
          \tag{$\beta$-reduction}
\\    & ~\pwsqsubseteq~
          \pwsqcup_i f^i (\overline{\pwbot})
                       (\varstmtin{\impwhile[\ell]{e}{s^{\ell_0}}}^{\alpha}
                        ~\pwsqcup~
                        \varstmtout{s^{\ell_0}}^{\alpha})
    \tag{by monotonicity of $f^i(\overline{\pwbot})$}
\\    & ~\pwsqsubseteq~
          \varstmtout{\impwhile[\ell]{e}{s^{\ell_0}}}^{\alpha}
          \tag{by Eq.~(\ref{flow:while})}
\end{align*}

\end{description}

\end{proof}

\section{Appendix: Proof that $\overline{\mathcal{D}}_{\alpha}\sbr{s}$ coincides with $\overline{\mathcal{A}}\sbr{\alpha(s)}$} \label{App5}

\begin{proof}
By induction on the structure of $\alpha \in Abs$ and $s \in Stm$.
Apart from the $\texttt{\#if}$-statement, for all other statements
the proof is immediate from definitions of
$\overline{\mathcal{D}}_{\alpha}$, $\overline{\mathcal{A}}$, and $\alpha(s)$.

Let us consider the case of $\impifdef{(\theta)}{s}$.
  \begin{description}
  \item[Case $\joina{'Z'}$:]
    \begin{align*}
   &  \overline{\mathcal{D}}_{\joina{} } \sbr{\impifdef{(\theta)}{s}} \overline{d}
        \tag{set of feat. is $\mathbb F$, set of configs. is $\Kk_{\psi}$}
\\ &=
       \left\{ \begin{array}{ll} \overline{\mathcal{D}}_{\joina{} }\sbr{s} \overline{d} & \quad \textrm{if} \ {\scriptstyle\bigvee_{k \in \Kpsi}}k \models \theta \\[1.5ex]
        \overline{d} \dot\sqcup \overline{\mathcal{D}}_{\joina{}}\sbr{s} \overline{d} & \quad \textrm{if sat} ({\scriptstyle\bigvee_{k \in \Kpsi }}k \!\land\! \theta) \land \textrm{sat} ({\scriptstyle\bigvee_{k \in \Kpsi}}k \!\land\! \neg \theta) \\[1.5ex]
		\overline{d} & \quad \textrm{if} \ {\scriptstyle\bigvee_{k \in \Kpsi}}k \models \neg \theta  \end{array} \right.
          \tag{by def. of $\overline{\mathcal{D}}_{\alpha}$}
\\ &=
       \left\{ \begin{array}{ll} \overline{\mathcal{A}}\sbr{\joina{'Z'} (s)} \overline{d} & \quad \textrm{if} \ {\scriptstyle\bigvee_{k \in \Kpsi }}k \models \theta \\[1.5ex]
        \overline{d} \dot\sqcup \overline{\mathcal{A}}\sbr{\joina{'Z'} (s)} \overline{d} & \quad \textrm{if sat} ({\scriptstyle\bigvee_{k \in \Kpsi}}k \!\land\! \theta) \land \textrm{sat} ({\scriptstyle\bigvee_{k \in \Kpsi}}k \!\land\! \neg \theta) \\[1.5ex]
		\overline{d} & \quad \textrm{if} \ {\scriptstyle\bigvee_{k \in \Kpsi }}k \models \neg \theta  \end{array} \right.
          \tag{by IH}
\\ &=
       \left\{ \begin{array}{l} \overline{\mathcal{A}}\sbr{\impifdef{(Z)}{ \joina{'Z'} (s)} } \overline{d} \qquad \textrm{if} \ {\scriptstyle\bigvee_{k \in \Kpsi}}k \models \theta \\[1.5ex]
        \overline{\mathcal{A}}\sbr{\impifdef{(Z)}{ \texttt{lub}(\joina{'Z'} (s),\impskip)}} \overline{d} \quad
          \textrm{if sat} ({\scriptstyle\bigvee_{k \in \Kpsi }}k \!\land\! \theta) \land \textrm{sat} ({\scriptstyle\bigvee_{k \in \Kpsi}}k \!\land\! \neg \theta) \\[1.5ex]
		\overline{\mathcal{A}}\sbr{\impifdef{(\neg Z)}{ \joina{'Z'} (s)} } \overline{d} \qquad  \textrm{if} \ {\scriptstyle\bigvee_{k \in \Kpsi}}k \models \neg \theta  \end{array} \right.
          \tag{by def. of $\overline{\mathcal{A}}$; renaming: set of feat. is $\{Z\}$, set of configs. is $\{Z\}$}
\\ &=
    \overline{\mathcal{A}} \sbr{\joina{'Z'} (\impifdef{(\theta)}{s})} \overline{d}
        \tag{by def. of $\overline{\mathcal{A}}$ and $\joina{'Z'} (\impifdef{(\theta)}{s})$}
    \end{align*}

      \item[Case $\proja{\phi}$:]
    \begin{align*}
   &  \overline{\mathcal{D}}_{\proja{\phi}} \sbr{\impifdef{(\theta)}{s}} \overline{d}
      \tag{set of feat. is $\mathbb F$, set of configs. is $\Kk_{\psi}$}
\\ &=
\prod_{\{k \in \Kpsi \mid k \models \phi \}}
         \left\{ \begin{array}{ll} \confproj{k}{ \overline{\mathcal{D}}_{\proja{\phi}}\sbr{s} \overline{d}} & \quad \textrm{if} \ k \models \theta \\[1.5ex]
         \confproj{k}{ \overline{d}} \sqcup \confproj{k}{ \overline{\mathcal{D}}_{\proja{\phi}}\sbr{s} \overline{d}} & \quad \textrm{if sat} (k \!\land\! \theta) \land \textrm{sat} (k \!\land\! \neg \theta) \\[1.5ex]
		\confproj{k}{\overline{d}} & \quad \textrm{if} \  k \models \neg \theta \end{array} \right.
          \tag{by def. of $\overline{\mathcal{D}}_{\alpha}$}
\\ &=
      \prod_{\{k \in \Kpsi \mid k \models \phi \}}
      \left\{ \begin{array}{ll} \confproj{k}{\overline{\mathcal{D}}_{\proja{\phi}}\sbr{s} \overline{d}} & \quad \textrm{if} \ k \models \theta \\[1.5ex]
       		\confproj{k}{\overline{d}} & \quad \textrm{if} \ k \not \models \theta  \end{array} \right.
    \tag{since $k$ is a valuation}
\\ &=
      \prod_{\{k \in \Kpsi \mid k \models \phi \}}
      \left\{ \begin{array}{ll} \confproj{k}{\overline{\mathcal{A}}\sbr{\proja{\phi}(s)} \overline{d}} & \quad \textrm{if} \ k \models \theta \\[1.5ex]
       		\confproj{k}{\overline{d}} & \quad \textrm{if} \ k \not \models \theta  \end{array} \right.
    \tag{by IH}
\\ &=
    \overline{\mathcal{A}}\sbr{\impifdef{(\theta)}{\proja{\phi}(s)}} \overline{d}
    \tag{by def. of $\overline{\mathcal{A}}$; renaming: set of feat. is $\mathbb F$, set of configs. is $\{k \in \Kpsi \mid k \models \phi \}$}
\\ &=
    \overline{\mathcal{A}}\sbr{\proja{\phi}(\impifdef{(\theta)}{s})} \overline{d}
    \tag{by def. of $\overline{\mathcal{A}}$ and $\proja{\phi}(\impifdef{(\theta)}{s})$}
    \end{align*}

      \item[Case $\alpha_{1} \otimes \alpha_{2}$:]
    \begin{align*}
   &  \overline{\mathcal{D}}_{\alpha_{1} \otimes \alpha_{2}} \sbr{\impifdef{(\theta)}{s}} (\overline{d})
      \tag{set of feat. is $\mathbb F$, set of configs. is $\Kk_{\psi}$}
\\ &=
      \prod_{k' \in \alpha_1 \otimes \alpha_2(\Kpsi)}
      \left\{ \begin{array}{ll} \confproj{k'}{\overline{\mathcal{D}}_{\alpha_{1} \otimes \alpha_{2}}\sbr{s} \overline{d}} & \quad \textrm{if} \ k' \models \theta \\[1.5ex]
      \confproj{k'}{\overline{d}} \sqcup \confproj{k'}{\overline{\mathcal{D}}_{\alpha_{1} \otimes \alpha_{2}}\sbr{s} \overline{d}} & \quad \textrm{if sat} (k' \!\land\! \theta) \land \textrm{sat} (k' \!\land\! \neg \theta) \\[1.5ex]
       		\confproj{k'}{\overline{d} } & \quad \textrm{if} \ k' \models \neg \theta  \end{array} \right.
    \tag{by def. of $\overline{\mathcal{D}}_{\alpha}$}
\\ &=
      \prod_{k' \in \alpha_1(\Kpsi)}
      \left\{ \begin{array}{ll} \confproj{k'}{\overline{\mathcal{D}}_{\alpha_{1}}\sbr{s} \pi_{\alpha_1(\Kpsi)}(\overline{d})} & \textrm{if} \ k' \models \theta \\[1.5ex]
      \confproj{k'}{\pi_{\alpha_1(\Kpsi)}(\overline{d})}  \sqcup \confproj{k'}{\overline{\mathcal{D}}_{\alpha_{1}}\sbr{s} \pi_{\alpha_1(\Kpsi)}(\overline{d}) } &  \textrm{if sat} (k' \!\land\! \theta) \land \textrm{sat} (k' \!\land\! \neg \theta) \\[1.5ex]
       		\confproj{k'}{\pi_{\alpha_1(\Kpsi)}(\overline{d}) } &  \textrm{if} \ k' \models \neg \theta  \end{array} \right.
\\ &\quad \times
      \prod_{k' \in \alpha_2(\Kpsi)}
      \left\{ \begin{array}{ll} \confproj{k'}{\overline{\mathcal{D}}_{\alpha_{2}}\sbr{s} \pi_{\alpha_2(\Kpsi)}(\overline{d})} &  \textrm{if} \ k' \models \theta \\[1.5ex]
      \confproj{k'}{\pi_{\alpha_2(\Kpsi)}(\overline{d})}  \sqcup \confproj{k'}{\overline{\mathcal{D}}_{\alpha_{2}}\sbr{s} \pi_{\alpha_2(\Kpsi)}(\overline{d}) } &  \textrm{if sat} (k' \!\land\! \theta) \land \textrm{sat} (k' \!\land\! \neg \theta) \\[1.5ex]
       		\confproj{k'}{\pi_{\alpha_2(\Kpsi)}(\overline{d}) } &  \textrm{if} \ k' \models \neg \theta  \end{array} \right.
       \tag{by def. of $\pi_{\alpha_1(\Kpsi)}$, $\pi_{\alpha_2(\Kpsi)}$ and $\alpha_{1} \otimes \alpha_{2}$}
\\ &=
      \prod_{k' \in \alpha_1(\Kpsi)}
      \left\{ \begin{array}{ll} \confproj{k'}{\overline{\mathcal{A}}\sbr{\alpha_{1}(s)} \pi_{\alpha_1(\Kpsi)}(\overline{d})} & \textrm{if} \ k' \models \theta \\[1.5ex]
      \confproj{k'}{\pi_{\alpha_1(\Kpsi)}(\overline{d})}  \sqcup \confproj{k'}{\overline{\mathcal{A}}\sbr{\alpha_{1}(s)} \pi_{\alpha_1(\Kpsi)}(\overline{d}) } &  \textrm{if sat} (k' \!\land\! \theta) \land \textrm{sat} (k' \!\land\! \neg \theta) \\[1.5ex]
       		\confproj{k'}{\pi_{\alpha_1(\Kpsi)}(\overline{d}) } &  \textrm{if} \ k' \models \neg \theta  \end{array} \right.
\\ &\quad \times
      \prod_{k' \in \alpha_2(\Kpsi)}
      \left\{ \begin{array}{ll} \confproj{k'}{\overline{\mathcal{A}}\sbr{\alpha_{2}(s)} \pi_{\alpha_2(\Kpsi)}(\overline{d})} &  \textrm{if} \ k' \models \theta \\[1.5ex]
      \confproj{k'}{\pi_{\alpha_2(\Kpsi)}(\overline{d})}  \sqcup \confproj{k'}{\overline{\mathcal{A}}\sbr{\alpha_{2}(s)} \pi_{\alpha_2(\Kpsi)}(\overline{d}) } &  \textrm{if sat} (k' \!\land\! \theta) \land \textrm{sat} (k' \!\land\! \neg \theta) \\[1.5ex]
       		\confproj{k'}{\pi_{\alpha_2(\Kpsi)}(\overline{d}) } &  \textrm{if} \ k' \models \neg \theta  \end{array} \right.
       \tag{by IH on $\alpha$}
\\ &=
      \prod_{\overline{k'} \in \alpha_1(\Kpsi)}
      \left\{ \begin{array}{ll} \confproj{\overline{k'}}{\overline{\mathcal{A}}\sbr{\overline{\alpha_{1}}(s,\theta)} \pi_{\alpha_1(\Kpsi)}(\overline{d})} & \quad \textrm{if} \ \overline{k'} \models \overline{\alpha_1}(\theta) \\[1.5ex]
      		\confproj{\overline{k'}}{\pi_{\alpha_1(\Kpsi)}(\overline{d}) } &  \quad \textrm{if} \ \overline{k'} \not \models \overline{\alpha_1}(\theta)  \end{array} \right.
\\ &\quad \times
      \prod_{\overline{k'} \in \alpha_2(\Kpsi)}
      \left\{ \begin{array}{ll} \confproj{\overline{k'}}{\overline{\mathcal{A}}\sbr{\overline{\alpha_{2}}(s,\theta)} \pi_{\alpha_2(\Kpsi)}(\overline{d})} & \quad \textrm{if} \ \overline{k'} \models \overline{\alpha_2}(\theta) \\[1.5ex]
      		\confproj{\overline{k'}}{\pi_{\alpha_2(\Kpsi)}(\overline{d}) } &  \quad \textrm{if} \ \overline{k'} \not \models \overline{\alpha_2}(\theta)  \end{array} \right.
       \tag{by def. $\overline{\alpha_1}$, $\overline{\alpha_2}$; renaming: to $\alpha_1 \otimes \alpha_2(\mathbb F)$, $\alpha_1 \otimes \alpha_2(\Kk_{\psi})$, See (*)}
\\ &=
      \overline{\mathcal{A}}\sbr{\alpha_{1}(\impifdef{(\theta)}{s})} \pi_{\alpha_1(\Kpsi)}(\overline{d})  \times
      \overline{\mathcal{A}}\sbr{\alpha_{2}(\impifdef{(\theta)}{s})} \pi_{\alpha_2(\Kpsi)}(\overline{d})
       \tag{by def. of $\overline{\mathcal{A}}$, $\alpha_1$, and $\alpha_2$}
\\ &=
      \left\{ \begin{array}{ll}  \overline{\mathcal{A}}\sbr{\impifdef{\big(\overline{\alpha_1}(\theta) \lor \overline{\alpha_2}(\theta)\big)\,}{\overline{\alpha_1}(s,\theta)}} & \quad  \textrm{if} \ \overline{\alpha_{1}} (s,\theta) = \overline{\alpha_{2}} (s,\theta) \\[1.5ex]
         \overline{\mathcal{A}}\sbr{\alpha_{1} ( \impifdef{(\theta)}{s}  ); \alpha_{2} ( \impifdef{(\theta)}{s}  )} & \quad  \textrm{otherwise } \end{array} \right.
    \tag{by def. of $\overline{\mathcal{A}}$, $\overline{\alpha_1}$, and $\overline{\alpha_2}$}
\\ &=
      \overline{\mathcal{A}}\sbr{\alpha_{1} \otimes \alpha_{2} (\impifdef{(\theta)}{s})} (\overline{d})
    \tag{by def. of $\alpha_{1} \otimes \alpha_{2}(\impifdef{(\theta)}{s})$}
    \end{align*}

(*) Note that $\overline{k'}$ is a renamed configuration of $k'$. The second case $\text{sat}(k' \land \theta) \land \text{sat}(k' \land \neg \theta)$
has collapsed into the first case when $\overline{k'} \models \overline{\alpha_1}(\theta)$ and $\overline{\alpha_1}(s,\theta)=lub(\alpha_1(s),skip)$
in the equation obtained after the renaming.
      \item[Case $\alpha_{2} \circ \alpha_{1}$:]
    \begin{align*}
   &  \overline{\mathcal{D}}_{\alpha_{2} \circ \alpha_{1}} \sbr{\impifdef{(\theta)}{s}} (\overline{d})
   \tag{set of feat. is $\mathbb F$, set of configs. is $\Kk_{\psi}$}
\\ &=
      \prod_{k'' \in \alpha_{2} \circ \alpha_{1}(\Kpsi)}
      \left\{ \begin{array}{ll} \confproj{k''}{\overline{\mathcal{D}}_{\alpha_{2} \circ \alpha_{1}}\sbr{s} \overline{d}} & \quad \textrm{if} \ k'' \models \theta \\[1.5ex]
      \confproj{k''}{\overline{d}} \sqcup \confproj{k''}{\overline{\mathcal{D}}_{\alpha_{2} \circ \alpha_{1}}\sbr{s} \overline{d}} & \quad \textrm{if sat} (k'' \!\land\! \theta) \land \textrm{sat} (k'' \!\land\! \neg \theta) \\[1.5ex]
       		\confproj{k''}{\overline{d} } & \quad \textrm{if} \ k'' \models \neg \theta  \end{array} \right.
    \tag{by def. of $\overline{\mathcal{D}}_{\alpha}$}
\\ &=
      \prod_{k'' \in \alpha_{2} \circ \alpha_{1}(\Kpsi)}
      \left\{ \begin{array}{ll} \confproj{k''}{ \overline{\mathcal{A}}\sbr{\alpha_2 \circ \alpha_{1}(s)} \overline{d}} & \quad \textrm{if} \ k'' \models \theta \\[1.5ex]
      \confproj{k''}{\overline{d}}  \sqcup \confproj{k''}{\overline{\mathcal{A}}\sbr{\alpha_2 \circ  \alpha_{1}(s)} \overline{d} } & \quad \textrm{if sat} (k'' \!\land\! \theta) \land \textrm{sat} (k'' \!\land\! \neg \theta) \\[1.5ex]
       		\confproj{k''}{\overline{d} } & \quad \textrm{if} \ k'' \models \neg \theta  \end{array} \right.
       \tag{by IH on $s$}
\\ &=
      \prod_{\overline{k''} \in \alpha_{2} \circ \alpha_{1}(\Kpsi)}
      \left\{ \begin{array}{ll} \confproj{\overline{k''}}{ \overline{\mathcal{A}}\sbr{\overline{\alpha_2} ( \overline{\alpha_{1}}(s,\theta), \overline{\alpha_1}(\theta))} \overline{d}} & \quad \textrm{if} \ \overline{k''} \models \overline{\alpha_2} ( \overline{\alpha_{1}} (\theta)) \\[1.5ex]
       		\confproj{\overline{k''}}{\overline{d} } & \quad \textrm{if} \ \overline{k''} \not \models \overline{\alpha_2} ( \overline{\alpha_{1}} (\theta))  \end{array} \right.
       \tag{by def. of $\overline{\alpha}$, renaming: to $\alpha_2 \circ \alpha_1(\mathbb F)$, $\alpha_2 \circ \alpha_1(\Kpsi)$, See (**)}
\\ &=
      \overline{\mathcal{A}}\sbr{\impifdef{\big(\alpha_{2} ( \alpha_{1} (\theta))\big)}{\overline{\alpha_2} ( \overline{\alpha_{1}}(s,\theta), \overline{\alpha_1}(\theta))}} \overline{d}
    \tag{by def. of $\overline{\mathcal{A}}$}
\\ &=
      \overline{\mathcal{A}}\sbr{\alpha_{2} \circ \alpha_{1} (\impifdef{(\theta)}{s})} (\overline{d})
    \tag{by def. of $\alpha_{2} \circ \alpha_{1}(\impifdef{(\theta)}{s})$}
    \end{align*}
    (**) Note that $\overline{k''}$ is a renamed configuration of $k''$, and $\overline{k''}$ is a valuation over $\alpha_2 \circ \alpha_1(\mathbb F)$. The second case $\text{sat}(k'' \land \theta) \land \text{sat}(k'' \land \neg \theta)$
has collapsed into the first case when $\overline{k''} \models \overline{\alpha_2}(\overline{\alpha_1}(\theta))$ and $\overline{\alpha_1}(s,\theta)$ or $\overline{\alpha_2}(\overline{\alpha_1}(s,\theta),\overline{\alpha_1}(\theta))$ is transformed
into $lub$ statement.
\end{description}
\end{proof}

\end{document}